\documentclass[12pt,letterpaper]{article}

\usepackage{geometry}
\usepackage[cp1252]{inputenc}
\usepackage[T1]{fontenc}
\usepackage{lmodern}
\usepackage{graphicx}
\usepackage{amsmath}
\usepackage{amssymb}
\usepackage{amsthm}
\usepackage{amsfonts}
\usepackage{authblk}
\usepackage{epstopdf}	
\usepackage{verbatim}
\usepackage{multicol}
\usepackage{tabularx}
\usepackage{natbib}
\usepackage{graphicx}
\usepackage{multirow}
\usepackage{array}
\usepackage{color}
\usepackage{float}
\restylefloat{table}
\usepackage{xcolor}
\usepackage{lscape}
\usepackage{ragged2e}
\usepackage{booktabs}
\usepackage{topcapt}
\usepackage{rotating}
\usepackage{setspace}
\usepackage{times}
\usepackage{txfonts}
\usepackage{bm}
\usepackage{titlesec}
\usepackage{indentfirst}
\usepackage{enumitem}
\usepackage[pdftex,breaklinks,colorlinks,citecolor=black,urlcolor=black,linkcolor=black]{hyperref}
\usepackage{tikz}
\usetikzlibrary{arrows,positioning} 
\usepackage{algorithmic,algorithm}
\usepackage{dcolumn}
\usepackage{pifont}
\usepackage{authblk}
\usepackage{xr}
\usepackage{cleveref}
\usepackage{multirow}

\hypersetup{
  colorlinks   = true, %Colours links instead of ugly boxes
  urlcolor     = blue, %Colour for external hyperlinks
  linkcolor    = red, %Colour of internal links
  citecolor   = blue %Colour of citations
}

\newtheorem{theorem}{Theorem}[section]
\newtheorem{assumption}{Assumption}[section]
\newtheorem{proposition}{Proposition}[section]

\newtheorem{corollary}{Corollary}[section]
\newtheorem{lemma}{Lemma}[section]
\newtheorem{definition}{Definition}[section]

\newtheorem{example}{Example}[section]

\definecolor{alizarin}{rgb}{0.82, 0.1, 0.26}

\newcolumntype{Y}{>{\raggedleft\arraybackslash}X}
\newcolumntype{C}{>{\centering\arraybackslash}X}
\newcolumntype{d}[1]{D{.}{.}{#1}}

\allowdisplaybreaks

\usepackage{dsfont}

\newcommand{\T}{{[0,T]}}
\newcommand{\tT}{{t\in[0,T]}}
\newcommand{\R}{{\mathds{R}}}
\newcommand{\N}{{\mathds{N}}}

\newcommand{\E}{{\mathbb{E}}}
\renewcommand{\P}{{\mathbb{P}}}
\newcommand{\Q}{{\mathbb{Q}}}

\newcommand{\bone}{\boldsymbol{{1}}}
\newcommand{\Id}{{\mathds{1}}}
\renewcommand{\d}{{\mathrm{d}}}

\renewcommand{\epsilon}{{\varepsilon}}
\newcommand{\ep}{{\varepsilon}}

\newcommand{\bfeta}{{\boldsymbol{{\eta}}}}

\newcommand{\btheta}{{\boldsymbol{{\theta}}}}

\newcommand{\bpi}{{\boldsymbol{{\pi}}}}

\newcommand{\mcD}{{\mathcal{D}}}

\newcommand{\mcA}{{\mathcal{A}}}
\newcommand{\F}{{\mathcal{F}}}
\newcommand{\mL}{{\mathcal{L}}}

\newcommand{\GLR}{\text{GLR}}
\newcommand{\VaR}{\text{VaR}}
\newcommand{\ES}{\text{ES}}
\newcommand{\UTE}{\text{UTE}}

\newcommand{\M}{{\mathcal{M}}}

\renewcommand{\Finv}{{\Breve{F}}}
\newcommand{\Finvsig}{{{\Breve{F}}_{\varsigma}}}
\newcommand{\Finvbench}{{{\Breve{F}}_{Y}}}
\newcommand{\Fbench}{{{F}_{Y}}}
\newcommand{\Ginv}{{\Breve{G}}}

\renewcommand{\L}{{\mathcal{L}}}
\newcommand{\mcM}{{\mathcal{M}}}
\newcommand{\mcH}{{\mathcal{H}}}
\newcommand{\BW}{{\mathcal{BW}}}
\newcommand{\DBW}{{\mathcal{BW}^\alpha}}

\newcommand{\mfH}{{\mathfrak{H}}}

\newcommand{\mfU}{{\mathfrak{U}}}

\newcommand{\uu}{{\, u^\dagger\, }}

\newcommand{\ueta}{{\underline{\eta} }}
\newcommand{\oeta}{{ \overline{\eta} }}
\newcommand{\ug}{{\underline{g} }}
\newcommand{\og}{{ \overline{g} }}

\newcommand{\lamcost}{{\lambda}^{\cost}}
\newcommand{\Ginvcost}{{\Ginv}^{\cost}}
\newcommand{\Gcost}{{G}^{\cost}}
\newcommand{\epcost}{{\ep}^{\infty}}

\newcommand{\lambw}{{\lambda}^{\BW}}
\newcommand{\Ginvbw}{{\Ginv}^{\BW}}

\newcommand{\xbw}{{x_0^{\infty}}}

\newcommand{\etaOneMin}{{\eta_1^{\min}}}
\newcommand{\etaTwoMin}{{\eta_2^{\min}}}
\newcommand{\etaTwoMax}{{\eta_2^{\max}}}

\newcommand{\cost}{{\mathfrak{c}}}
\newcommand{\bw}{{\mathfrak{b}}}

\DeclareMathOperator*{\argmin}{argmin}

\newcommand{\hW}{{\widehat{W}}}

\usepackage[colorinlistoftodos]{todonotes}

\newcommand{\Title}{Outperforming a Benchmark with $\alpha$-Bregman Wasserstein divergence
}

\newcommand{\Acknowledgements}{TN gratefully acknowledges the financial support from the Natural Sciences and Engineering Research Council of Canada (NSERC) and Laval University, and SP is grateful for the financial support from NSERC (RGPIN-2025-05847). The authors thank Chiaki Hara (Kyoto University) for insightful discussions and as well as the participants of the 2026 ``Probabilistic mass transport -- from Schr\"odinger to stochastic transport'' at the Erwin Schr\"odinger International Institute for Mathematics and Physics (ESI) for fruitful conversations.}

\newcommand{\Authors}{
Silvana M. Pesenti\footnote{Department of Statistical Sciences, University of Toronto,  700 University Avenue, Toronto, ON M5G 1X6, Canada e-mail: \href{mailto:silvana.pesenti@utoronto.ca}{silvana.pesenti@utoronto.ca}}
\quad \quad 
Thai Nguyen\footnote{\'Ecole  d'actuariat, Universit\'e Laval, 2425 rue de l'Agriculture, Qu\'ebec, Qu\'ebec, Canada G1V 0A6, email: \href{mailto:thai.nguyen@act.ulaval.ca}{thai.nguyen@act.ulaval.ca}}
}

\newcommand{\Abstract}{}

\newcommand{\Keywords}{.}

\newcommand{\JEL}{G22, G11, J11.}

\makeatletter
\newcommand*{\addFileDependency}[1]{% argument=file name and extension
\typeout{(#1)}% latexmk will find this if $recorder=0
\@addtofilelist{#1}
\IfFileExists{#1}{}{\typeout{No file #1.}}
}\makeatother

\date{}
\author{\Authors}
\title{\Title\footnote{\Acknowledgements}}

\titleformat{\section}{\normalfont\bfseries}{\thesection}{1em}{}
\titleformat{\subsection}{\normalfont\bfseries}{\thesubsection}{1em}{}
\titleformat{\subsubsection}{\normalfont\bfseries}{\thesubsubsection}{1em}{}

\titlespacing*{\section}{0pt}{6pt}{6pt}
\titlespacing*{\subsection}{0pt}{6pt}{6pt}
\titlespacing*{\subsubsection}{0pt}{6pt}{6pt}

\begin{document}
\makeatletter
\g@addto@macro{\normalsize}{%
    \setlength{\abovedisplayskip}{4pt}
    \setlength{\abovedisplayshortskip}{4pt}
    \setlength{\belowdisplayskip}{4pt}
    \setlength{\belowdisplayshortskip}{4pt}}
\makeatother

\maketitle
\thispagestyle{empty}

\begin{abstract}
\Abstract 
We consider the problem of active portfolio management, where an investor seeks the portfolio with maximal expected utility of the difference between the terminal wealth of their strategy and a proportion of the benchmark's, subject to a budget and a deviation constraint from the benchmark portfolio. 
As the investor aims at outperforming the benchmark, they choose a divergence that asymmetrically penalises gains and losses as well as penalises underperforming the benchmark more than outperforming it. This is achieved by the recently introduced $\alpha$-Bregman-Wasserstein divergence, subsuming the Bregman-Wasserstein and the popular Wasserstein divergence. We prove existence and uniqueness, characterise the optimal portfolio strategy, and give explicit conditions when the divergence constraints and the budget constraints are binding. We conclude with a numerical illustration of the optimal quantile function in a geometric Brownian motion market model.


\noindent {\it JEL Classification:} \JEL
\end{abstract}
\vspace{-0.25cm}

\onehalfspacing
% \newpage 

%%%%%%%%%%%%%%%%%%%%%%%%%%%%%%%%%%%%%%%%%%%%%%%%%%%%%
% Introduction
%%%%%%%%%%%%%%%%%%%%%%%%%%%%%%%%%%%%%%%%%%%%%%%%%%%%%

\section{Introduction\label{sec:intro}}

In active portfolio choice investors aim at outperforming the terminal wealth of a benchmark portfolio. Common criteria to measure outperformance include those related to the returns of the investor's strategy relative to that of the benchmark's \citep{browne2000MS}, a risk profile or utility \citep{Pesenti2023SIAM,ng2023portfolio,Jaimungal2022SIAM,Pesenti2024WP}, monetary goals \citep{Capponi2024MS}, efficiency of the benchmark \citep{kan2024MS}, and stochastic dominance \citep{Dentcheva2003SICAMO,Luo2025MS}. In this work, the investor's criterion is the expected utility applied to the positive part of the difference between the terminal wealth of their strategy and a proportion of the benchmark's. Thus, the investor is passionate about outperforming a proportion of the benchmark for each outcome of the economy. The multiplicative factor in front of the benchmarks terminal wealth within the utility criteria, could be chosen proportional to the cost of the benchmark, which might be out of reach to the investor. As the investor aims to outperform the benchmark's terminal wealth, the investor must deviate from the benchmark strategy. A key objective for the investor is that they wish to penalise absolute gains differently to losses (relative to the benchmark's gains and losses) and, furthermore, give a higher penalty when underperforming the benchmark. This is achieved by the recently introduced $\alpha$-Bregman-Wasserstein (BW) divergence together with the constraint that the investor's terminal wealth is comonotone with that of the benchmark. The $\alpha$-BW divergence is a generalisation of the asymmetric Bregman-Wasserstein (BW) divergence (which does not distinguish between under- and outperformance), and the Wasserstein distance (which is symmetric); see \cite{Pesenti2024ORL} who introduced the $\alpha$-BW divergence. While the BW divergence has been studied also in relation to portfolio optimisation, this is the first work that considers the $\alpha$-BW divergences in an optimal portfolio framework -- \cite{Pesenti2024ORL} only provide the definition and a representation of the $\alpha$-BW divergence. A related work is \cite{EspinosaTouzi2015MF}, who consider a finite number of agents, each of which optimise the expected utility of the convex combination of the terminal wealth of their strategy and that of the average performance of all agents. While they analyse the problem though an equilibrium approach, here we only have one investor that has a constraint on the deviation to the benchmarks terminal wealth.

Our work is situated in a complete market, where we solve the investor's optimisation problem by first deriving a quantile reformulation of the investor's optimisation problem and second solving the resulting optimisation problem using Lagrangian methods. We provide conditions on the optimisation problem to be well-posed, prove the uniqueness of its solution, and give an explicit characterisation of the optimal quantile functions when only one of the constraints is binding (that is either the budget or the $\alpha$-BW constraint). Other key contributions are the representation of the optimal quantile function when both constraints are binding and (under additional assumptions) the existence of a solution. The proofs of the representation and the existence are non-standard, as the $\alpha$-BW divergences requires careful distinction between the cases of under- and outperformance, which is in contrast to the BW divergence and the Wasserstein distance. Specifically, to distinguish between under- and outperformance, the $\alpha$-BW divergences includes an indicator function that depends on the investor's and the benchmark's quantile function. While the divergence is still convex in its first argument, the Lagrangian methods, which we employ to solve the investor's optimisation problem, require a separate analysis for under- and outperformance. The different treatment for under- and outperformance is also observed in the optimal quantile function, which is split into two cases, one for underperforming and one for outperforming the benchmark. Moreover, the existence of the optimal strategy requires technical assumptions, as the $\alpha$ in the $\alpha$-BW divergences destroys the monotonicity of the constraint functions with respect to the Lagrange multipliers.

A closely connected stream of literature is distributionally robust portfolio optimisation, where the uncertainty set is described by optimal transport divergences, indicatively see \citep{pflug2007QF,Mohajerin2018MP,blanchet2022MS,Jaimungal2022SIAM,Gao2023MoR}. While distributional robustness assumes partial knowledge of the distribution of the portfolio's terminal wealth or the assets' returns, here we use optimal transport divergences to quantify deviation from the distribution of the benchmark's terminal wealth. While the majority of works on distributional robustness portfolio optimisation consider the symmetric Wasserstein distance, the BW divergence -- which arises by choosing the cost function of the Monge-Kantorovich problem to be a Bregman divergence -- has recently gained interest. We refer to \cite{Guo2017WP,Tam2025WP} for distributional robust optimisation problems under the BW divergence, and \cite{Pesenti2024WP} for an application to portfolio choice. While the work of \cite{Pesenti2024WP} is close to this manuscript, there are important differences. The authors of \cite{Pesenti2024WP} consider an investor who maximises the expected utility of their terminal wealth's strategy subject to a budget constraint and that the investor's terminal wealth is close to that of a benchmark in a BW divergence. First, here we consider the criterion of the expected utility of the difference between the terminal wealth of the investor's and a proportion of the benchmark's strategy. Second, we additionally require that the investor's terminal wealth is comonotone with that of the benchmark, which allows for state-by-state comparison and penalising under- versus outperformance. Third, we consider the $\alpha$-BW divergences which subsumes the BW divergence for the choice of $\alpha = \frac12$.

Through an in-depth case study in a geometric Brownian market model, we investigate the investor's optimal terminal wealth for different choices of the $\alpha$-BW divergence. Specifically, we consider the parametric Power family of the $\alpha$-BW divergence. The Power family includes the Wasserstein distance as a special case, and allows the investor to be, e.g., penalise deviations from the left tail of the benchmark's terminal wealth more than its right tail. We observe that the freedom to deviate from the benchmark's right tail leads to larger gains than that of the benchmark as measure by the utility, gain-loss-ration, and upper tail expectation.

The manuscript is organised as follows. \Cref{sec:model} details the investor's optimisation problem including the underlying market model and the definition of the $\alpha$-BW divergence. \Cref{sec:quantile-reform} is devoted to the quantile reformulation and uniqueness, and \Cref{sec:unique-well-posed} discusses the well-posedness of the investor's optimisation problem. In \Cref{sec:boundary-cases} we derive the optimal quantile function when only one constraint (either the budget or the $\alpha$-BW constraint) is binding. The optimal quantile function when both constraints are binding is given in \Cref{sec:optimal-quantile} and we derive explicit conditions when the optimal quantile function exists in \Cref{sec:existence}. We conclude in \Cref{sec:ex} illustrating the optimal quantile functions in a geometric Brownian motion market model. The appendix collect auxiliary results and details on the numerical example.

\section{Investor's portfolio choice}
In this section, we detail the investor's preference and criterion, the market model, and the $\alpha$-Bregman-Wasserstein divergence, that they use to penalise between under- and outperformance.

\subsection{Investor's preferences}\label{sec:model}
Throughout, we work on a complete filtered probability space $(\Omega, \F,\P, \mathds{F}:= \{\F_t\}_{\tT})$, where $\mathds{F}$ denotes the natural filtration generated by the underlying processes. 
Over a finite time horizon $\T$, $T >0$, we consider an investor with utility preferences who aims at outperforming a benchmark strategy. Available to the investor are $d\in\N$ risky assets, $S^{i} := (S^{i}_{t})_{\tT}$, where $i \in\mcD:= \{1, \ldots, d\}$, and a risk-free bank account. 
The investor employs an admissible trading strategy that is an $\mathds{F}$-predictable process, self-financing, and belongs to $L^2(\Omega, [0,T])$. We denote admissible trading strategies by $ \btheta:= (\btheta_t)_\tT$, with $\btheta_{t}:= (\theta_{t}^1, \ldots, \theta_{t}^d)$, and $\theta_{t}^i$ represents the proportion of wealth invested in asset $i\in\mcD$ at time $\tT$, and use $\mcA$ for the set of admissible strategies. For an admissible strategy $\btheta \in \mcA$, the wealth process of the investor $X^{\btheta}:= (X^{\btheta}_t)_\tT$ evolves via the following stochastic differential equation (SDE)
\begin{equation}\label{eq:MM-SDE}
    \d X_t^{\btheta} := X_t^{\btheta}  ( 1- \bone^\intercal \btheta_t) \,r_t \, \d t + X_t^{\btheta} \sum_{i \in \mcD} \theta_t^i \;\frac{\d S_t^i}{S_t^i}\,,
\end{equation}
where $\bone$ is the vector of ones, i.e., $\bone := (1, \ldots, 1)$ and $(r_t)_{\tT}$ denotes the stochastic interest rate.
We further assume that the market is complete and denote by $ (\varsigma_t)_{\tT}$ the stochastic discount factor (SDF). The next assumption is classical and required throughout the exposition.

\begin{assumption}\label{asm:sdf-cont}
    The SDF at terminal time $\varsigma_T$ is continuously distributed.
\end{assumption}

The investor has a benchmark strategy $\bpi:= (\bpi_t)_{\tT}\in\mcA$, with $\bpi_t:= (\pi_t^1, \ldots, \pi_t^d)$, whose terminal wealth $X_T^{\bpi}$ they wishe to outperform. The investor assesses a strategy $\btheta \in \mcA$ via the expected utility applied to the difference between the terminal wealth of the chosen strategy $X_T^{\btheta}$ and that of the benchmark $X_T^{\bpi}$. The investor's utility satisfies standard properties, given next.

\begin{definition}\label{def:utility}
    The investor's utility function $U \colon \R \to \R$ satisfies $U(x) = - \infty$ for all $x < 0$, $U(x) \ge 0$ for all $x \ge 0$, and $U(0) = 0$. Moreover $U$ is when restricted to $(0, +\infty)$ strictly increasing, strictly concave,  twice differentiable, and satisfies the Inada condition:
    \begin{equation*}
        U'(0):= \lim_{x\downarrow 0} U'(x)=+\infty \quad \mbox{and}\quad  \lim_{x\to +\infty} U'(x)=0.
    \end{equation*}
Moreover, we use the convention that $U'(x) = +\infty$, whenever $x <0$, so that $U'$ is continuous on $\R$.
\end{definition}

The investor aims at outperforming the benchmark on a scenario by scenario case and moreover penalises asymmetrically if the benchmark is out- or underperformed. To establish scenario by scenario comparison, the investor restricts to strategies $\btheta$, such that their terminal wealth $X_T^{\btheta}$ is comonotone to that of the benchmark's $X_T^{\bpi}$. Recall that two random variables $(Y_1, Y_2)$ are comonotone if their are generated by the same random variable, i.e., if $(Y_1, Y_2) = \big(g_1(Z), \, g_2(Z)\big)$ $\P$-a.s. for $g_1, g_2\colon \R \to \R$ non-decreasing and a random variables $Z$.

To quantify out- and underperformance relative to the benchmark, the investor utilises the asymmetric $\alpha$-Bregman-Wasserstein ($\alpha$-BW) divergence, recently introduced in \cite{Pesenti2024ORL}, and recalled next. For a random variable $Z$, we write $F_Z(\cdot) := \P(Z \le \cdot)$  and $\Finv_Z(\cdot):=\inf\{y \in \R ~|~ F_Z(y) \ge u\}$, $u \in (0,1)$ to denote its  cumulative distribution functions (cdf) and (left-continuous) quantile function, respectively.

\begin{definition}[$\alpha$-Bregman-Wasserstein divergence - \cite{Pesenti2024ORL}]
Let $\phi\colon \R \to \R$ be a strictly convex and differentiable function,\footnote{For the exposition it is enough that $\phi$ is defined on $(c \, \Finvbench(0^+),  +\infty)$, where $\Finvbench(0^+):= \lim_{x \searrow 0}\Finvbench(x)$, and $c$ given in optimisation problem \eqref{opt:main-original}.} $\alpha \in(0,1)$, and $F_1, F_2$ be two univariate cdfs. Then, the $\alpha$-Bregman-Wasserstein divergence from $F_1$ to $F_2$ is defined by
\begin{equation*}
    \DBW(F_1, F_2)
    :=        \inf_{\pi\in\Pi(F_1,\,F_2)} \;\left\{\,\int_{\R^2} \big|\Id_{z_1\le z_2} - \alpha \big|\, B_\phi(z_1,z_2)
    \,\pi(\d z_1,\d z_2)\, \right\}\,,
\end{equation*}
 where $\Pi(F_1,F_2)$ is the set of all bivariate cdfs with marginal cdfs $F_1$ and $F_2$, respectively, and $B_\phi$ denotes the Bregman divergence given by
 \begin{equation}
 \label{eqn:Bregman}
     B_\phi\big(z_1, z_2\big)
    := \phi(z_1) - \phi(z_2) - \phi'(z_2) (z_1-z_2)
    \,,\quad z_1,z_2\in\R\,,
 \end{equation}
and where $\phi'(z):= \frac{d}{dz} \phi(z)$ is the derivative of $\phi$.
\end{definition}

The $\alpha$-BW divergence is an asymmetric divergence that arises from the Monge-Kantorovich problem when the cost function is a statistical scoring function that elicits the $\alpha$-expectile, and we refer to \cite{Pesenti2024ORL} for a detailed discussion.\footnote{The $\alpha$-expectile is elicitable \citep{Gneiting2011JASA}, that is it has representation
\begin{equation*}
    e_\alpha(F)
    =
    \argmin_{z \in \R} \int_\R\big|\Id_{y\le z} - \alpha \big|\, B_\phi(y,z)\, \d F(y)\,,
\end{equation*}
for any cdf $F$ such that the argmin exists and for any choice of Bregman divergence $B_\phi$ defined in \eqref{eqn:Bregman}.} 
As the $\alpha$-BW divergence from $F_1$ to $F_2$ is a divergence on the space of cdfs, we interchangeably write $\DBW(F_1, F_2)$ and $\DBW(\Finv_1, \Finv_2)$, where $\Finv_i$ denotes the (left-continuous) quantile function of $F_i$. For $\alpha = \frac12$, the $\alpha$-BW divergence reduces to one-half the Bregman-Wasserstein divergence \citep{carlier2007monge}, which we denote by $\BW(\cdot, \cdot):= \BW^\frac12(\cdot, \cdot) $. If furthermore $\phi(x) = x^2$, the $\frac12$-BW divergence coincides with one-half the squared 2-Wasserstein distance. While there is some literature on the BW divergence, this is the first work to consider an optimisiation problem with an $\alpha$-BW divergence.

By \cite{Pesenti2024ORL} the $\alpha$-BW divergence admits representation
\begin{equation}\label{eq:DWB}
    \DBW\left(F_1, F_2\right)
    =
    \int_0^1  \big|\Id_{\Finv_1(u)\le \Finv_2(u)} - \alpha \big|\; B_\phi\left(\Finv_1(u),\Finv_2(u)\right)
    \;\d u
    \,.
\end{equation}
Assume that the terminal wealth of a strategy $\btheta\in \mcA$ is comonotone with that of the benchmark, i.e., $(X_T^{\btheta}, X_T^{\bpi})$ is comonotone. This implies that $(X_T^{\btheta}, X_T^{\bpi}) = (\Finv_{X_T^{\btheta}}(U), \Finv_{X_T^{\bpi}}(U))$ $\P$-a.s. for the uniform random variable $U:= F_{X_T^{\bpi}}(X_T^{\bpi})$, and the $\alpha$-BW divergence from the investor's strategy's terminal wealth $X_T^{\btheta}$ to that of the benchmark $X_T^{\bpi}$ is $\DBW(\Finv_{X_T^{\btheta}}, \Finv_{X_T^{\bpi}})$. The asymmetry arises due to the Bregman divergence $B_\phi$ and $\phi$ should be chosen to penalise absolute differently to gains. Furthermore, outperformance and underperformance relative to the benchmark is also penalised asymmetrically. Indeed underperformance, i.e., $\Finv_{X_T^{\btheta}}(u) \le \Finv_{X_T^{\bpi}}(u)$, is penalised by $1-\alpha$, whereas outperformance, i.e., $\Finv_{X_T^{\btheta}}(u) > \Finv_{X_T^{\bpi}}(u)$, is penalised by $\alpha$. Thus, as the investors wish to outperform the benchmark, they choose $\alpha \le 0.5$. We note that the comonotonicity assumption between the investor's and the benchmark's terminal wealth is essential for the interpretation of under- versus outperformance. Indeed, by comonotonicity $\Finv_{X_T^{\btheta}}(u) \le \Finv_{X_T^{\bpi}}(u)$ corresponds to scenarios where $X_T^{\btheta}(\omega) \le X_T^{\bpi}(\omega)$, $\omega \in \Omega$, thus to underperforming the benchmark's terminal wealth.

Next, we establish some properties of the $\alpha$-BW divergence.

\begin{lemma}\label{lemma:convex-alpha-BW}
The $\alpha$-BW divergence is convex in its first argument on the space of quantile functions, i.e. $\DBW(w \Ginv_1 + (1-w) \Ginv_2 , \Finv) \le w\, \DBW(\Ginv_1, \Finv) + (1-w)\, \DBW(\Ginv_2, \Finv)$, for all $w \in [0,1]$ and quantile functions $\Ginv_1, \Ginv_2$, and $\Finv$.
\end{lemma}

\begin{proof}
Define the function 
\begin{equation*}
    S(x,y):=
    |\Id_{x\le y} - \alpha |\; B_\phi\left(x,y\right)\,.
\end{equation*}
Since $\phi$ is differentiable, the function $S(x,y)$ is differentiable in $x$ with derivative
\begin{equation*}
    \frac{\partial}{\partial x} \,S(x,y)
    =
    |\Id_{x\le y} - \alpha | \;\big(\phi'(x) - \phi'(y)\big)
    =
    \big(\Id_{x\le y} (1-\alpha) + \Id_{x> y} \alpha\big) \;\big(\phi'(x) - \phi'(y)\big)\,.
\end{equation*}
We observe that $\frac{\partial}{\partial x} \,S(x,y)$ is non-decreasing in $x$, for all $y$. To see this, note that by strict increasingness of $\phi'(\cdot)$ we have that $\frac{\partial}{\partial x} \,S(x,y)\Id_{x< y} <0< \frac{\partial}{\partial x} \,S(x,y)\Id_{x> y} $ and that $\frac{\partial}{\partial x} \,S(x,y)$ is non-decreasing on the sets $\{x \le y\}$ and $\{x > y\}$. Thus $S(x,y)$ is convex in $x$, for all $y$. Since $\DBW(\Finv_1, \Finv_2 ) = \int_0^1 S\big(\Finv_1(u), \Finv_2(u)\big)d u$, the $\DBW$ is convex in its first argument. 
\end{proof}

\begin{lemma}\label{lemma:ineq-BW}
    Let $\Finv_1, \Finv_2 $ be quantile functions. Then 
\begin{equation*}
    \DBW(\Finv_1,\Finv_2) \le\max\{ \alpha, 1-\alpha\} \, \BW (\Finv_1,\Finv_2)
    \le \BW (\Finv_1,\Finv_2)\,.
\end{equation*}   
\end{lemma}

\begin{proof}
For any quantile functions $\Finv_1, \Finv_2$, it holds that
\begin{equation*}
  |\Id_{\Finv_1(u)\le \Finv_2(u)} - \alpha |
  \le \max\{\alpha,  1-\alpha\}\,, \quad \text{for all} \quad u \in (0,1) \,.
\end{equation*}
Combining with \Cref{eq:DWB} yields the first inequality, the second follows as $\alpha \le 1$.
\end{proof}

Recall that the investor uses the $\alpha$-BW divergence to measure deviation from the benchmark. Specifically, the investor only considers alternative strategies whose $\alpha$-BW divergences to that of the benchmark is less or equal to a tolerance distance $\ep>0$. By \Cref{lemma:ineq-BW}, using the $\alpha$-BW divergence with $\alpha \ge \frac12$ allows the investor to choose from a larger set of strategies compared to using the BW divergence. 

While in our setting the market model is general, we provide a running example within a geometric Brownian motion framework.

\begin{example}[Geometric Brownian motion market model]
\label{ex: GBM-MM}
A classic example, which we use for numerical implementation, is the geometric Brownian motion (GBM) market model, in which the filtration is generated by a $d$-dimensional $\P$-Brownian motion $W=(W_t)_\tT$ with correlation structure $d[W^i,W^j]_t=\rho_{ij}\,dt$, for $i,j\in\mcD$, and we use $\rho$ to denote the matrix whose entries are $\rho_{ij}$. The risky asset prices satisfy the SDEs
\begin{align}
    dS_t^i &= S_t^i\left(\mu^i_t\,dt + \sigma^i_t \,dW_t^i \right), \qquad  i\in\mcD,
    \label{eqn:dS-GBM}
\end{align}
where $\mu_t=(\mu^1_t,\dots,\mu^d_t)$ is the vector of (deterministic) drifts and $\sigma_t=(\sigma_t^1,\dots,\sigma_t^d)$ is the vector of (deterministic) instantaneous volatilities, which are strictly positive. Moreover, the SDF $(\varsigma_t)_{\tT}$ satisfies the SDE
\begin{equation}
    d\varsigma_t = -\varsigma_t\,\left(r_t\,dt+ \lambda_t^\intercal\,dW_t\right)\,,
    \label{eqn:sdf-gbm}
\end{equation}
where $\lambda=(\lambda_t)_\tT$ denotes the vector-valued market price of risk with  $ \lambda_t =\rho^{-1}\frac{\mu_t-r_t\,\mathbf{1}}{\sigma_t}$ (where the division is component-wise) and the short rate of interest $r_t$ is deterministic.

We further consider a deterministic benchmark portfolio, that is a deterministic $\pi$, and refer to \cite{basak2006risk,ng2023portfolio} for stochastic benchmark examples. Then the benchmark's terminal wealth distribution and quantile function are
\begin{subequations}
\begin{align}
    F_{X_T^{\bpi}}(x) &= \Phi\left(\frac{\log(x/X_0^\pi)-\Gamma+\tfrac12 \Psi^2}{\Psi}\right)\,, \qquad \text{respectively}
    \label{eqn:F-for-deterministic-bench}
    \\[0.25em]
    \Finv_{X_T^{\bpi}}(u) &= X_0^\pi\,\exp\left\{(\Gamma-\tfrac12 \Psi^2)+\Psi\,\breve{\Phi}(u)\right\}, \label{eqn:Finv-for-deterministic-bench}
\end{align}
\end{subequations}
where, 
\begin{equation}
\label{eqn:Gamma-and-Psi}
    \Gamma := {\textstyle\int_0^T} \big(\left(\mu_s-r_s\,\mathbf{1}\right)^\intercal\,\pi_s + r_s\big) \,\d s\;,
    \quad 
    \Psi^2 := {\textstyle\int_0^T} \pi_s^\intercal \Upsilon_s\,\pi_s\,\d s\,,
\end{equation}
$\Upsilon$ is the covariance matrix with $\Upsilon^{ij}_t=\sigma^i_t\rho_{ij}\sigma^j_t, \;i,j\in\mcD$, and where $\Phi$ and $\breve{\Phi}$ are the cdf and quantile function of a standard normal distribution.

\begin{figure}[th]
    \centering
    \begin{minipage}{0.45\textwidth}
    \includegraphics[width=\linewidth]{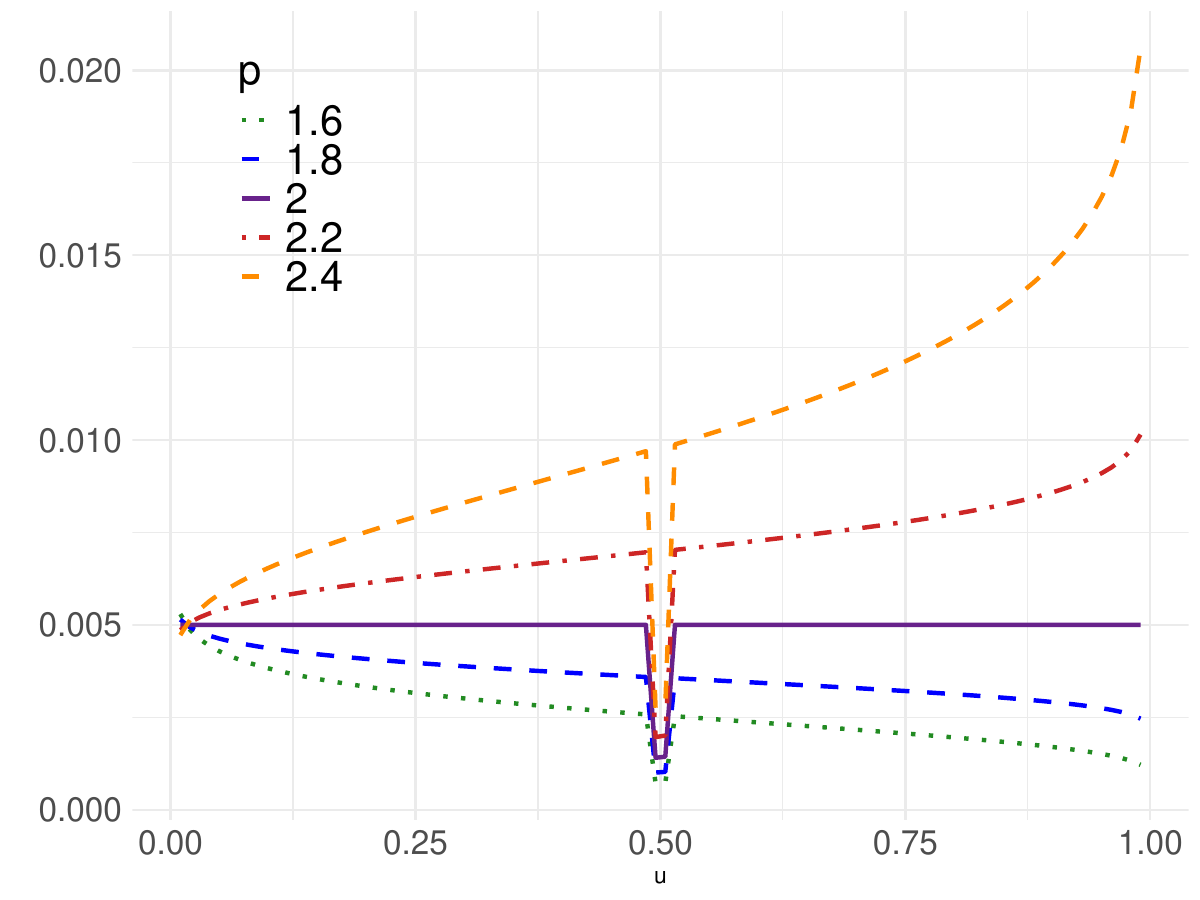}
            \end{minipage}%
    \hfill
       \begin{minipage}{0.48\textwidth}
    \includegraphics[width=\linewidth]{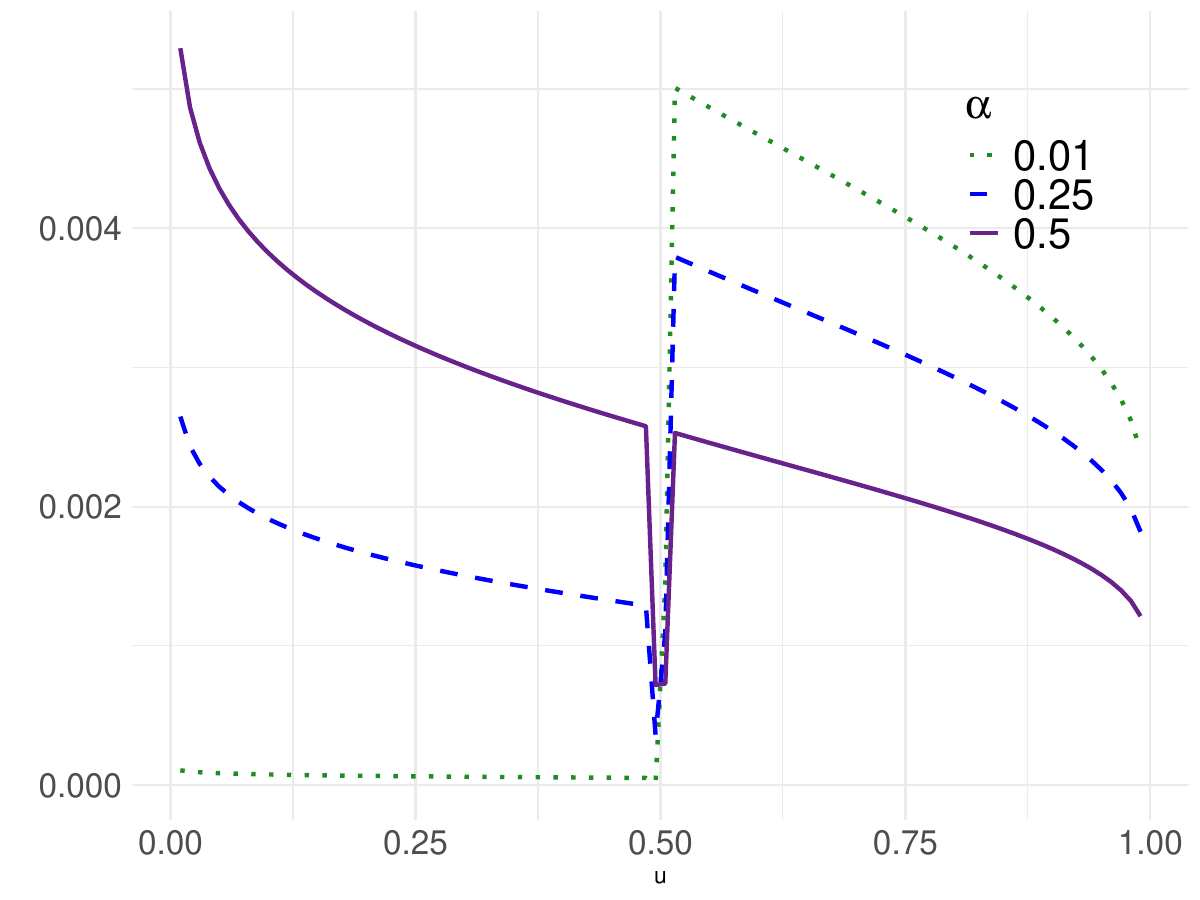}
        \end{minipage}%
 
    \vspace{1em}
    \centering    
      \begin{minipage}{0.45\textwidth}
      \includegraphics[width=\linewidth]{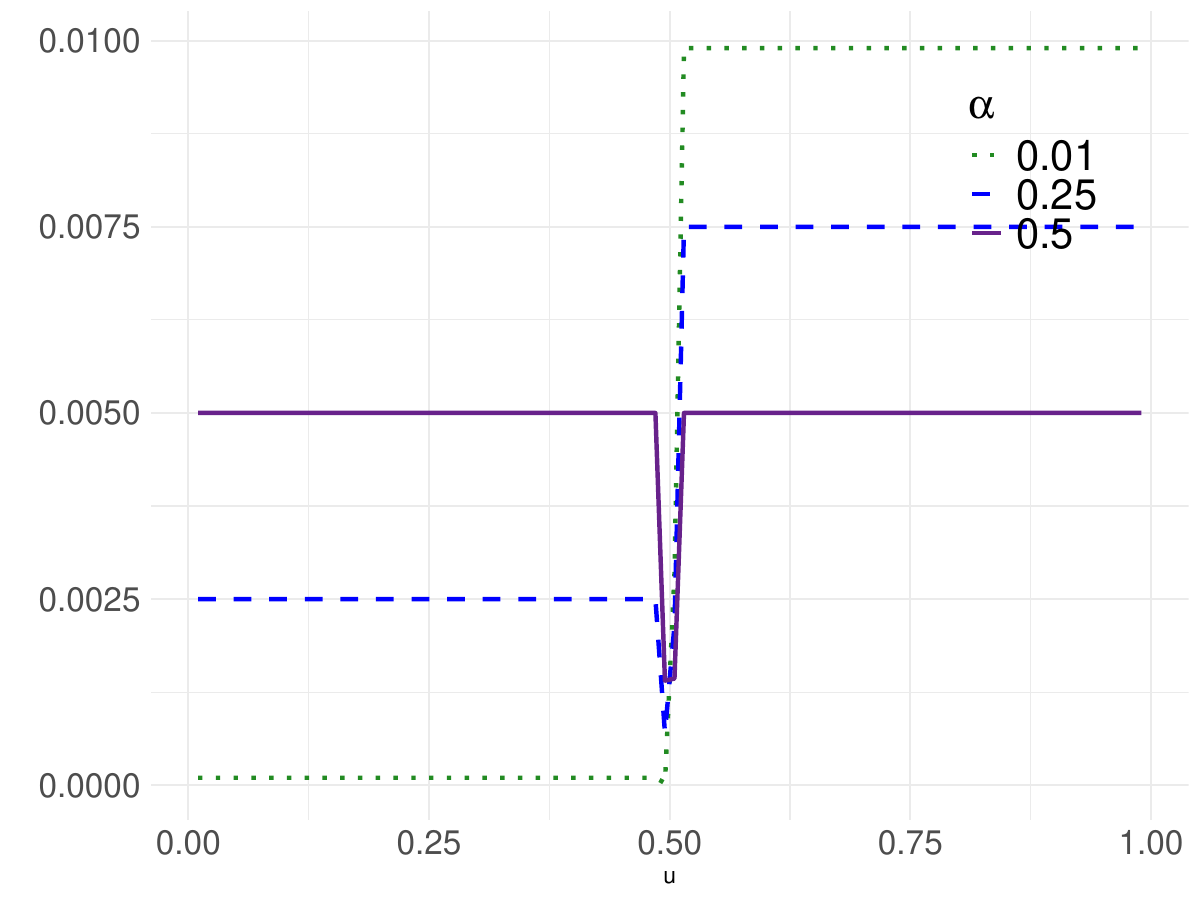}
    \end{minipage}%
    \hfill
       \begin{minipage}{0.48\textwidth}
           \includegraphics[width=\linewidth]{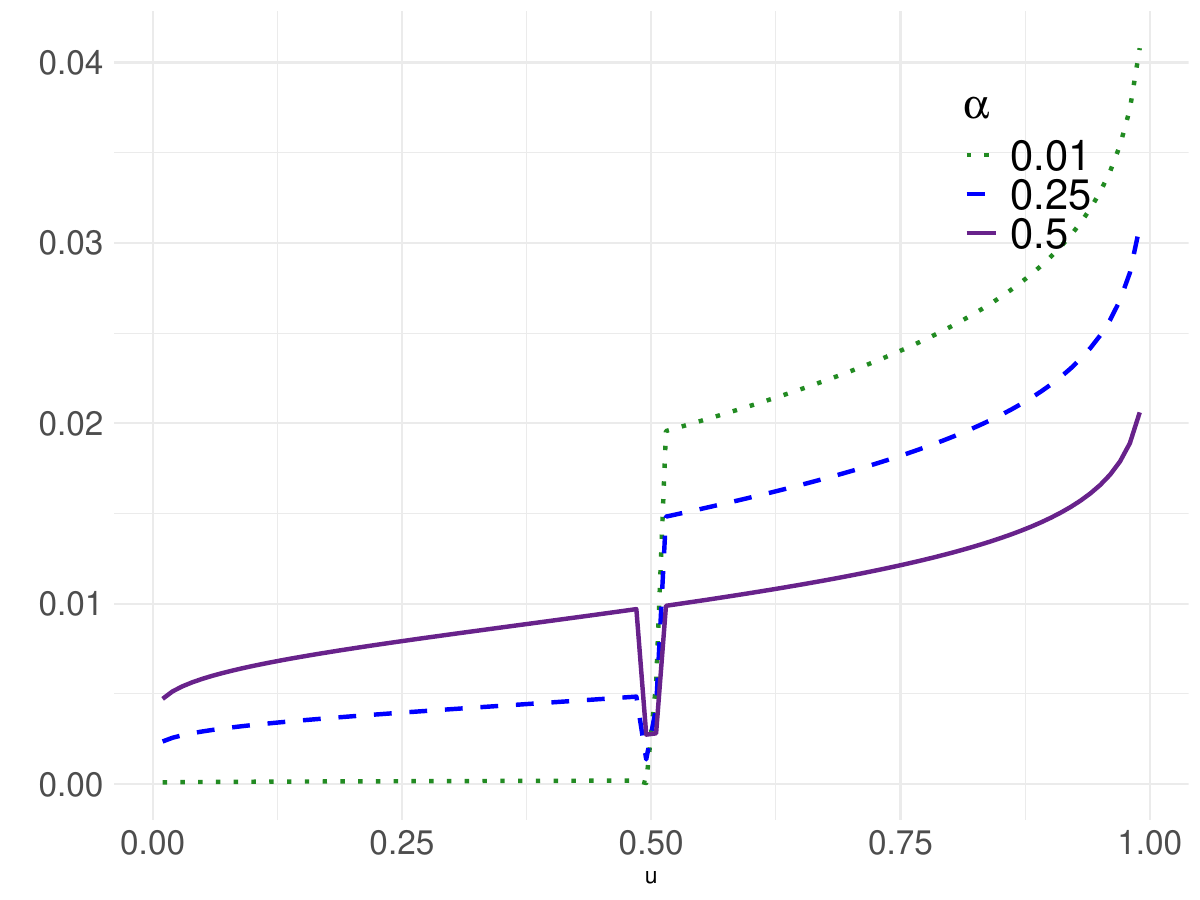}
           \end{minipage}%
    \caption{Integrand of the $\alpha$-BW divergence for the Power family $\phi_p$ from the quantile function $\Ginv$ given in \eqref{eq:modified-bench} to the benchmark quantile $\Finv_{X_T^{\bpi}}$. Top left panel: $p \in \{1.6, 1.8, 2, 2.2, 2.4\}$ and $\alpha = 0.5$. Top right panel: $\alpha \in  \{0.01, 0.25, 0.5\}$ and  $p = 1.6$. Bottom left panel: $\alpha \in \{0.01, 0.25, 0.5\}$ and  $p = 2$. Bottom right panel: $\alpha \in \{0.01, 0.25, 0.5\}$ and  $p = 2.4$.
  }
    \label{fig:BW-integrand-benach-linear-shift}
\end{figure}
\Cref{fig:BW-integrand-benach-linear-shift} illustrates the $\alpha$-BW divergence by plotting the integrand of the $\alpha$-BW divergence (see \eqref{eq:DWB}), from a quantile function $\Ginv$ to the the benchmark quantile $\Finv_{X_T^{\bpi}}$, i.e., the integrand of $\DBW(\Ginv, \Finv_{X_T^{\bpi}})$. The quantile function $\Ginv$ is constructed to dominate $\Finvbench$ for small values of $u \in (0,1)$ and to be smaller than the benchmark's quantile for large values of $u \in (0,1)$. Specifically
\begin{equation}\label{eq:modified-bench}
    \Ginv(u) := \min\left\{\Finvbench(u) + 0.1, \Finvbench(0.5)\, \right\} + \max\left\{\Finvbench(u)  -  \Finvbench(0.5)- 0.1, 0\, \right\}\,.
\end{equation}
Thus, the quantile function $\Ginv$ dominates the benchmark by 0.1, for all $u$ satisfying $\Finvbench(u) \le \Finvbench(0.5)-0.1$. Moreover, the benchmark dominates $\Ginv$ by 0.1 whenever $\Finvbench(u)  \ge \Finvbench(0.5) + 0.1$, and for all other $u$, the quantile function $\Ginv$ is equal to $\Finvbench(0.5)$. By construction, the absolute difference between $\Ginv$ and $\Finvbench$ is the same for small and large values of $u$. For calculating the benchmark, we choose the market parameters $\Gamma = 2$, $\Psi= 0.8$, $X_0^\pi = 1$, which means that the benchmark has a lognormal distribution $LNorm(\nu, \vartheta^2$), with $\nu= 1.68$ and $\vartheta = 0.8$, where parameters are on the log scale. 

We consider the Power family of Bregman generator given by $\phi_p(x) = \frac{2 x^p}{p(p-1)}$, $x>0$ and $p>1$, which is compatible with the setup that the benchmark quantile function is non-negative. The scale factor of $2$ guarantees that the $\phi_2(\cdot) = x^2$, that is for $p = 2$, we recover the squared Wasserstein distance.

All panels of \Cref{fig:BW-integrand-benach-linear-shift} display the integrand of the $\alpha$-BW divergence with the Power generator $\phi_p$ from $\Ginv$ to $\Finvbench$ for different choices of parameters $p$ and $\alpha$ of the $\alpha$-BW divergence. The top left panel shows the behaviour for different $p \in \{1.6, 1.8, 2, 2.2, 2.4\}$ and fixed $\alpha = 0.5$. Recall that for $\alpha = 0.5$ there is no distinction between under- and outperforming the benchmark, thus in the top left panel we only observe the impact of the Bregman divergence. We observe that for $p = 2$ (which corresponds to one-half the symmetric squared Wasserstein distance), we obtain strait lines, whenever the absolute difference between $\Ginv$ and $\Finvbench$ is equal to $0.1$, i.e. on the set $A:= \{u \in (0,1) ~|~|\Ginv(u) - \Finvbench(u)| = 0.1\}$. When $p<2$, the integrands are, when restricted to $A$, non-increasing functions, thus exaggerating deviations from the benchmark whenever $u$ is small. This is in contrast to $p>2$, where the integrands are, when restricted to $A$, non-decreasing functions; exaggerating deviations from the benchmark when $u$ is large. The other three plots display the integrand of the $\alpha$-BW divergence for $\alpha \in \{0.01, 0.25, 0.5\}$ and difference choices of $p = 1.6$ (top right panel), $p = 2$ (bottom left panel), $p = 2.4$ (bottom right panel). Recall that $\alpha\neq \frac12$ allows to differently penalises underperformance with respect to the benchmark with the factor $\alpha$, while outperformance with respect to the benchmark with factor $1-\alpha$; thus $\alpha < 0.5$ penalises underperformance more. In all three plots we observe that the smaller $\alpha$, the less deviations from the benchmark are penalised in the left tail where outperformance occurs. This is different to the right tail (underperformance), where we observe a larger penalisation for smaller values of $\alpha$. In conclusion, an investor who wishes to deviate from the left tail of benchmark may choose $p>2$ (allowing for smaller and larger losses), and an investor who wishes to be deviate from to right tail of the benchmark considers $p<2$ (allowing for smaller or larger gains).

\vspace{1em}   
\end{example}

\subsection{The investor's optimisation problem}

The investor aims to find a strategy that maximises the utility of the difference between the selected portfolio's terminal wealth and a proportion of the benchmark's while (i) satisfying a budget constraint, that is begin not more expensive than the budget $x_0>0$, (ii) being comonotone to the terminal wealth of the benchmark, and (iii) over- and underperformance is penalised via the $\alpha$-BW divergence.

We denote the cost of the benchmark strategy by $y_0:= \E\big[ \varsigma_TX_T^\pi\big]$, which is not necessarily equal to the budget $x_0$, and make the following assumption on the benchmark's terminal wealth.

\begin{assumption}\label{asm:bench-cont}
    The terminal wealth of the benchmark $X_T^\pi$ is continuously distributed and  has finite expected utility, i.e., $\E[U(X_T^\pi)]<+\infty$.
\end{assumption}

The investor's optimisation problem is given as follows.

\begin{definition}
Let $c \in [0,1]$. An optimal strategy $\btheta^*\in\mcA$ is a solution to the optimisation problem
\begin{align}
\label{opt:main-original}
\max_{\btheta \in \mcA}  \; \E\left[U\big(X^{\btheta}_T - c \, X^{\bpi}_T\big)\right]\,,\qquad
& \text{subject to} \qquad
(i) \quad  \E\left[\varsigma_T X^{\btheta}_T\right] \le x_0\,, 
\\[0.5em]
&\phantom{\text{subject to}}\qquad (ii) \quad (X^{\btheta}_T, X^{\bpi}_T) \quad \text{are comonotone,}\quad \text{and}
\notag
\\[0.5em]
&\phantom{\text{subject to}}\qquad(iii) \quad \DBW\left(F_{X^{\btheta}_T}, F_{X^{\bpi}_T}\right)  \le \ep\,.\notag
\end{align}
\end{definition}

By definition of the utility function, we have $U(x) = - \infty$, whenever $x < 0$. Thus the objective function in optimisation problem \eqref{opt:main} is equivalent to 
\begin{equation*}
    \max_{\btheta \in \mcA}  \; \E\left[U\left(X^{\btheta}_T - c \, X^{\bpi}_T\right)\right]
    =
    \max_{\btheta \in \mcA}  \; \E\left[U\left(\big(X^{\btheta}_T - c \, X^{\bpi}_T\big)_+\right)\right]\,,
\end{equation*}
that is the investor only considers strategies, whose terminal wealth $X_T^\btheta$ is $\P$-a.s. larger than a proportion of the benchmark terminal wealth $c X_T^\bpi$. The exogenous constant $c\in [0,1]$ chosen by the investor determines how much they assess their utility relative to the benchmark. If $c = 0$, then the investor solely assesses the utility of their strategy. If $c>0$, the investor cares about outperforming a proportion of the terminal wealth of the benchmark, i.e. $cY$. This setting may be of interest to an investor whose bonus is solely determined by the outperformance of a benchmark. The constant $c$ could for example be chosen inverse proportional to the wealth of the benchmark -- as we do not make the assumption that the benchmark's cost is $x_0$. We reiterate that the comonotonicity assumption is key for the distributional $\alpha$-BW divergence to be interpreted as an out- and underperformance of the benchmark's terminal wealth. Indeed if $(X_T^\btheta, X_T^\bpi)$ is a comonotonic vector, then it holds that
\begin{equation*}
    \DBW\left(F_{X^{\btheta}_T}, F_{X^{\bpi}_T}\right) 
    = 
    \E\left[ \big|\Id_{X^{\btheta}_T\le {X^{\bpi}_T}} - \alpha \big|\; B_\phi\left(X^{\btheta}_T,{X^{\bpi}_T}\right)
    \right]\,,
\end{equation*}
thus underperformance occur whenever $X^{\btheta}_T\le {X^{\bpi}_T}$. 
As the utility equals $-\infty$ whenever evaluated at a negative number, any feasible strategy $\btheta$ satisfies $X_T^\btheta \ge c\, X_T^\bpi$, which by monotonicity is equivalent to $F_{X^{\btheta}_T}(u) \ge c\, F_{X^{\bpi}_T}(u)$. Thus, in the $\alpha$-BW divergence $\DBW\big(F_{X^{\btheta}_T}, F_{X^{\bpi}_T}\big)$ outperformance occurs if $F_{X^{\btheta}_T}(u) \ge F_{X^{\bpi}_T}(u)$ while underperformance materialises when $cF_{X^{\bpi}_T}(u)\le F_{X^{\btheta}_T}(u) \le F_{X^{\bpi}_T}(u)$.

We note that the investor's optimisation problem \eqref{opt:main-original}, that is its objective function and the constraints, only pertains to the terminal wealth of admissible strategies and that of the benchmark's, i.e. to $\F_T$-measurable random variables. Since the market model is complete, solving for the optimal terminal wealth random variable -- rather than the optimal strategy -- is sufficient, as the construction of the unique strategy that attains the terminal wealth follows by standard replicating strategies via, e.g., PDE or Monte Carlo regression methods or using Gaussian processes; indicatively see \cite{giles2005smoking} and \cite{Ludkovski2020WP}.

Thus, we first rewrite optimisation problem \eqref{opt:main-original} as an optimisation problem over $\F_T$-measurable rv. For simplicity of notation, we write $X:= X_T^{\btheta}$, $Y:= X_T^{\bpi}$, and $\varsigma:= \varsigma_T$. Moreover, we denote by $F_X(\cdot)$, $\Fbench(\cdot)$, $F_\varsigma(\cdot)$ their cdfs and by $\Finv_X(\cdot)$, $\Finvbench(\cdot)$, $\Finvsig(\cdot)$ their quantile functions. With this notation, the investor considers the optimisation problem

\begin{align}
\label{opt:main}\tag{P}
\max_{X \;\;  \F_T\text{-measurable} }  \; \E\Big[U\big(X - c Y\big)
\Big]\,, \qquad
& \text{subject to} \qquad
(i) \quad  \E\left[\varsigma X\right] \le x_0\,, 
\notag
\\
&\phantom{\text{subject to}}\qquad (ii) \quad (X, Y) \quad \text{are comonotone,}\quad \text{and}
\notag
\\[0.5em]
&\phantom{\text{subject to}}\qquad(iii) \quad \DBW\left(F_{X}, \Fbench\right)  \le \ep\,.\notag
\end{align}

We observe that for $c = 1$ the $\alpha$-BW constraint reduces to $\BW(F_X, F_Y) \le \frac{\ep}{2\, \alpha}$. To see this, note that $U(x) =  - \infty $, for $x<0$, implies that in this case the maximum in \eqref{opt:main} is taken over all $X \ge Y$, which by comonotonicity of $(X,Y)$ is equivalent to $\Finv_X(u) \ge \Finvbench(u) $ for all $u \in (0,1)$. Thus, in this case the $\alpha$-BW divergence is
\begin{align*}
    \BW^\alpha(\Finv_X, \Finvbench) 
    &=
    \alpha  \int_0^1   B_\phi\left(\Finv_X(u),\Finvbench(u)\right)
    \;\d u    
    =2\,\alpha \, \BW(\Finv_X, \Finvbench)\,,
\end{align*}
where we recall that $\BW(\cdot, \cdot) =\BW^\frac12(\cdot, \cdot) $. Therefore, the choice $c = 1$ reduces the $\alpha$-BW divergence to the BW divergence and a tolerance distance of $\frac{\ep}{2 \alpha}$.

We recall the usual terminology of constraints being binding and satisfied, which we will use throughout. Specifically, let $X$ be a feasible solution to optimisation problem \eqref{opt:main}. Then we say that a constraint, e.g., the budget constraint, is satisfied, if $\E[\varsigma X] \le x_0$. We say that a constraint is binding, e.g., the budget constraint, if $\E[\varsigma X] = x_0$.

\subsection{Quantile reformulation}\label{sec:quantile-reform}

We employ the quantile technique (see e.g., \cite{Dybvig_88,Foellmer_Schied_04,He_Zhou_11}) and rewrite optimisation problem \eqref{opt:main} over random variables as an optimisation problem over the space of quantile functions. The comonotonicity constraint between the terminal wealth of the investor's and the benchmark's strategy, renders the quantile formulation not classical. Closest is \cite{Pesenti2023SIAM} who consider a copula constraint between the terminal wealth of the investor's and the benchmark's strategy.

For this, we define the set of functions
\begin{equation*}
 \mcH:=\big\{\; H ~|~H : (0,1)\to \R \; \big\}\,
\end{equation*}
and the set of quantile functions as 
\begin{equation*}
 \mcM:=\left\{ \;\Ginv ~|   ~ \Ginv: (0,1)\to \R \,,\quad \text{non-decreasing} \quad  \&\quad \text{left-continuous}\; \right\} \subset \mcH\,.
\end{equation*}

Next, we rewrite the cost of a strategy in terms of the quantile function of its terminal wealth distribution. As we restrict to strategies whose terminal wealth is comonotone to that of the benchmark, the copula between $(X, \varsigma)$ is predetermined. Indeed by comonotonicity, $X$ is a non-decreasing transformation of $Y$, thus the copula (which is invariant under non-decreasing transformations) of $(X, \varsigma)$ equals the copula of $(Y, \varsigma)$. The next result gives a quantile representation of the budget constraint.
\begin{proposition}[Budget constraint]\label{prop:budget-como}
Let Assumptions \ref{asm:sdf-cont} and \ref{asm:bench-cont} be satisfied and $X$ be comonotone to $Y$. Then,
\begin{equation}\label{eq:xi}
    \E[\varsigma X] = \int_0^1 \Finv_X(u) \xi(u) \, du\,,
    \qquad \text{where} \qquad
    \xi(u) := \frac{\hat{f}_Y\big(\Finvbench(u)\big)}{f_Y\big(\Finvbench(u)\big)}\,,
\end{equation}
    and where $f_Y(x):=  \frac{d}{dy}F_Y(y)$ and $\hat{f}_Y(y):= \frac{d}{dy }\E[\varsigma\, \Id_{\{Y\le y\}}]$. 
    
\end{proposition}

\begin{proof}
    Let $X$ be such that $(X,Y)$ is comonotone. Then, it holds that $X = \Finv_X\big(\Fbench(Y)\big)$ a.s., see e.g., proposition 2.1 (d) in \cite{cuestaalbertos1993JMA}. Moreover, 
    \begin{equation}\label{eq:budget-conditional}
        \E[\varsigma X]
        =
        \E\left[\varsigma \,  \Finv_X\big(\Fbench(Y)\big)\right]
        =
        \E\left[\, \E\left[\varsigma \mid\sigma(Y)\right]\,  \Finv_X\big(\Fbench(Y)\big) \right]\,,
    \end{equation}
    where $\sigma(Y)$ is the natural $\sigma$-algebra generated by $Y$. Next, we claim that 
    \begin{equation}\label{eq:xi-conditional}
        \E\left[\varsigma \mid\sigma(Y)\right] = \frac{\hat{f}_Y(Y)}{f_Y(Y)}\,.
    \end{equation}
To show that the conditional expectation satisfies \eqref{eq:xi-conditional}, we show that the left- and right-hand side, when multiplied by $\psi(Y)$, where $\psi \colon \R \to \R$ is a Borel-measurable test function, are equal under expectation. For the right-hand side we fist note that 
\begin{equation*}
    \hat{f}_Y(y)
    = 
    \frac{d}{dy }\E\left[\frac{\varsigma}{\E[\varsigma]}\, \Id_{\{Y\le y\}}\right]\, \E[\varsigma]
    = 
    \frac{d}{dy }\E^{\Q_\varsigma}\left[\Id_{\{Y\le y\}}\right]\, \E[\varsigma]
    =
    f_Y^{\Q_\varsigma}(y) \, \E[\varsigma]\,,
\end{equation*}
where we define the probability measures $\Q_\varsigma$ by the Radon-Nikodym derivative $\frac{d \Q_\varsigma}{d \P}:= \frac{\varsigma}{\E[\varsigma]}$ and denote the density of $Y$ under $\Q_\varsigma$ by $f_Y^{\Q_\varsigma}(\cdot)$. Using this, we obtain that 
\begin{align*}
    \E\left[ \frac{\hat{f}_Y(Y)}{f_Y(Y)}\; \psi(Y) \right]
    &=
    \int \frac{ f_Y^{\Q_\varsigma}(y) }{f_Y(y)}\; \psi(y) f_Y(y)dy\, \E[\varsigma]
    \\
    &=
    \int f_Y^{\Q_\varsigma}(y)\; \psi(y) dy\, \E[\varsigma]
    \\
    &=
    \E^{\Q_\varsigma}\left[ \psi(Y) \right]\, \E[\varsigma]
    \\
    &= 
    \E\left[\frac{\varsigma}{\E[\varsigma]} \psi(Y) \right]\, \E[\varsigma]
    \\
    &= 
    \E\left[\varsigma \psi(Y) \right]\,.
\end{align*}
For the left-hand side, we have
\begin{equation*}
      \E\big[\E\left[\varsigma \mid\sigma(Y)\right] \, \psi(Y)\big]
      = 
      \E\big[\varsigma  \, \psi(Y)\big]\,,
\end{equation*}
thus equation \eqref{eq:xi-conditional} holds. Next we return to \eqref{eq:budget-conditional}, and use \eqref{eq:xi-conditional} in the second equation and that $\Fbench(Y) = U$ a.s., where $U\sim U(0,1)$,
\begin{align*}
    \E[\varsigma\, X] 
    &=
    \E\left[\, \E\left[\varsigma \mid\sigma(Y)\right]\,  \Finv_X\big(\Fbench(Y)\big) \right]
    =
    \E\left[\, \frac{\hat{f}_Y(Y)}{f_Y(Y)}\,  \Finv_X\big(\Fbench(Y)\big) \right]
    \\
    &=
    \E\left[\, \frac{\hat{f}_Y\big(\Finvbench(U)\big)}{f_Y\big(\Finvbench(U)\big)}\,  \Finv_X(U) \right]
    =
    \int_0^1  \xi(u)\, \Finv_X(u) \d u\,,
\end{align*}
    which concludes the proof.
\end{proof}

With the quantile representation of a strategy's cost, we can now state the quantile reformulation of optimisation problem \eqref{opt:main}.

\begin{theorem}(Quantile reformulation)
\label{thm:quantile-reform}
    Let Assumptions \ref{asm:sdf-cont} and \ref{asm:bench-cont} be satisfied and consider the optimisation problem
\begin{align*}
\max_{\Ginv \in \mcM}  \; \int_0^1 U\Big(\Ginv(u) - c \Finvbench(u)\Big)\, \d u \,,
    \qquad \text{subject to} \qquad
    &(i) \quad \int_0^1 \Ginv(u) \xi (u)\d u\le x_0\,, \quad \text{and}
    \notag
    \\[0.5em]
    &(ii) \quad \DBW\big(\Ginv, \Finvbench\big)  \le \ep\,.\tag{\text{$\breve{P}$}}
    \label{opt:main-prime}
\end{align*}
Then, if a solution to \eqref{opt:main-prime} exists it is unique. Moreover,  denote by $\Ginv^*\in \mcM$ a solution (if it exists) to \eqref{opt:main-prime}, then $X^*:= \Ginv^*\big(\Fbench(Y)\big)$ is a solution to optimisation problem \eqref{opt:main}. Conversely, a solution to optimisation problem \eqref{opt:main} (if it exists) has quantile function $\Ginv^*$.
\end{theorem}

\begin{proof}
The optimisation problem \eqref{opt:main-prime} has a strictly concave objective function, since the utility function is strictly concave. Moreover, the budget constraint is linear and the $\alpha$-BW divergence is convex, see \Cref{lemma:convex-alpha-BW}. Therefore, if a solution exists it is unique. 

Since $(X,Y)$  are comonotone, it holds that $(X,Y) = (\Finv_X(U), \Finvbench(U))$ $\P$-a.s. for any uniform random variable $U \sim U(0,1)$. Thus, the expected utility admits quantile representation $\E\big[U(X - c Y)\big] = \int_0^1 U\big(\Ginv(u) - c \Finvbench(u)\big)\, \d u $. Moreover, by \Cref{prop:budget-como} the cost of $X$ is $\E[\varsigma X] = \int_0^1 \Finv_X(u) \xi(u) \, \d u$.  

Let $\Ginv^*$ be the solution to \eqref{opt:main-prime} and define $X^*:= \Ginv^*\big(\Fbench(Y)\big)$. Then $(X^*,Y)$ is comonotone  -- $X^*$ is a non-decreasing function of $Y$ --, and $X^*$ is admissible for optimisation problem \eqref{opt:main}. Moreover, for any random variable $X$ with quantile function $\Ginv$ that is admissible to \eqref{opt:main}, it holds that 
\begin{equation*}
    \E\Big[U\big(X - c Y\big)
\Big]\,= \int_0^1 U\Big(\Ginv(u) - c \Finvbench(u)\Big)\, \d u \,\le \int_0^1 U\Big(\Ginv^*(u) - c \Finvbench(u)\Big)\, \d u \,=\E\Big[U\big(\hat{X} - c Y\big)
\Big]\,,
\end{equation*}
implying that $X^*$ is a solution to optimisation problem \eqref{opt:main}. Conversely, if $X^\star$ is a solution to \eqref{opt:main}, then by equivalence of the objective function and the constraints of \eqref{opt:main} and \eqref{opt:main-prime}, the quantile function of $X^\star$ has to be a solution to \eqref{opt:main-prime}. By uniqueness of the solution to \eqref{opt:main-prime}, the quantile function of $X^\star$ has to be $\Ginv^*$. 
\end{proof}

If the benchmark and market model are specified, the $\xi$ function can be computed. Here, we give an explicit formula for $\xi$ in the GMB market model and delegate its derivation to \Cref{app:GBM}.

\begin{example}[GMB market model]\label{ex:gbm-xi}
Continuing Example \ref{ex: GBM-MM} we have that $\xi(u)$, as defined in \eqref{eq:xi}, is given by
\begin{align}
    \xi(u) 
    &= e^{-R}\,\frac{\phi\left(\breve{\Phi}(u) + \frac{\Gamma-R}{\Psi}\right)}
    {\phi\left(\breve{\Phi}(u)\right)}\,,
    \label{eqn:R-N-deterministic-delta}
\end{align}
where $\phi$ denote the density of the standard normal distribution and $R:=\int_0^T r_s\,ds$ is the total interest.

Note that $\xi$ in \eqref{eqn:R-N-deterministic-delta} is monotone. This can be seen by computing $\frac{d}{du}\xi(u)$, which, after calculations, is given by
\begin{equation*}
    \frac{d}{du}\xi(u) = -\frac{\Gamma-R}{\Psi} \,\frac{\xi(u)}{\phi\left(\breve{\Phi}(u)\right)}.
\end{equation*}
Hence, as $\Psi>0,\, \xi(u)>0$, and $\mbox{sgn}(\frac{d}{du}\xi(u)) = - \mbox{sgn}(\Gamma-R)$, we have that $\xi$ is strictly increasing if $\Gamma<R$ and strictly decreasing if $\Gamma>R$. As $\Gamma$ represents the expected cumulative log return of the benchmark, we expect that any sensible investor chooses a benchmark portfolio such that $\Gamma>R$. Thus, for the remainder of this example, we assume $\Gamma>R$, which implies that $\xi$ is strictly decreasing.    
\end{example}

\subsection{Well-posedness}\label{sec:unique-well-posed}
This section is devoted to the well-posedness of optimisation problem \eqref{opt:main-prime}. Throughout, we use the terminology $\alpha$-BW ball to refer to the set of quantile functions that lie within an $\ep$ tolerance (measured with the $\alpha$-BW divergence) of the benchmark's quantile function, i.e. $\{ \Ginv \in \mcM ~|~ \BW^\alpha(\Ginv, \Finvbench) \le \ep \ \}$. A naming often used with respect to the Wasserstein distance.

\begin{proposition}\label{prop:one-const-binding}
Let Assumptions \ref{asm:sdf-cont} and \ref{asm:bench-cont} be satisfied. If a solution to optimisation problem \eqref{opt:main-prime} exists, then at least one of the constraints (the budget or the $\alpha$-BW) are binding.
\end{proposition}
\begin{proof}
We show that at least one of the constraints, the budget or the $\alpha$-BW divergence constraint, is binding. For this let $Z:= \Finv_Z(\Fbench(Y))$, where $\Finv_Z$ is an optimal solution to optimisation problem \eqref{opt:main-prime} satisfy $\DBW(\Finv_Z, \Finvbench) =:\ep_Z < \ep$ and $\int_0^1 \Finv_Z(u) \xi(u) \, du=: z_0< x_0$. For $\delta\ge0$, define the random variable $Z_\delta:= Z + \delta$. Then the budget and the $\alpha$-BW divergence are
\begin{align*}
    c(\delta):&= \int_0^1\Finv_{Z_\delta}(u) \xi(u) \d u 
    =
    z_0 + \delta \quad \text{and} 
    \\
    b(\delta):&= \DBW(\Finv_{Z_\delta}, \Finvbench)
    =
    \DBW(\Finv_{Z} + \delta, \Finvbench)\ge 0\,.
\end{align*}
Now as both functions are continuous and satisfy $\lim_{\delta \to 0} c(\delta) = z_0$ and $\lim_{\delta \to 0} b(\delta) = \ep_Z$, there exists $\delta^*>0$ such that $c(\delta^*) \le x_0$ and $b(\delta^*) \le \ep$. Thus  $Z_{\delta^*}$ is comonotone to $Y$, satisfies the budget constraint and the $\alpha$-BW constraint, and has a strictly larger expected utility, leading to a contradiction to the optimality of $Z$. Thus, the solution has to bind at least one constraint.
\end{proof}

It is natural to set a lower bound on the tolerance distance of the $\alpha$-BW divergence, such that there exists a random variable satisfying both the $\alpha$-BW divergence and the budget constraint. Thus, we define the minimal tolerance distance
\begin{equation}\label{opt:ep-min}
    \epsilon_{\min}:=\min_{\Ginv \in \mcM}  \; \DBW(\Ginv, \Finvbench) \,,
    \qquad \text{subject to} \qquad
   \int_0^1 \Ginv(u) \xi(u) \, du \le x_0\,,  
   \tag{$\breve{P}_{\ep_{\min}}$}
\end{equation}
which is implicitly dependent on $x_0$. For simplicity, we omit its dependence on $x_0$.

If $\ep < \ep_{\min}$, then the constraint set is empty, meaning there does not exist a random variable that simultaneously fulfils the budget and the $\alpha$-BW divergence constraint. Thus for optimisation problems \eqref{opt:main-prime} to be well-posed it is necessary that $\epsilon \ge \ep_{\min}$. Clearly if the benchmark portfolio satisfies the budget constraint, then $\ep_{\min} = 0$.

Key to the representation of $\ep_{\min}$ and a solution to \eqref{opt:main-prime} is the so-called isotonic projection, recalled next. The isotonic (antitonic) projection of a function $\ell\colon (0,1)\to \R$ is its $L^2$ projection onto the non-decreasing (non-increasing) functions.\footnote{The isotonic projection of $\ell$ is intimately connected to the concave envelope. Indeed $\ell^\uparrow$ is the derivative of the concave envelope of the integral of $\ell$ \citep{brighi1994SIAMNA}.} We refer the reader to e.g., \cite{brunk1965AMS}, \cite{Barlow1972JASA} and \cite{brighi1994SIAMNA}. 

\begin{definition}[Isotonic and antitonic projection]
Let $\ell\colon (0,1) \to \R$ be a square integrable function. Then its isotonic projection $\ell^\uparrow$ of $\ell$ is its $L^2$ projection onto the non-decreasing and left-continuous functions, that is the unique solution to
    \begin{equation*}
        \ell^\uparrow : = \argmin_{g \in \mcM } \; \int_0^1 \big(\ell(u) - g(u)\big)^2\, du\,.
    \end{equation*}
    The antitonic projection $\ell^\downarrow$ of $\ell$ is its $L^2$ projection onto the non-increasing and left-continuous functions, that is the unique solution to
    \begin{equation*}
        \ell^\downarrow : = \argmin_{-g\in \M} \; \int_0^1 \big(\ell(u) - g(u)\big)^2\, du\,.
    \end{equation*}
\end{definition}

The next result provides the representation for $\ep_{\min}$.

\begin{proposition}[Minimal tolerance distance]\label{Prop:epsilon-bound}
Let Assumptions \ref{asm:sdf-cont} and \ref{asm:bench-cont} be satisfied and assume that $\phi$ satisfies $\lim_{x\to -\infty} (\phi')^{-1}(x)\le 0$. Then, $\ep_{\min}$ is unique and given by
    \begin{equation*}
    \epsilon_{\min}
    =\;\;
        \begin{cases}
        \; 0& \quad \mbox{if}\qquad y_0 \le x_0\,,
        \\[0.5em]
          \;\DBW\left(\Ginv_{\min}, \Finvbench\right)& \quad \mbox{if}\qquad y_0 > x_0\,,
        \end{cases}
    \end{equation*}
    where 
    \begin{equation}\label{eq:G-min}
        \Ginv_{\min}(u):=(\phi')^{-1}\bigg(\Big[\phi'\big(\Finvbench(u)\big) - \tfrac{\lambda_{\min}}{1-\alpha} \; \xi(u)\Big]^\uparrow\bigg)\in \mcM\,,
    \end{equation}
and $\lambda_{\min}>0$ is such that $\int_0^1 \Ginv_{\min}(u)\xi (u)\d u=x_0$. 
\end{proposition}

\begin{proof}
By \Cref{prop:budget-como}, the budget of the benchmark portfolio is $y_0 = \int_0^1 \Finvbench(u) \xi (u)\d u$.
Next note that the $\alpha$-BW divergence $\DBW\left(F_Z, \Fbench\right)$ is zero if and only if $\Finv_Z \equiv \Finvbench$. Thus, $\ep_{\min}=0$ if and only if the benchmark satisfies the budget constraint, i.e., $y_0\le x_0$.

Next we calculate $\ep_{\min}$ for the case when $\ep_{\min}>0$. For $\beta \in (0,1)$ consider a related Lagrangian with Lagrange parameter $\lambda \ge0$
\begin{align*}
    \mathcal{L}^{(\beta)}(\Ginv, \lambda)
    :&=
    \int_0^1 \beta \, B_\phi\big(\Ginv(u),\Finvbench(u)\big)+\lambda \, \Ginv(u)\, \xi(u)\, du
    \\
    &=
    \beta\int_0^1  \, \phi(\Ginv(u)) 
    - \Big(\phi'(\Finvbench(u))     -\frac{\lambda}{\beta} \, \xi(u)\Big) \Ginv(u)\, du
    \\
    & \quad 
    - \int_0^1 \phi(\Finvbench(u))-\phi'(\Finvbench(u))\Finvbench(u) \, du
    \,.   
\end{align*}
The solution to optimisation problem \eqref{opt:ep-min} is equivalent to the saddle point of the Lagrangian \citep{Rockafellar1970}
\begin{equation*}
    \sup_{\lambda \ge 0}\;  \inf_{\Ginv \in \mcM}\; \mL^{(\beta)}(\Ginv, \lambda)\,.
\end{equation*}
Thus, for fixed $\lambda \ge 0$ and using \Cref{thm:generalised-iso} the minimum over $\Ginv \in \mcM$ of $\mathcal{L}^{(\beta)}(\Ginv, \lambda)$ is attained at
\begin{equation*}
    G^{(\beta)}_\lambda(u):=(\phi')^{-1}\bigg(\Big[\phi'\big(\Finvbench(u)\big) - \tfrac{\lambda}{\beta}\;  \xi(u)\Big]^\uparrow\bigg)\,,\quad  u\in(0,1)\,.
\end{equation*}
Next to solve the optimisation problem \eqref{opt:ep-min}, consider the pointwise in $u$ Lagrangian 
\begin{equation*}
   \mathcal{L}_\lambda(u,x)
   :=
   \mathcal{L}_\lambda^{(\alpha)}(u,x)\Id_{x\ge \Finvbench(u)} + \mathcal{L}_\lambda^{(1-\alpha)}(u,x)\Id_{x< \Finvbench(u)}\,,\quad u\in(0,1)\,,
\end{equation*}
where
\begin{equation*}
    \mathcal{L}_\lambda^{(\beta)}(u,x)
    :=    \beta \, B_\phi\big(x,\Finvbench(u)\big)+\lambda \, x\, \xi(u)\,, \quad \text{for} \qquad \beta \in (0,1)\,.
\end{equation*}
By \Cref{prop:isotonic-properties} $ii)$ and since $\xi \ge 0, \lambda\ge0, \beta\ge0$, it holds that $G^{(\beta)}_\lambda(u) \le \Finvbench(u)$ for all $u \in (0,1)$ and $\beta \in (0,1)$. Thus, for each $u \in (0,1)$, the infimum of the Lagrangian $ \L_\lambda(u,\cdot)$ is attained in the region $\{ x \le \Finvbench(u)\}$ and 
\begin{equation*}
    \arg\min_{x\in\mathbb R}\L_\lambda(u,x)=\breve G_\lambda^{(1-\alpha)}(u)\,, \quad u \in (0,1)\, .
\end{equation*}
By definition of the isotonic projection, $ G^{(1-\alpha)}_\lambda(\cdot)$ is non-decreasing implying that $G^{(1-\alpha)}_\lambda \in \mcM$.

Next, we show the existence of a Lagrange multiplier that attains the budget constraint. Observe that the function $\cost(\lambda):=\int_0^1 \Ginv_\lambda^{(1-\alpha)}(u)\xi (u)\d u$ is continuous and decreasing on $(0,+\infty)$, by \Cref{prop:isotonic-properties} $ii)$ and $iii)$. Moreover, $\lim_{\lambda\to 0}\cost(\lambda)= y_0>x_0$  and 
$\lim_{\lambda\to +\infty}\cost(\lambda)\le 0$, since $\lim_{x\to -\infty} (\phi')^{-1}(x)\le 0$. Thus, there exists a $\lambda_{\min}$ that solves $\cost(\lambda_{\min})=x_0$. Uniqueness follows by strictly convexity of the $\alpha$-BW divergence.
\end{proof}

If the function $\xi$ is non-increasing, such as in the GMB market  model, see \Cref{ex:gbm-xi}, then $\Ginv_{\min}$ simplifies to $\Ginv_{\min}(u) = (\phi')^{-1}\Big(\phi'\big(\Finvbench(u)\big) - \tfrac{\lambda_{\min}}{1-\alpha} \; \xi(u)\Big)$.

To establish well-posedness of the investor's optimisation problem, that is the existence of a quantile function that satisfies all constraints and whose objective function of \eqref{opt:main-prime} is finite valued, we require the following additional assumption.

\begin{assumption}\label{asm:e_min}
The minimal tolerance $\ep_{\min}$ is finite, i.e. $\ep_{\min}<+\infty$.
\end{assumption}

\begin{proposition}[Well-posedness]
    Let Assumptions \ref{asm:sdf-cont}, \ref{asm:bench-cont}, and \ref{asm:e_min} be satisfied. If $\ep\ge \ep_{\min}$, then optimisation problem \eqref{opt:main-prime} is well-posed.
\end{proposition}

\begin{proof}
    We have to show that there exists a quantile function satisfying all constraints and which has a finite objective function. 
    First, we consider the case when $\ep_{\min} = 0$. This means that by \Cref{Prop:epsilon-bound} the benchmark satisfies the budget constraint, i.e., $y_0 \le x_0$. Moreover, the objective function evaluated at the benchmark is
    \begin{equation*}
     \int_0^1 U\big((1-c)\Finvbench(u)\big)\d u
     \le \int_0^1 U\big(\Finvbench(u)\big)\d u < +\infty\,,
    \end{equation*}
    where the last inequality follows by \Cref{asm:bench-cont}. Thus, the benchmark satisfies all constraint and has a finite objective function.

    Second, we consider the case $\ep_{\min} \in(0,+\infty)$. Then the quantile function $\Ginv_{\min}$ given in \eqref{eq:G-min} satisfies the budget and $\alpha$-BW divergence constraints. Moreover, as $\xi(u) \ge 0$, for all $u \in (0,1)$, we have by \Cref{prop:isotonic-properties} $ii)$, that $\Ginv_{\min}(u) \le  (\phi')^{-1}\Big(\big[\phi'\big(\Finvbench(u)\big)\big]^\uparrow\Big)= \Finvbench(u)  $ for all $u \in (0,1)$.
    Therefore, 
    \begin{align*}
        \int_0^1 U\big(\Ginv_{\min}(u)-c\Finvbench(u)\big)\,  \d u
        &\le
        \int_0^1 U\big(\Finvbench(u)\big) \, \d u< +\infty\,,
    \end{align*}
    which concludes the proof.
\end{proof}

\section{Boundary cases}\label{sec:boundary-cases}
In this section we first establish the solution to the boundary cases, that is the cases when only one of the constraints is satisfied. In \Cref{sub:ep-infinite} we study the solution to optimisation problem \eqref{opt:main-prime}, when the tolerance distance $\ep$ is large enough such the $\alpha$-BW divergence is not binding. And in \Cref{sub:infinite-x0} we consider the case when the investor has infinite wealth, that is when the budget constraint is absent and only the $\alpha$-BW divergence constraint is binding.

\subsection{Infinite tolerance distance}
\label{sub:ep-infinite}

The first result pertains to the case when the $\alpha$-BW divergence constraint is absent. Note that when the $\alpha$-BW divergence constraint is absent, optimisation problem \eqref{opt:main-prime} reduces to 
\begin{equation} \label{opt:ep-infty}
    \max_{\Ginv \in \mcM}  \; \int_0^1 U\Big(\Ginv(u) - c \Finvbench(u)\Big)\, \d u 
    \qquad \text{subject to} \qquad
    \quad \int_0^1 \Ginv(u) \xi (u)\d u\le x_0\,.
    \tag{\text{$\breve{P}_{\ep = \infty}$}}
\end{equation}
As the utility of a negative payoff is $- \infty$, the investor only considers portfolios whose terminal wealth $\P$-a.s. dominate $cY$. Thus, the investor's budget $x_0$ needs to be at least that of $c y_0$. Thus, for the next results we require that $c \le \min\{\frac{x_0}{y_0}, 1\}$.

We proceed by using Lagrangian methods and use the superscript ``$\cost$'' -- as only the cost of the budget is enforced -- for the optimal quantile function $\Ginvcost$ and the Lagrange multiplier $\lamcost$ of optimisation problem \eqref{opt:ep-infty}.

\begin{theorem}[Infinite tolerance]\label{thm:ep-infty}
Let Assumptions \ref{asm:sdf-cont} and \ref{asm:bench-cont} be satisfied, and $ c\le \min\{\frac{x_0}{y_0}, 1\}$. Then there exists a unique solution to \eqref{opt:ep-infty}. Moreover, the following holds
\begin{enumerate}[label = $\roman*)$]
    \item if $x_0 = c y_0$, then $\Ginvcost(u) = c \Finvbench(u)$,
\item if $x_0 > c y_0$, 
then
    \begin{equation}\label{eq:Ginv-ep-infty}
        \Ginvcost(u)= (U')^{-1}\left(\lamcost \, \xi(u)^\downarrow\right)+ c \Finvbench(u)\,,
    \end{equation}
    where $\lamcost>0$ is the unique solution to $\int_0^1  \Ginvcost(u) \xi(u) \, du = x_0$. We further denote by $\epcost:= \DBW\big(\Gcost, \Fbench\big)$ its $\DBW$-divergence to the benchmark.
    \end{enumerate}

\end{theorem}

\begin{proof}
We first prove case $ii)$ and note that optimisation problem \eqref{opt:ep-infty} is equivalent to 
\begin{equation}\label{eq:opt-proof-ep-infty}
    \max_{\substack{X \\ (X,Y) \text{ comonotone}}} \E\big[U\big(X - cY\big)\big] 
     \qquad \text{subject to} \qquad
     \E[(X - cY)\varsigma] \le x_0 - cy_0\,.
\end{equation}
Next consider the following alternative optimisation problem 
\begin{equation}\label{eq:opt-alternative}
    \max_{Z} \E\big[U(Z)\big] 
     \qquad \text{subject to} \qquad
     \int_0^1 \Finv_Z(x) \xi(u) \d u  \le x_0 - cy_0\,.
\end{equation}
If in \eqref{eq:opt-alternative} we additionally impose the constraint that $(Z+ cY, Y)$ is comonotone, then the two optimisation problems are equivalent. Next we solve optimisation problem \eqref{eq:opt-alternative}. For this, note that its solution only depends on the distribution of $Z$. That is any random variable whose quantile function is a solution to
\begin{equation}\label{eq:opt-ep-infty-related}
    \max_{\Ginv  \in \mcM}  \; \int_0^1 U\big(\Ginv(u)\big)\, \d u \,,
    \quad \text{subject to} \quad
    \quad \int_0^1 \Ginv(u) \xi (u)\d u\le x_0 - c y_0\,,
\end{equation}
is a solution to \eqref{eq:opt-alternative}.

The associated Lagrangian to \eqref{eq:opt-ep-infty-related} with Lagrange parameter $\lambda \ge 0$ is
\begin{equation*}
    \mL (\Ginv, \lambda) = \int_0^1 -U\big(\Ginv(u)\big) + \lambda \xi(u) \Ginv(u)\, du - \lambda(x_0 - c y_0)
\end{equation*}
and the solution of optimisation problem \eqref{eq:opt-ep-infty-related} is equivalent to the saddle point of the Lagrangian \citep{Rockafellar1970}
\begin{equation}\label{eq:lagrange-iso}
    \sup_{\lambda \ge 0}\;  \inf_{\Ginv \in \mcM}\; \mL(\Ginv, \lambda)\,.
\end{equation}
For fixed $\lambda \ge 0$, the inner problem of \eqref{eq:lagrange-iso} has by \Cref{thm:generalised-iso} the unique solution 
\begin{equation*}
    \Ginv_\lambda(u) = (-U')^{-1}\left(\big[- \lambda \xi(u)\big]^\uparrow\right)\,.
\end{equation*}
Since $(-U')^{-1}(\cdot) = (U')^{-1}(-\, \cdot)$ and by \Cref{prop:isotonic-properties}  $i)$, we have $[- \lambda \xi(u)\big]^\uparrow = - \lambda \xi(u)^\downarrow$, and obtain
\begin{equation*}
    \Ginv_\lambda(u) = (U')^{-1}\left(\lambda \, \xi(u)^\downarrow\right)\,.
\end{equation*}
Moreover, since $(U')^{-1}(\cdot)$ is non-increasing and $\lambda \, \xi(u)^\downarrow>0 $ for all $u \in (0,1)$ and $\lamcost>0$,  we have that $\Ginv_\lambda \in \mcM$. 

To prove the existence of the Lagrange multiplier note that $ \lim_{\lambda\to 0}\int_0^1  \Ginv_\lambda(u) \xi(u) \, du = + \infty$. Moreover, by strict concavity of the utility, $\int_0^1\Ginv_\lambda(u) \xi(u) \, du$ is as a function of $\lambda$ strictly decreasing, and $\lim_{\lambda\to +\infty} \int_0^1 \Ginv_\lambda(u) \xi(u) \, du = 0$. Thus there exists a unique $\lamcost\ge0$ such that $\int_0^1  \Ginv_{\lamcost}(u) \xi(u) \, du = x_0$. 

Thus, $\Ginv_{\lamcost}$ is the unique solution to \eqref{eq:opt-ep-infty-related} (uniqueness follows by strict concavity of the utility) and any random variable with quantile function $\Ginv_{\lamcost}$ is a solution to \eqref{eq:opt-alternative}. Next, we choose $Z^M:= \Ginv_{\lamcost}(\Finvbench(Y))$ and define $X^M:= Z^M + cY$. As $Z^M$ is a solution to \eqref{eq:opt-alternative} and $(Z^M + cY, Y) = (X^M, Y)$ is comonotone, it holds that $X^M$ is also a solution to optimisation problem \eqref{eq:opt-proof-ep-infty}. Noting that the quantile function of $X^M$ is given by $\Ginvcost(u) := \Ginv_{\lamcost}(u) + c \Finvbench(u)$ concludes the proof. Uniqueness follows by strict concavity of the utility.

For case $i)$, note that the budget is zero, thus the limiting case of the Lagrangian applies and the optimal Lagrange multiplier is equal to infinity. By Inada's conditions the optimal quantile function simplifies to $\Ginvcost(u) = c\, \Finvbench$. 
\end{proof}

\begin{corollary}\label{cor:sol-P-inv-ep-large}
Let Assumptions \ref{asm:sdf-cont}, \ref{asm:bench-cont}, and \ref{asm:e_min} be satisfied, $c \le\min\{\frac{x_0}{y_0}, 1\}$, and $\ep \ge \epcost$. Then, there exists a unique solution to \eqref{opt:main-prime} given by $\Ginvcost$ defined in \eqref{eq:Ginv-ep-infty}.
\end{corollary}

\begin{proof}
By \Cref{thm:ep-infty}, $\Ginvcost$ solves optimisation problem \eqref{opt:main-prime} without the $\alpha$-BW constraint and satisfies $\DBW(\Gcost, \Fbench) = \epcost$. Thus $\Ginvcost$ is admissible for optimisation problem \eqref{opt:main-prime} with $\ep \ge \epcost$. As \eqref{opt:main-prime} has an additional constraint compared to the infinite tolerance distance setting, $\Ginvcost$ needs to be optimal for \eqref{opt:main-prime}. Existence and uniqueness follow from \Cref{thm:ep-infty}.
\end{proof}

As seen above, if the tolerance distance is large enough, i.e., $\ep \ge \epcost$, then the solution to the optimisation problem \eqref{opt:main-prime} is given by \Cref{thm:ep-infty}.

\subsection{Infinite wealth solution}\label{sub:infinite-x0}

Here we consider the case where the investor chooses any admissible portfolio in an $\ep$-$\DBW$ ball to maximise their expected utility but without a budget constraint. We find that for this infinite wealth problem, the investor picks an admissible portfolio that dominates, in first order stochastic dominance, the reference portfolio $Y$. Due to the concavity of the utility and convexity of the Bregman divergence, larger deviation from $Y$ in the $\DBW$ divergence lead to larger expected utility. Therefore, optimality is attained on the boundary of the $\DBW$ ball. The following theorem confirms this intuition and provides the corresponding optimal quantile function. 

For $\epsilon>0$, consider the problem of maximal expected utility within an $\alpha$-$\BW$ ball but with infinite budget, i.e.,
\begin{equation*}\label{opt:x-infty}
    \max_{\Ginv \in \mcM}  \; \int_0^1 U\big(\Ginv(u) - c \Finvbench(u)\big)\, \d u 
    \qquad \text{subject to} \qquad
    \quad \DBW\left(\Ginv, \Finvbench\right)  \le \ep\,.\tag{$\breve{P}_{x_0 = \infty}$}
\end{equation*}
The next statement characterises the optimal quantile function to the infinite wealth problem \eqref{opt:x-infty}. Similarly to \Cref{sub:ep-infinite}, we use the superscript ``$\BW$'' for the optimal quantile function $\Ginvbw$ and Lagrange multiplier $\lambw$ of optimisation problem \eqref{opt:x-infty}.

\begin{theorem}[Infinite wealth solution]
\label{thm:infinite-wealth}
Let Assumptions \ref{asm:sdf-cont} and \ref{asm:bench-cont} be satisfied, $c \in [0,1]$, and $\epsilon>0$. Then there exists a unique solution to optimisation problem \eqref{opt:x-infty} given by 
\begin{equation}\label{eq:G-infinite-wealth}
    \Ginvbw(u)=H^{-1}\Big(-\alpha \lambw\phi'\big(\Finvbench(u)\big)\Big)\,,
\end{equation} 
where $H(x) := U'\big(x - c \Finvbench(u)\big) - \lambw \,\alpha \,\phi'(x)$ and where $\lambw>0$ is the unique solution to
$\DBW(\Ginvbw,\Finvbench)=\epsilon$. We further denote its cost by $x_0^\infty:= \int_0^1 \Ginvbw(u)\xi(u) \d u$.
\end{theorem}

\begin{proof}
We start with the associated Lagrangian to \eqref{opt:x-infty}. Note that since we minimise (for fixed Lagrange multiplier) the Lagrangian over quantile functions $\Ginv$, we only write the Lagrangian for the cases when it is not equal to $+\infty$ (that is only for $\Ginv(u) \ge c\Finvbench(u)$, as otherwise the utility is $-\infty$). Thus, the pointwise (in $u$) Lagrangian of \eqref{opt:x-infty} (wherever it is finite) is
\begin{align*}
    \L_\lambda(u, x)
    &= 
     -U\big(x-c\Finvbench(u)\big) \Id_{x \ge c\Finvbench(u)}
    + 
    \lambda (1-\alpha) \, B_\phi(x, \Finvbench(u))\, \Id_{x\in [c\Finvbench(u),  \Finvbench(u))}
    + \lambda \alpha \, B_\phi(x, \Finvbench(u))\, \Id_{x \ge \Finvbench(u)} 
    \\
    &= 
     \L_\lambda^{(\alpha)}(u, x) \Id_{x \ge \Finvbench(u)}
    + 
    \L_\lambda^{(1-\alpha)}(u, x) \Id_{x \in [c\Finvbench(u),  \Finvbench(u))}\,,
\end{align*}
where 
\begin{align*}
    \L_\lambda^{(\beta)}(u, x) 
    &= 
    -U\big(x-c\Finvbench(u)\big)
    + 
    \lambda \beta \, B_\phi(x, \Finvbench(u))\,.
\end{align*}
To find the global minimiser (in $x$) of $\L_\lambda(u,x)$, we first examine the minimiser of $\L_\lambda^{(\beta)}(u, x)$ for all $\beta \in (0,1)$. 
For $\beta \in (0,1)$, we have that (using pointwise optimisation)
\begin{equation}\label{eq:ginv-lambda}
    \argmin_{x \in \R}\L_\lambda^{(\beta)}(u, x)
     = \Ginv^{(\beta)}_\lambda(u)
    :=
    H^{-1}\Big(- \lambda\, \beta\,  \phi'\big(\Finvbench(u)\big)\Big)\,,\quad  u\in[0,1]\,,
\end{equation}
where $H(x) := U'\big(x - c \Finvbench(u)\big) - \lambda \,\beta \,\phi'(x)$. Note that for each fixed $u \in (0,1)$, the inverse of $H$ in $x$ exists, as $H$ is strictly increasing in $x$. Next, we show that for each $\beta \in (0,1)$, $\Ginv^{(\beta)}_\lambda\in \mcM$ and that $\Ginv^{(\beta)}_\lambda(u) \ge \Finv(u)$ for all $u \in (0,1)$. Thus, for each $u \in (0,1)$, the minimiser of $\L_\lambda(u, \cdot)$ is attained in the region $\{x \ge \Finvbench(u)\}$, which implies that the global minimum of $\L_\lambda(u,\cdot)$ is attained at $\Ginv^{(\alpha)}_\lambda(u)$, $u \in (0,1)$.

We first show that $\Ginv_\lambda^{(\beta)}(\cdot)$ is non-decreasing. Note that $H$ is a function of both $x$ and $u$, and for the purpose of showing that $\Ginv_\lambda^{(\beta)}(\cdot)$ is non-decreasing, we make the dependence on $u$ explicit and write $H_u(x)$. For fixed $u \in (0,1)$ the function $x \mapsto H_u(x)$ is non-increasing, and for fixed $x \ge 0$, the function $u \mapsto H_u(x)$ is non-decreasing. This implies that $x \mapsto H^{-1}_u(x)$ is non-increasing and $u \mapsto H^{-1}_u(x)$ is non-decreasing. Further, for $\delta >0$ such that $u + \delta <1$, we have
\begin{align*}
    \Ginv_\lambda^{(\beta)}(u + \delta)
    &= 
    H^{-1}_{u+\delta}\Big(- \lambda\, \beta\,  \phi'\big(\Finvbench(u + \delta)\big)\Big)
    \\
    &\ge
    H^{-1}_{u+\delta}\Big(- \lambda\, \beta\,  \phi'\big(\Finvbench(u)\big)\Big)
    \\
    &\ge
    H^{-1}_u\Big(- \lambda\, \beta\,  \phi'\big(\Finvbench(u)\big)\Big)
    \\
    &=
    \Ginv^{(\beta)}_\lambda(u)\,,
\end{align*}
and $\Ginv^{(\beta)}_\lambda$ is non-decreasing and $\Ginv^\beta_\lambda\in \mcM$.

We next show that for all $\beta \in (0,1)$ it holds that $\Ginv^{(\beta)}_\lambda(u) \ge \Finvbench(u)$ for all $u \in (0,1)$. For $u \in (0,1)$, the inequality $\Ginv^{(\beta)}_\lambda(u) \ge \Finv(u)$ is equivalent to 
\begin{equation*}
    - \lambda\, \beta\,  \phi'\big(\Finvbench(u)\big)
    \le
    H\big(\Finvbench(u)\big)
    =
    U'\big((1 - c) \Finvbench(u)\big) - \lambda \,\beta \,\phi'\big(\Finvbench(u)\big)\,,
\end{equation*}
which is equivalent to $U'\big((1 - c) \Finvbench(u)\big) \ge0$, which, by strict increasingness of $U$ on $(0, +\infty)$, is true for all $u \in (0,1)$.  Thus, $\Ginv_\lambda^{(\beta)}(u)\ge \Finvbench(u)$, $u \in (0,1)$, and we conclude that indeed $\Ginv_\lambda^{(\beta)}$ is the global minimiser of $\int_0^1\L_\lambda(u, x)\d u$.

Finally we show that the $\alpha$-BW constraint is binding, note that $\lim_{\lambda \to 0}\DBW(\Ginv_\lambda^{(\alpha)}, \Finvbench) = + \infty $ and that $\lim_{\lambda \to +\infty}\Ginv_\lambda^{(\alpha)}(u) = \Finvbench(u)$, which implies $\lim_{\lambda \to +\infty}\DBW(\Ginv_\lambda^{(\alpha)}, \Finvbench) = 0$. To see that $\lim_{\lambda \to +\infty}\Ginv_\lambda^{(\alpha)}(u) = \Finvbench(u)$, note that for each $u\in (0,1)$, $\Ginv_\lambda^{(\alpha)}(u)$ satisfies the following equation in $x$
\begin{equation}\label{eq:Ginv-lambda-limit}
    \phi'(x) - \phi'\big(\Finvbench(u)\big)
    =
    \frac{1}{\lambda \alpha}U'\big(x-c \Finvbench(u)\big) \,.
\end{equation}
Taking the limit for $\lambda \to + \infty$ in \eqref{eq:Ginv-lambda-limit}, we obtain by strict convexity of $\phi(\cdot)$ that $\lim_{\lambda \to + \infty}\Ginv_\lambda(u) = \Finvbench(u)$, for all $u \in (0,1)$. 

Uniqueness follows by strict concavity of the utility function.
\end{proof}

If the infinite wealth solution with quantile function $\Ginvbw$ satisfies the budget constraint, i.e. if $x_0^\infty \le x_0$, then $\Ginvbw$ solves Problem \eqref{opt:main-prime}. We summaries this in the next corollary. 

\begin{corollary}\label{cor:sol-xo-infty}
Let Assumptions \ref{asm:sdf-cont}, \ref{asm:bench-cont}, and \ref{asm:e_min} be satisfied, $c \in [0,1]$, and $x_0 \ge x^\infty$. Then, there exists a unique solution to \eqref{opt:main-prime} given by $\Ginvbw$ defined in \eqref{eq:G-infinite-wealth}.
\end{corollary}

\begin{table}[t]
    \centering
    \begin{tabular}{c c c c c }
     $\ep$    &  $x_0$ & budget & $\alpha$-BW & solution \\
     \toprule \toprule
     $\ep < \ep_{\min}$    & $x_0>0$ & -- & -- & no solution
     \\
     \midrule
     \multirow{2}{*}{$\ep_{\min}\le \ep < \epcost$}
     & $ 0 < x_0 < x_0^\infty$ & binding & binding & \Cref{Thm:main-optimal-quantile}
     \\
    & $x_0 \ge x_0^\infty$  & not binding & binding & \Cref{cor:sol-xo-infty}
    \\
    \midrule
    $\ep \ge \epcost$ & $x_0 > 0$ & binding & not binding & \Cref{cor:sol-P-inv-ep-large}
    \end{tabular}
    \caption{Solution to \eqref{opt:main-prime} for different choice of tolerance distance $\ep$ and budget $x_0$. The budget and $\alpha$-BW column indicate whether the respective constraints are binding. Recall that $\ep_{\min}$ and $\epcost$ depend on $x_0$, and that $x_0^\infty$ depends on $\ep$. }
    \label{tab:Lagrangme-mulitplier}
\end{table}
\Cref{tab:Lagrangme-mulitplier} collects the combinations of tolerance distance $\ep$ and budget $x_0$, and states whether there exists a solution to optimisation problem \eqref{opt:main-prime} and which of the constraints are binding. Recall that $\ep_{\min}$ depends on the $x_0$, implying that there is no solution if either the $\DBW$-ball or the budget it too small. Further, if $x_0$ is chosen to satisfy $\int_0^1 \Finvbench(u) \Finvsig(1-u)\, \d u \le x_0$, that is the benchmark satisfies the budget, then $\ep_{\min} = 0$.

\section{Optimal solution for binding constraints}

In this section we prove the representation of an optimal quantile function of optimisation problem \eqref{opt:main-prime} when both constraints are binding (\Cref{sec:optimal-quantile}) and we show that under additional assumptions a unique solution to optimisation problem \eqref{opt:main-prime} exists (\Cref{sec:existence}).

For this we first establish the conditions when both constraints are binding. 
From the last section we observe that for both constraints to be binding, it is necessary that $\ep_{\min} \le \ep < \epcost$ and $0< x_0 < x_0^\infty$, see also \Cref{tab:Lagrangme-mulitplier} for an overview of the different cases. 
The next result shows that a converse holds in the following sense. If the tolerance distance $\ep$ and the budget $x_0$ are suitably chosen and an optimal solution to optimisation problem \eqref{opt:main-prime} exists, then it needs to bind both constraints. 

\begin{proposition}[Binding constraints]\label{prop:constraints-binding}
Let $c\le \min\{\frac{x_0}{y_0}, 1\}$, $\ep_{\min} \le \ep< \epcost$, $0<x_0 < x_0^\infty$, and Assumptions \ref{asm:sdf-cont}, \ref{asm:bench-cont}, and \ref{asm:e_min} be satisfied. If a solution to optimisation problem \eqref{opt:main-prime} exists, then both constraints are binding. 
\end{proposition}

\begin{proof}
Recall that by \Cref{prop:one-const-binding} at least one of the constraints is binding. Thus, we split the proof into two cases and proceed by contradiction. First, we assume that only the budget constraint is binding and in a second step, we assume that only the $\alpha$-BW divergence is binding. 

\underline{Part 1:} Budget is binding. 
Let $\Ginv$ denote the solution to optimisation problem \eqref{opt:main-prime}, satisfying $\cost(\Ginv) = x_0$ and $\DBW(\Ginv, \Finvbench) = \ep^\dagger< \ep$. Next, recall $\Ginvcost$ from \Cref{cor:sol-P-inv-ep-large}. $\Ginvcost$ satisfies $\cost(\Ginvcost) = x_0$ and $\DBW(\Ginv, \Finvbench) = \epcost$. Moreover, as $\Ginvcost$ is the solution to \eqref{opt:ep-infty}, we have $\int_0^1 U(\Ginvcost(u)) \d u \ge \int_0^1 U(\Ginv(u)) \d u $, as \eqref{opt:ep-infty} has less constraints than \eqref{opt:main-prime}.

Next, for $w \in [0,1]$, consider the family of quantile functions $\Ginv^w(u):= w \Ginv(u) +(1- w) \Ginvcost(u)$. Clearly $\cost(\Ginv^w) = x_0$ and by strict concavity of the utility we have
\begin{equation*}
    \int_0^1U\big(\Ginv^w(u)\big)\d u
    > w\int_0^1U\big(\Ginv(u)\big)\d u + (1-w) \int_0^1U\big(\Ginvcost(u)\big)\d u
    \ge \int_0^1U\big(\Ginv(u)\big)\d u\,.
\end{equation*}
Moreover, by  convexity of the $\alpha$-BW divergence in its first component, it holds that 
\begin{equation*}
    \DBW\big(\Ginv^w, \Finvbench) \le
    w \DBW\big(\Ginv, \Finvbench)  + (1 - w) \DBW\big(\Ginvcost, \Finvbench) 
    =w \ep^\dagger + (1-w) \ep^\infty. 
\end{equation*}
Thus, there exists a $w^\dagger\in (0,1)$ such that $\DBW\big(\Ginv^{w^\dagger}, \Finvbench) \le \ep$. As $\Ginv^{w^\dagger} $ is feasible for optimisation problem \eqref{opt:main-prime} and has a strictly larger utility, we arrive at a contradiction.

\underline{Part 2:} $\alpha$-BW divergence is binding. 
Let $\Ginv$ denote the solution to optimisation problem \eqref{opt:main-prime} satisfying $\cost(\Ginv)  = x^\dagger < x_0$ and $\DBW(\Ginv, \Finvbench) = \ep$. Next, recall $\Ginvbw$ from \Cref{cor:sol-xo-infty}. $\Ginvbw$ satisfies $\cost(\Ginvbw) = \xbw$ and $\DBW(\Ginvbw, \Finvbench) = \ep$. Moreover, as $\Ginvbw$ is the solution to \eqref{opt:x-infty}, we have $\int_0^1 U(\Ginvbw(u)) \d u \ge \int_0^1 U(\Ginv(u)) \d u $, as \eqref{opt:x-infty} has one less constraints than \eqref{opt:main-prime}.

Next, for $w \in [0,1]$, define the family of quantile functions $\Ginv^w(u):= w \Ginv(u) +(1- w) \Ginvbw(u)$. Using similar arguments as in part 1, we have $\cost(\Ginv^w) = w x^\dagger + (1-w)\xbw$, $\int_0^1 U\big(\Ginv^w(u)\big)\d u> \int_0^1 U\big(\Ginv(u)\big)\d u$, and $\DBW(\Ginv^w, \Finvbench) \le \ep$. Thus there exists a $w^\dagger \in (0,1)$, such that $\cost(\Ginv^{w^\dagger}) = x_0$, which implies that $\Ginv^{w^\dagger}$ is a feasible to \eqref{opt:main-prime} and has a strictly larger utility than $\Ginv$, providing a contradiction. 

Combining, if a solution exists, both constraints have to be binding. 
\end{proof}

\subsection{Optimal quantile function}\label{sec:optimal-quantile}

Before stating the representation of the optimal quantile function, we define a family of quantile functions that are instrumental for the solution of optimisation problem \eqref{opt:main-prime}. For $\bfeta \in [0, \infty)^2$ and $\beta \in (0,1)$ define the quantile functions 
\begin{equation}\label{eq:def-opt-sol}
    \Ginv^{(\beta)}_{\bfeta}(u)
    :=
    \tilde{H}_{\beta}^{-1}\Big(\; \Big[ \eta_2 \, \beta\, \phi'\big(\Finvbench(u)\big) - \eta_1 \xi(u) \Big]^\uparrow\Big)\,,
\end{equation}%
where $\tilde{H}_{\beta}(x) := - U'\big(x - c\Finvbench(u)\big) + \eta_2\, \beta\, \phi'\big(x\big)$. For simplicity, we omit the dependence of $\tilde{H}_\beta$ on $u$. 

The next result characterises the optimal quantile function of optimisation problem \eqref{opt:main-prime}, if a solution exists. We observe that the $\alpha$-BW divergence creates two cases for the optimal quantile function. One for underperforming and another for outperforming the benchmark. We note that the optimal quantile function depends on the utility and the constant $c$ though the function $\tilde{H}_\beta$.

\begin{theorem}[Representation of the optimal quantile function]\label{Thm:main-optimal-quantile}
Let $c\le \min\{\frac{x_0}{y_0}, 1\}$, $\ep_{\min} \le \ep< \epcost$, $0<x_0 < x_0^\infty$, and Assumptions \ref{asm:sdf-cont}, \ref{asm:bench-cont}, and \ref{asm:e_min} be satisfied. Define $\alpha^\dagger := \alpha\Id_{\{\alpha> \frac12\}}  + (1-\alpha)\Id_{\{\alpha \le \frac12\}}$. If a solution to optimisation problem \eqref{opt:main-prime} exists, then it is unique and given by 
\begin{equation*}
    \Ginv_{\bfeta^*}(u) :=
    \begin{cases}
    \Ginv^{(\alpha)}_{\bfeta^*}(u) & \text{if} \qquad  \Finvbench(u) \le \Ginv^{(\alpha^\dagger)}_{\bfeta^*}(u)\,,
    \\[0.5em]
     \Ginv^{(1-\alpha)}_{\bfeta^*}(u) & \text{if} \qquad   \Finvbench(u) > \Ginv^{(\alpha^\dagger)}_{\bfeta^*}(u)
    \,,
\end{cases}
\end{equation*}
where $\bfeta^*\in (0, \infty)^2$ are such that the budget and the $\alpha$-BW constraints are binding, that is $\newline\int_0^1 \Ginv_{\bfeta^*}(u) \xi(u) \d u= x_0$ and $\DBW \big(\Ginv_{\bfeta^*}, \Finvbench) = \ep$. Moreover, $\Ginv_{\bfeta^*}(\cdot)$ is continuous.
\end{theorem}

\begin{proof}
We first prove the case $\alpha \ge \frac12$ by driving the Lagrangian associated to problem \eqref{opt:main-prime}. Since we minimise the Lagrangian (for fixed Lagrange parameter) over the set of quantile function, we only consider the Lagrangian for the cases when it not equal to $+\infty$, that is for $\Ginv\in \mcM$ satisfying $\Ginv(u) \ge c\Finvbench(u)$, $u \in (0,1)$. Thus, the Lagrangian associated to problem \eqref{opt:main-prime} with Lagrange parameter $\bfeta:= (\eta_1, \eta_2)\ge 0$ is 
\begin{align}\label{eq:proof-lagrangian}
    \L(\Ginv, \bfeta )
    :&= 
    \int_0^1 \L_{\bfeta}^{(1-\alpha)} (\Ginv, u ) \Id_{\Ginv(u) \in[ c\Finvbench(u), \, \Finvbench(u))}\, \d u
    +    
    \int_0^1 \L_{\bfeta}^{(\alpha)} (\Ginv, u ) \Id_{\Ginv(u) > \Finvbench(u)}\, \d u
\,,
\end{align}
where for all $\beta \in (0,1)$
\begin{align*}
    \L_{\bfeta}^{(\beta)} (\Ginv, u ):&= - U\Big(\Ginv(u)-c \Finvbench(u)\Big) + \eta_1 \Ginv(u) \xi(u) 
    + \eta_2 \, \beta \, B_\phi \big(\Ginv(u), \, \Finvbench(u)\big)
    \,.
\end{align*} 
We want to apply \Cref{lemma:auxiliary-2-main-proof} with $h = \L_{\bfeta}^{(1-\alpha)}$ and $\ell = \L_{\bfeta}^{(\alpha)}$. Thus we first check that the assumptions of \Cref{lemma:auxiliary-2-main-proof} are satisfied. 

Note that $\L_{\bfeta}^{(\beta)}$, $\beta \in (0,1]$ are strictly convex in their first component and they satisfy 
\begin{enumerate}[label = $\roman*)$]
    \item $\L_{\bfeta}^{(\alpha)}(\Ginv, u) \ge \L_{\bfeta}^{(1-\alpha)}(\Ginv, u)$ for all $\Ginv \in \mcM$ and $u \in (0,1)$, 
    
    \item 
    $\L_{\bfeta}^{(\alpha)}(\Finvbench, u) = \L_{\bfeta}^{(1-\alpha)}(\Finvbench, u)$ for all $u\in(0,1)$,
\end{enumerate}
and clearly $c\Finvbench(u) \le \Finvbench(u)$ for all $u \in (0,1)$. Thus assumptions $i)$ and $ii)$ of \Cref{lemma:auxiliary-2-main-proof} are satisfied for $h = \L_{\bfeta}^{(1-\alpha)}$, $\ell = \L_{\bfeta}^{(\alpha)}$, and $\Ginv_1 = \Finvbench$. 

Next we define the integrals corresponding to $ \bar{h}$ and $\bar{\ell}$ in \Cref{lemma:auxiliary-2-main-proof}. That is for $\beta \in (0,1)$, we set
\begin{align*}
    \bar{\L}_{\bfeta}^{(\beta)} (\Ginv)
    :=
    \int_0^1 \L_{\bfeta}^{(\beta)} (\Ginv, u ) \, du.
\end{align*}
For fixed $\bfeta$, we calculate the minimisers of $\bar{\L}_{\bfeta}^{(\beta)}$
\begin{align}\label{eq:proof-g-beta-argmin}
\begin{split}
    \argmin_{\Ginv \in \mcM} \bar{\L}_{\bfeta}^{(\beta)}
    =
    \argmin_{\Ginv \in \mcM} 
    \int_0^1 &- U\Big(\Ginv(u)-c \Finvbench(u)\Big) + \eta_1 \Ginv(u) \xi(u) 
    \\
    & \qquad + \eta_2 \, \beta\,  \Big(\phi\big(\Ginv(u)\big) - \phi'\big(\Finvbench(u)\big) \Ginv(u)\Big)\, du\,.
\end{split}
\end{align}
As by definition of the utility it holds that $U(x) = - \infty$ whenever $x<0$, the solution to \eqref{eq:proof-g-beta-argmin} which we denote by $\Ginv^{(\beta)}$ must satisfy $\Ginv^{(\beta)}(u) \ge c\Finvbench(u)$, $u \in (0,1)$. Applying \Cref{thm:generalised-iso} to characterise $\Ginv^{(\beta)}$, we obtain
\begin{equation*}
    \Ginv^{(\beta)}
    :=
    \argmin_{\Ginv\in\mcM} \bar{\L}_{\bfeta}^{(\beta)}
    =\tilde{H}^{-1}\Big(\; \Big[ \eta_2 \, \beta\, \phi'\big(\Finvbench(u)\big) - \eta_1 \xi(u) \Big]^\uparrow\Big)\,,
\end{equation*}
where $\tilde{H}(x) := - U'\big(x - c\Finvbench(u)\big) + \eta_2\, \beta\, \phi'\big(x\big)$. 

As we have that $G^{(\beta)}(u) \ge c \Finvbench(u)$, all assumption of \Cref{lemma:auxiliary-2-main-proof} are satisfied with $\Ginv_0 = c \Finvbench$, $\Ginv_1 = \Finvbench$, $x^*_h = \Ginv^{(1-\alpha)}$, and $x^*_\ell = \Ginv^{(\alpha)}$. Thus applying \Cref{lemma:auxiliary-2-main-proof}, the minimiser of $ \L(\Ginv, \bfeta )$, defined in \eqref{eq:proof-lagrangian}, has representation 
\begin{equation*}
    \Ginv_{\bfeta}(u) =
    \begin{cases}
    \Ginv^{(\alpha)}_{\bfeta}(u) & \text{if} \qquad  \Finvbench(u) \le \min\big\{\Ginv^{(\alpha)}_{\bfeta}(u), \Ginv^{(1-\alpha)}_{\bfeta}(u)\big\}
    \\[0.5em]
     \Ginv^{(1-\alpha)}_{\bfeta}(u) & \text{if} \qquad   \Finvbench(u) > \max\big\{\Ginv^{(\alpha)}_{\bfeta^*}(u), \Ginv^{(1-\alpha)}_{\bfeta}(u)\big\}    \,.
\end{cases}
\end{equation*}
As by \Cref{asm:bench-cont}, the benchmark quantile is continuous, we obtain, again by \Cref{lemma:auxiliary-2-main-proof}, that $\Ginv_{\bfeta}\in \mcM$ and that $\Ginv_{\bfeta}(\cdot)$ is continuous.

Next, we simplify the representation of $\Ginv_\bfeta$. For this we apply \Cref{lemma:Ginv-monotonicity-beta}, and see that $\Ginv_\bfeta^{(\beta)}$ is strictly decreasing in $\beta$ on the set $\{u \in (0,1)~|~ \Finvbench(u)  < \Ginv_\bfeta^{(\beta)}\}$. Thus, since $\alpha> 1- \alpha$, we have that 
\begin{equation*}
     \Finvbench(u)\le \min\big\{\Ginv^{(\alpha)}_{\bfeta}(u), \Ginv^{(1-\alpha)}_{\bfeta}(u)\big\} = \Ginv_\bfeta^{\alpha}(u)\,.
\end{equation*}
Similarly, again by \Cref{lemma:Ginv-monotonicity-beta} it holds that $\Ginv_\bfeta^{(\beta)}$ is strictly increasing in $\beta$ on the set $\{u \in (0,1)~|~ \Finvbench(u)  > \Ginv_\bfeta^{(\beta)}\}$. Thus, 
\begin{equation*}
     \Finvbench(u)> \max\big\{\Ginv^{(\alpha)}_{\bfeta}(u), \Ginv^{(1-\alpha)}_{\bfeta}(u)\big\} = \Ginv_\bfeta^{\alpha}(u)\,.
\end{equation*}
Thus, the representation of $\Ginv_\bfeta$ simplifies to the one in the statement.

Finally, by \Cref{prop:constraints-binding} the constraints are binding, thus the optimal Lagrange multipliers $\bfeta^*$ have to be chosen to bind the constraints. Uniqueness follows by \Cref{prop:one-const-binding}.

The case $\alpha < \frac12$ follows using similar steps, by applying \Cref{lemma:auxiliary-2-main-proof-alpha-small} instead of \Cref{lemma:auxiliary-2-main-proof} with $\ell = \L_{\bfeta}^{(1-\alpha)}$ and $h = \L_{\bfeta}^{(\alpha)}$.
\end{proof}

Next, we provide three simplified examples. The first pertains to the BW divergence ($\alpha = \frac12$), the second is the symmetric and popular (squared) 2-Wasserstein distance ($\alpha = \frac12$ and $\phi(x) = x^2$), and third if in addition $c = 0$. We observe that for $\alpha = \frac12$, the optimal quantile function simplifies and the investor is indifferent between under- and outperformance. However, the investor still penalises (absolute) gains and losses differently through the choice of the Bregman generator $\phi$.

\begin{corollary}[BW divergence and Wasserstein distance]
Let $c\le \min\{\frac{x_0}{y_0}, 1\}$, $\ep_{\min} \le \ep< \epcost$, $0<x_0 < x_0^\infty$, and Assumptions \ref{asm:sdf-cont}, \ref{asm:bench-cont}, and \ref{asm:e_min} be satisfied. If a solution to optimisation problem \eqref{opt:main-prime} exists, then it is unique and the following hold: 
\begin{enumerate}[label = $\roman*)$]
    \item If $\alpha = \frac12$, then the solution to optimisation problem \eqref{opt:main-prime} is given by $ \Ginv^{\frac12}_{\bfeta^*}(u) $ defined in \eqref{eq:def-opt-sol}.
    
    \item 
    if $\alpha = \frac12$ and $\phi(x) = x^2$, which implies that the $\alpha$-BW equals to one-half the squared Wasserstein distance. Then the solution to optimisation problem \eqref{opt:main-prime} is given by
\begin{align*}
    \Ginv^{W}_{\bfeta^*}(u)
    :=
    \hat{H}^{-1}\Big(\; \Big[ \eta_2^* \, \Finvbench(u) - \eta_1^* \xi(u) \Big]^\uparrow\Big)\,,
\end{align*}
with $\hat{H}(x) := - U'\big(x - c\Finvbench(u)\big) +  \eta_2^*\,x$. 

\item If $c = 0$, that is if the investor only considers the utility to the optimal strategy's wealth (and not the difference to a proportion of the benchmark's), then case $ii)$ simplifies to 
\begin{equation*}
    \Ginv^{\frac12,0}_{\bfeta^*}(u)
   : =
    \hat{H}_{0}^{-1}\Big(\; \Big[ \eta_2^* \, \Finvbench(u) - \eta_1^* \xi(u) \Big]^\uparrow\Big).
\end{equation*}
where $\hat{H}_{0}(x) := - U'\big(x \big) +  \eta_2^*\,x$. 

\end{enumerate}
\end{corollary}

We compare our solution to the works of \cite{Pesenti2024WP}, who consider an investor maximising expected utility subject to a budget and BW constraint. Here, if $c = 0$ and $\alpha = \frac12$, then the agents considers only the utility of their own strategy and the BW divergence, however, because of the comonotonicity constraint, our results are different from those in \cite{Pesenti2024WP}. Specifically, without the comonotonicity constraint, the optimal terminal wealth is countermonotone to the stochastic discount factor. Thus the optimal quantile function of \cite{Pesenti2024WP} arises from \eqref{eq:def-opt-sol} with $c = 0$, $\alpha = \frac12$, $\ep$ rescaled to $2\ep$, and when $\xi(\cdot)$ is replaced by $\Finvsig(1 - \cdot)$. As $\Finvsig$ is a quantile function the isotonic projection is not required in \cite{Pesenti2024WP}.

From the above corollary we observe, that the $\alpha$-BW divergence with $\alpha\neq\frac12$ gives raise to treating outperformance and underperformance relative to the benchmark differently. Indeed if $\alpha = \frac12$, then by case $i)$, we observe that there is no distinction between outperforming or underperforming the benchmark. However, absolute gains are still treated in an asymmetric way through the Bregman generator $\phi$. If additionally $\phi(x) = x^2$, in which case the $\alpha$-BW divergence equals one-half the squared Wasserstein distance, that is case $ii)$ in the above corollary, then gains and losses are treated symmetrically.

The next results state that the optimal quantile function is continuous with respect to convergence of the Lagrange multipliers and we recover the boundary cases as limiting cases. Specifically, case $i)$ states that the solution to optimisation problem \eqref{opt:main-prime} converges (for $\eta_1$ going to 0) pointwise to the solution of optimisation problem \eqref{opt:x-infty}. Moreover, case $ii)$ states that the solution to optimisation problem \eqref{opt:main-prime} converges (for $\eta_2$ going to 0) pointwise to the solution of optimisation problem \eqref{opt:ep-infty}.

\begin{proposition}
    Let $c\le \min\{\frac{x_0}{y_0}, 1\}$, $\ep \ge \ep_{\min}$ and let Assumptions \ref{asm:sdf-cont}, \ref{asm:bench-cont}, and \ref{asm:e_min} be satisfied and let a solution to optimisation problem \eqref{opt:main-prime} exists. Then, the following holds

    \begin{enumerate}[label = $\roman*)$]
        \item $\displaystyle\lim_{\eta_1 \searrow 0} \Ginv_{\bfeta^*}(u) =\Ginvbw(u) $, for all $u \in (0,1)$, where $\Ginvbw$ is given in \eqref{eq:G-infinite-wealth}. 

        \item $\displaystyle\lim_{\eta_2 \searrow 0} \Ginv_{\bfeta^*}(u) =\Ginvcost(u) $, for all $u \in (0,1)$, where $\Ginvcost$ is given in \eqref{eq:Ginv-ep-infty}. 
    \end{enumerate}
\end{proposition}

\begin{proof}
Case $i)$ follows by evaluating $\Ginv_{\bfeta}^{(\beta)}$ at $\eta_1 = 0$, which gives 
$\Ginv^{(\beta)}_{(0, \eta_2)}(u) = \Ginv^{(\beta)}_{\lambda}$, where $\Ginv^{(\beta)}_{\lambda}$ is defined in \eqref{eq:ginv-lambda}. Moreover, from the proof of \Cref{thm:infinite-wealth} it follows that $\Ginv^{(\beta)}_{\lambda}(u) \ge\Finvbench(u)$, for all $u \in (0,1)$ and $\beta \in (0,1)$, thus only the first case of $\Ginv_{\bfeta^*}$ applies.

Case $ii)$ follows by evaluating $\Ginv_{\bfeta}^{(\alpha)}$ at $\eta_2 = 0$ and noting that it is independent of $\alpha$. 
\end{proof}

\subsection{Existence of a solution}\label{sec:existence}

This subsection provides existence of a solution to optimisation problem \eqref{opt:main-prime}, when both constraints are binding. For this we require additional assumptions. One is that $\xi$ is non-increasing and thus the isotonic projection is not needed. 

\begin{assumption}\label{asm:xi-non-increasing}
    The function $\xi$ is non-increasing.
\end{assumption}

With this assumption at hand, we recall the representation of the optimal quantile function  $\Ginv_\bfeta$ as defined in \Cref{Thm:main-optimal-quantile} but for general $\bfeta \in [0, \infty)^2$, that is
\begin{equation}\label{eq:Ginv-eta}
    \Ginv_{\bfeta}(u) :=
    \begin{cases}
    \Ginv^{(\alpha)}_{\bfeta}(u) & \text{if} \qquad  \Finvbench(u) \le \Ginv^{(\alpha^\dagger)}_{\bfeta}(u)\,,
    \\[0.5em]
     \Ginv^{(1-\alpha)}_{\bfeta}(u) & \text{if} \qquad   \Finvbench(u) > \Ginv^{(\alpha^\dagger)}_{\bfeta}(u)
    \,,
\end{cases}
\end{equation}
 where 
\begin{equation}\label{eq:Ginv-beta}
    \Ginv^{(\beta)}_{\bfeta}(u)
    :=
    \tilde{H}_{\beta}^{-1}\Big(\; \Big[ \eta_2 \, \beta\, \phi'\big(\Finvbench(u)\big) - \eta_1 \xi(u) \Big]\Big)\,, \quad \text{for} \quad \beta \in (0,1)\,
\end{equation}
with $\tilde{H}_{\beta}$ defined in \eqref{eq:def-opt-sol}. By \Cref{asm:xi-non-increasing}, $\Ginv^{(\beta)}_{\bfeta}$ given in \eqref{eq:Ginv-beta} is already non-decreasing, thus the isotonic projection is not needed. 

To prove existence of a solution to optimisation problem \eqref{opt:main-prime} when both constraints are binding, we show that the optimal Lagrange multipliers are strictly positive. For ease of notation, we define the constraints of $\Ginv_\bfeta$ as functions of the Lagrange multipliers, that is
\begin{align*}
\cost(\bfeta)&:=  \int_0^1 \Ginv_{\bfeta }(u) \Finvsig(1-u)\d u
\quad \text{and} 
    \\[0.5em]
   \bw(\bfeta)&:=\; \DBW\left(\Ginv_{\bfeta}, \;\Finvbench\right)\,. 
\end{align*}

For existence we furthermore require integrability, in particular that the constraints and the expected utility of the candidate solutions are finite. From the proof of \Cref{Thm:main-optimal-quantile} we see that $\eta_1$ is the Lagrange multiplier for the budget constraint and $\eta_2$ for the $\alpha$-BW constraint.

\begin{assumption}[Integrability]\label{asm:integrability}
Let $c\le \min\{\frac{x_0}{y_0}, 1\}$. Then for all $\bfeta \in(0,\lamcost)\times(0,\lambw]$, it holds that $\cost(\bfeta)< +\infty$, $\bw(\bfeta) < +\infty$, and 
\begin{equation*}
    \int_0^1 U\big(\Ginv_\bfeta(u) - c\Finvbench(u)\big)\, du < +\infty\,,
\end{equation*}
where $\Ginv_\bfeta$ is given in \eqref{eq:Ginv-eta}.
\end{assumption}

Due to the $\alpha$-BW divergence, the constraint functions $\bw$ and $\cost$ are not monotone in both arguments. Thus, we further require that the Lagrange multipliers $(\eta_1,\lambw)$, for all $\eta_1 \in (0, \lamcost)$, is either not a feasible candidate or its $\alpha$-BW divergence is equal to $\ep_{\min}$. Moreover, we need that the Lagrange multipliers $(\lamcost,\eta_2)$, for all $\eta_2 \in (0, \lambw]$, does not bind the budget constraint.

\begin{assumption}\label{asm:admissibility}
It holds that $\bw(\eta_1, \lambw) \le \ep_{\min}$ for all $\eta_1 \in (0, \lamcost)$, and that $\cost(\lamcost, \eta_2)< x_0$ for all $\eta_2 \in (0,\lambw]$.
\end{assumption}

\Cref{asm:admissibility} is weaker than requiring that $\cost$ and $\bw$ are strictly decreasing in both arguments. Simple examples when \Cref{asm:admissibility} holds are $c=1$ and $\alpha=1/2$ (in which case the constraint functions become monotone). In \Cref{sec:ex}, we numerically verify that all assumptions are satisfied.

Next, we define the smallest value that $\eta_1$ can attain 
\begin{equation}
    \etaOneMin:= \sup\big\{ \eta_1 >0~|~ \cost(\eta_1, \lambw)\ge x_0\big\}\, 
    \label{eq:eta1-min}
\end{equation}
and we are ready to state the result on existence.
\begin{theorem}\label{thm:existence}
Let $c\le \min\{\frac{x_0}{y_0}, 1\}$, $\ep_{\min} \le \ep< \epcost$, and $0<x_0 < x_0^\infty$. Moreover let Assumptions \ref{asm:sdf-cont}, \ref{asm:bench-cont}, \ref{asm:e_min}, \ref{asm:xi-non-increasing}, \ref{asm:integrability}, and \ref{asm:admissibility} be fulfilled. 

Then there exists a unique solution to optimisation problem \eqref{opt:main-prime}. Moreover, both constraints, the budget and the $\alpha$-BW divergence, are binding and the optimal Lagrange multipliers satisfy $(\eta_1^*, \eta_2^*) \in \big[\etaOneMin, \lamcost\big) \times \big(0, \lambw\big]$. 
\end{theorem}

The proof of Theorem \ref{thm:existence} is split into a sequence of technical lemmas. First, we show that the constraints are continuous and monotone (in one of their components). Second, we derive the range of the Lagrange multipliers and third, we show the existence of a fixed point.

\begin{lemma}\label{lemma:c-b-decreasing}
Let the assumptions in \Cref{thm:existence} be satisfied. Then the functions $\cost(\bfeta)$ and $\bw (\bfeta)$ are continuous in both $\eta_1$ and $\eta_2$. Moreover,
\begin{enumerate}[label = $\roman*)$]
    \item For fixed $\eta_2>0$, $\cost(\bfeta)$ is strictly decreasing in $\eta_1$.
    
    \item For fixed $\eta_1>0$, $\bw (\bfeta)$ is strictly decreasing in $\eta_2$.

\end{enumerate}
\end{lemma}
\begin{proof}
Recall that the SDF $\varsigma$ and the benchmark quantile $\Finvbench$ are continuous. Thus $\Ginv^{(\beta)}_{\bfeta }$ defined in \eqref{eq:Ginv-beta} is continuous in $\bfeta$ for all $\beta \in (0,1)$. This implies that both $\cost(\bfeta)$ and $\bw (\bfeta)$ are continuous in $\eta_1$ and $\eta_2$. 

For $i)$ let $\eta_2>0$ be fixed. As the budget $\cost$ is linear in $\Ginv_\bfeta$ and as by \Cref{Le:monoG}  $\Ginv_\bfeta$ is strictly decreasing in $\eta_1$, we obtain that $\cost$ is also strictly decreasing in $\eta_1$. 

For case $ii)$ let $\eta_1>0$ be fixed. Moreover, note that, for fixed $y$, the Bregman divergence $B_\phi(x,y)$ is strictly increasing in $x$ on the interval $[y, +\infty$), and strictly decreasing in $x$  on the interval $[- \infty, y)$. Moreover, by \Cref{Le:monoG} we have that $\Ginv_{\bfeta}(u)$ is strictly decreasing in $\eta_2$ on $\{u\in[0,1]~|~ \Ginv_{\bfeta}(u)\ge \Finvbench(u)\}$ and strictly increasing on $\{u\in[0,1]~|~ c\Finvbench(u)< \Ginv_{\bfeta}(u)<\Finvbench(u)\}$. Combining these two facts implies that $\bw$ is decreasing in $\eta_2$.
\end{proof}

Next we recall the boundary cases discussed in \Cref{sec:boundary-cases} and that $\eta_1$ is the Lagrange multiplier for the budget and $\eta_2$ the Lagrange multiplier for the $\alpha$-BW constraint. For the case when $\ep \ge\epcost$, that is the $\alpha$-BW constraint is not binding and optimisation problem \eqref{opt:main-prime} reduces to \eqref{opt:ep-infty}, it must hold that $\eta_2 =  0$. Moreover, by \Cref{cor:sol-P-inv-ep-large}, we have that $\bw(\lamcost, 0) = \epcost$ and $(\lamcost, 0)$ are the optimal Lagrange multiplier for optimisation problem \eqref{opt:main-prime} with $\ep \ge\epcost$. Furthermore, for the case when $x_0>x_0^\infty$, that is when the budget constraint is not binding and optimisation problem \eqref{opt:main-prime} reduces to \eqref{opt:x-infty}, it must hold that $\eta_1 = 0$. Then by \Cref{cor:sol-xo-infty} we have that $\cost(0, \lambw) = \xbw$ and $(0, \lambw)$ are the optimal Lagrange multipliers for optimisation problem \eqref{opt:main-prime} with $x_0>x_0^\infty$. Moreover, since by \Cref{prop:constraints-binding} both constraints are binding, the Lagrange multipliers for optimisation problem \eqref{opt:main-prime} must be strictly positive.

Next, for each $\eta_1\in[ \etaOneMin, \lamcost)$, we define
\begin{subequations}\label{eq:eta-min-max}
    \begin{align}
    \eta_2^{\min}(\eta_1):&= \sup\big\{ \eta_2 > 0~|~ \bw(\eta_1, \eta_2)\ge \epcost\big\}
     \quad \text{and}
    \label{eq:eta1-min}
    \\[0.5em]
    \eta_2^{\max}(\eta_1):&= \inf\; \big\{ \eta_2 >0~|~ \bw(\eta_1, \eta_2)\le \ep_{\min}\big\}\,,
    \label{eq:eta2-max}
    \end{align}
where we use the convention that $\sup \emptyset := 0$ and $\inf\emptyset := + \infty$. 
\end{subequations}
Thus, for each $\eta_1 \in [ \etaOneMin, \lamcost)$, the range of $\eta_2$ as a function of $\eta_1$ is given by $[\eta_2^{\min}(\eta_1), \eta_2^{\min}(\eta_1)]$.

\begin{lemma}\label{lemma:multiplier-range}
Let the assumptions in \Cref{thm:existence} be satisfied. Then it holds that $0< \etaOneMin \le \lamcost$ and $\etaTwoMin(\eta_1) < \etaTwoMax(\eta_1)$, for all $\eta_1 \in [\etaOneMin, \lamcost)$.
\end{lemma}

\begin{proof}
First note that $\cost(0, \lambw) = \xbw$ and that by definition of $\etaOneMin$ we have $x_0 = \cost(\etaOneMin, \lambw) < \cost(0, \lambw)$. Thus, as $\cost$ is strictly decreasing in its first argument (see \Cref{lemma:c-b-decreasing}), we obtain that $\etaOneMin>0$. Next we show that $\etaOneMin \le\lamcost$. To see this we proceed by contradiction and let $\etaOneMin >\lamcost$. By definition of $\etaOneMin$, as $\cost$ is strictly decreasing in its first argument, and using \Cref{asm:admissibility}, we obtain the contradiction $\cost(\etaOneMin,\lambw)=x_0<\cost(\lamcost,\lambw)<x_0$. Thus $\etaOneMin \le \lamcost$.

Finally, let $\eta_1 \in [\etaOneMin, \lamcost)$, then by definition $ \bw(\eta_1, \etaTwoMin(\eta_1)) = \epcost > \ep_{\min} = \bw(\eta_1, \etaTwoMax(\eta_1))$. Since $\bw$ is strictly decreasing in its second argument (see \Cref{lemma:c-b-decreasing}), we have $\etaTwoMin(\eta_1) < \etaTwoMax(\eta_1)$. 
\end{proof}

Next, we derive the existence of implicit functions such that for each $\eta_1 \in [\etaOneMin, \lamcost)$ the constraints are binding. 

\begin{lemma}\label{lemma:inplicit-fct}
Let the assumptions in \Cref{thm:existence} be satisfied. Then, there exist continuous functions $k,\ell\colon [ \etaOneMin, \lamcost)\mapsto [0, +\infty)$ such that $\cost\big(\eta_1, k(\eta_1)\big)=x_0$ and $\bw\big(\eta_1,\ell(\eta_1)\big)=\epsilon$, for all $\eta_1\in [ \etaOneMin, \lamcost)$. Moreover, $k$ is given by 
\begin{equation}\label{eq:k-implicit}
     k(\eta_1): = \sup\big\{ y \in [0, \lambw]~|~ \cost(\eta_1, y) = x_0\big\}\,, 
     \quad \text{for all} \quad 
     \eta_1 \in [\etaOneMin, \lamcost)\,.
\end{equation}
\end{lemma}

\begin{proof}
By Lemma \ref{lemma:c-b-decreasing}, for fixed $\eta_1> 0$, the functions $\cost(\eta_1,\cdot)$ and $\bw(\eta_1,\cdot)$ are continuous. Next, we observe that for 
$\eta_1\in [\etaOneMin, \lamcost)$, we have by strict decreasingness of $\cost$ in its first argument (see \Cref{lemma:c-b-decreasing}), that
$\cost(\eta_1,0)>  \cost(\lamcost,0)=x_0$, where the equality follows by \Cref{cor:sol-P-inv-ep-large}. Moreover, again by strict decreasingness of $\cost$, we have that $\cost(\eta_1,\lambw) < \cost(\etaOneMin,\lambw) =x_0$, where the equality follows by definition of $\etaOneMin$. Therefore, 
by the Intermediate Value Theorem, there exists $\eta_2\in(0,\lambw)$ that satisfies $\cost(\eta_1,\eta_2)=x_0$. As $\eta_2$ might not be unique, we define
\begin{equation*}
    k(\eta_1): = \sup\big\{ y \in [0, \lambw]~|~ \cost(\eta_1, y) = x_0\big\},
\end{equation*}
which satisfies $\cost(\eta_1, k(\eta_1)) = x_0$ for all $\eta_1 \in [\etaOneMin, \lamcost)$. 

To see the existence of $\ell$, note that by \eqref{eq:eta-min-max}, it holds that $\bw(\eta_1,\etaTwoMax(\eta_1)) = \ep_{\min}< \epcost =\bw(\eta_1,\eta_2^{\min}(\eta_1))$ for any $\eta_1\in [\etaOneMin, \lamcost)$.
Thus, for  $\eta_1\in [ \etaOneMin, \lamcost)$  there exists $\ell(\eta_1)\in [\eta_2^{\min}(\eta_1),\eta_2^{\max}(\eta_1))$ such that $\bw(\eta_1, \ell(\eta_1))=\epsilon$. Moreover, by strict decreasingness of $\bw$ in its second argument, $\ell$ is unique. 
\end{proof}

The next lemma shows that the two functions $k$ and $\ell$ cross. 

\begin{lemma}\label{lemma:crossing}
Let the assumptions in \Cref{thm:existence} be satisfied. Then for $k$ and $\ell$ defined in \Cref{lemma:inplicit-fct} the following holds 
\begin{enumerate}[label = $\roman*)$]
    \item
    $
    k(\etaOneMin)=\lambw$ and $k(\lamcost)=  0$,
    
    \item $\ell(\etaOneMin) \le \lambw$ and $\ell(\lamcost) > 0$.
\end{enumerate} 
\end{lemma}

\begin{proof}
For $i)$, we have by definition of $\etaOneMin$, that $x_0 = \cost(\etaOneMin, \lambw) = \cost(\etaOneMin, k(\etaOneMin))$, which implies that $k(\etaOneMin) \ge \lambw$. 
Moreover, by definition of $k$, see \eqref{eq:k-implicit}, we have $k(\etaOneMin) \le \lambw$, thus $k(\etaOneMin) = \lambw$. 
For the second part of $i)$, note that by \Cref{cor:sol-P-inv-ep-large} $x_0 = \cost(\lamcost, 0) = \cost\big(\lamcost, k(\lamcost)\big)$. Moreover, by \Cref{asm:admissibility}, it holds that for all $\eta_2 \in (0, \lambw)$ that $\cost(\lamcost, \eta_2)<x_0$, which implies that $k(\lamcost)$ is unique and thus $k(\lamcost) = 0$.

For $ii)$ note that from the proof of \Cref{lemma:inplicit-fct}, we know that $\ell(\etaOneMin)\le  \etaTwoMax(\etaOneMin)$. Next we show that $\etaTwoMax(\etaOneMin) \le \lambw$. To see this, we proceed by contradiction and assume that $\lambw < \etaTwoMax(\etaOneMin)$. Then by definition of $\etaTwoMax$ and as $\bw$ is strictly decreasing in its second argument, we obtain $\ep_{\min} = \bw(\etaOneMin, \etaTwoMax(\etaOneMin))< \bw(\etaOneMin, \lambw) \le \ep_{\min}$, where the last inequality follows by \Cref{asm:admissibility}. Thus, we arrive at a contradiction and indeed $\etaTwoMax(\etaOneMin) \le \lambw$, which implies that $\ell(\etaOneMin) \le \lambw$. For the second part of $ii)$, i.e., $\ell(\lamcost)>0$, note that $\bw(\lamcost, 0)= \epcost >  \ep = \bw(\lamcost, \ell (\lamcost))$, and as $\bw$ is strictly decreasing in its second argument (see \Cref{lemma:c-b-decreasing}), we obtain that $\ell (\lamcost) > 0$.
\end{proof}

We are ready to prove \Cref{thm:existence}, that is we show a fixed point of the functions $k$ and $\ell$.

\begin{proof}[Proof of \Cref{thm:existence}]
    By \Cref{lemma:inplicit-fct}, the functions $k, \ell$ are continuous and satisfy
\begin{equation*}
    \cost\big(\eta_1, k(\eta_1)\big)=x_0
    \quad \text{and}
    \quad
    \bw\big(\eta_1,\ell(\eta_1)\big)=\epsilon\,,
\end{equation*}
for all $\eta_1 \in [\etaOneMin, \lamcost)$. 
Next, by \Cref{lemma:crossing}, it holds that $k(\etaOneMin) = \lambw\ge \ell(\etaOneMin)$ and $k(\lamcost) = 0 < \ell(\lamcost)$.
Thus, by continuity of $k, \ell$ and $\bw, \cost$, there exists a crossing point, i.e., $\bfeta^\dagger: = (\eta_1^\dagger, \eta_2^\dagger)$ such that $k(\eta_1^\dagger) = \ell(\eta_1^\dagger)$ and moreover, $\bw(\eta^\dagger_1,k(\eta_1^\dagger))  = \ep$ and $\cost(\eta^\dagger_1,k(\eta_1^\dagger)) = x_0$. Thus $(\eta_1^\dagger, k(\eta_1^\dagger)) $ is the optimal Lagrange multiplier. Note that $\eta_1^\dagger \in [\etaOneMin, \lamcost)$ and that $\eta_2^\dagger \in (0, \lambw]$, where the latter follows by definition of $k$. 

By \Cref{prop:constraints-binding}, both constraints are binding, that is the Lagrange multipliers are strictly positive. Thus, we showed the existence of Lagrange multipliers that are strictly positive, therefore there exists a solution to optimisation problem \eqref{opt:main-prime} and both constraints are binding. 
\end{proof}

\section{Numerical illustration}\label{sec:ex}

Here we work in a GBM market model and continue the setup of \Cref{ex:gbm-xi}. For the numerical implementation we choose the following market parameters $R = 1$, $\Gamma = 2$, $\Psi = 0.8$ and set the benchmark's initial wealth to $y_0 = 1$. The investor has a CRRA utility function given by 
\begin{equation*}
    U(x) = \frac{x^{1-\gamma}}{1- \gamma}, \quad \text{for} \quad 
    \gamma \in  (0,1) \,.
\end{equation*}
For the $\alpha$-BW divergence, the investor considers the asymmetric Bregman divergence generated by the Power family, that is for $p>1$ 
\begin{equation*}
 \phi_p(x) = \frac{2 x^p}{p(p-1)} \, , \quad x>0    \,,
\end{equation*}
which is illustrated in \Cref{ex: GBM-MM}. For $p = 2$, we obtain one-half the symmetric squared Wasserstein distance. In the examples, the parameters are chosen such that all assumptions for the existence of a solution to \cref{opt:main-prime} are satisfied. 

\begin{figure}[h]
    \centering
    \includegraphics[width=0.48\textwidth]{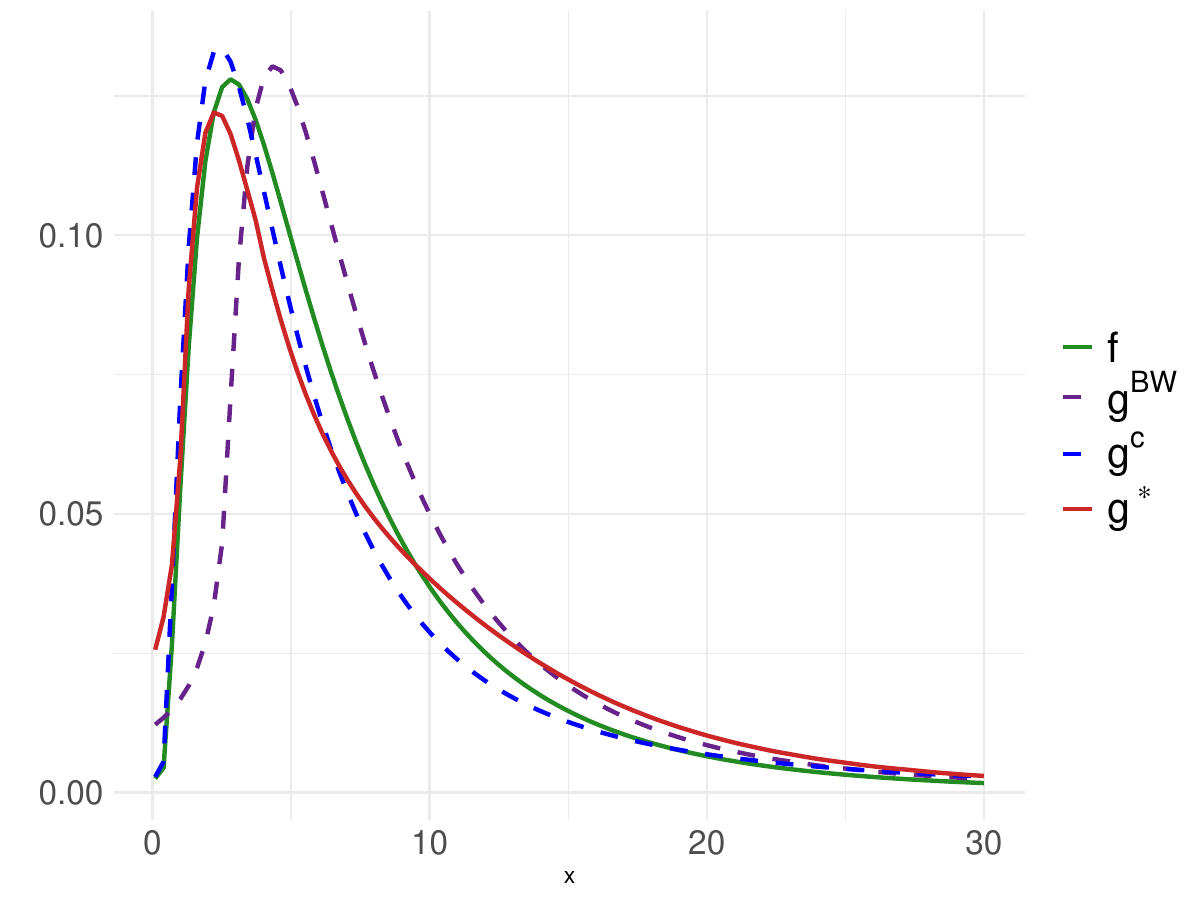}
    \includegraphics[width=0.48\textwidth]{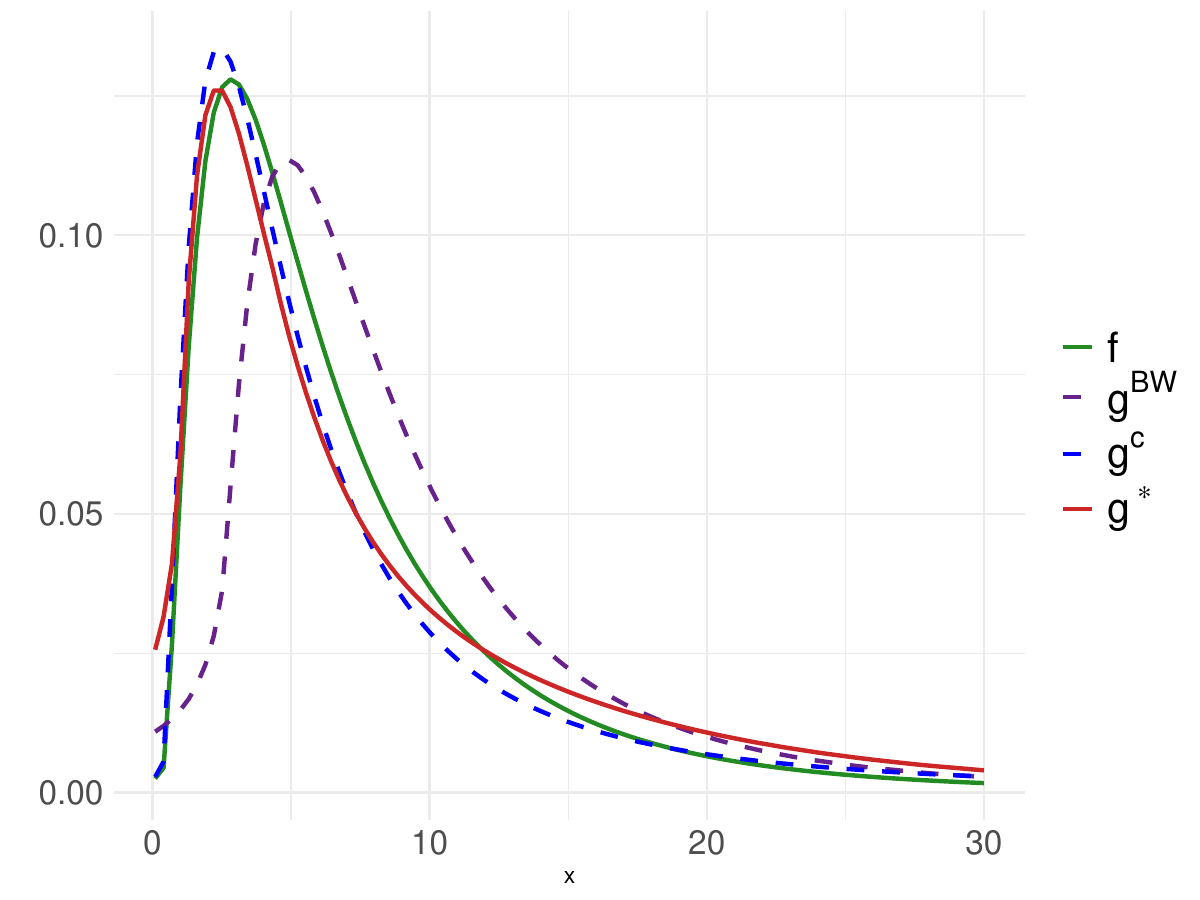}

    \caption{Optimal density functions for the Power Bregman generator $\phi_p$. Left: $p  = 2$ (one-half the squared Wasserstein distance). Right: $p = 1.6$,. Common parameters $\alpha = 0.25$, $x_0 = 1$, $c = 0.9$, $\gamma = \frac12$, and $\ep = \frac12$.}
    \label{fig:opt-density}
\end{figure}

\Cref{fig:opt-density} displays the density functions corresponding to $\Ginv_{\bfeta^*}$ to \eqref{opt:main-prime} with both constraint binding (which in the plots we simply denote as $g^{*}$). It also displays  $g^\cost$ the density corresponding to $\Ginvcost$ the solution to \eqref{opt:ep-infty} -- the optimisation problem with only the cost constraint -- and the density $g^{\BW}$ corresponding to $\Ginvbw$ the solution to \eqref{opt:x-infty} -- the optimisation problem with only the $\alpha$-BW divergence -- as well as the benchmark density $f$. The left panel displays the one-half squared Wasserstein distance ($p = 2$) and the right panel the Power Bregman generator with parameter $p = 1.6$. Recall from the discussion of \Cref{ex: GBM-MM}, that $p = 1.6<2$ and $\alpha = 0.25$ illustrates an investor who wishes that the left tail of their terminal wealth is close to that of the benchmark -- a feature we clearly observe in \Cref{fig:opt-density} -- and penalises underperformance relative to the benchmark more. Indeed, the left tails of the densities of the benchmark $f$ in green and the density of the optimal terminal wealth $g^*$ are close to indistinguishable. However, the right tail of $g^*$ is significantly heavier tailed indicating larger gains than the benchmark. Comparing $g^\cost$ and $g^\DBW$ with the benchmark $f$, we observe that the $g^\DBW$ (which only binds the divergence constraint) is shifted to the right, however, the strategy's cost of 1.463 for $p = 2$ (1.55 for $p =1.6$) is outside the budget, see  also \Cref{tab:opt-quant}. The density $g^\cost$ (only binding the budget constraint) on the other hand has a lighter right tail and deviates further from the benchmark; realising an $\alpha$-BW divergence of 43.55 for $p = 2$ and 9.12 for $p = 1.6$. 
\begin{table}[ht]
    \centering
    \begin{tabular}{c c c c c }
     $\phi_p(\cdot)$  &  constraints & $\Ginv^\DBW$ & $\Ginv^c$  &  $\Ginv_{\bfeta^*}$\\
         \toprule \toprule 
    \multirow{3}{*}{$p = 2$}   &    $\DBW(\cdot, \Finvbench)$ 
            & 0.5  & 43.55 & 0.5          \\
        & $\cost(\cdot)$  & 1.463 & 1.0  & 1.0\\
        & $\E[U(\cdot - c \Finvbench)]$  & 2.918 & 2.713  & 2.146 \\[0.5em]
        \midrule
        \multirow{3}{*}{$p = 1.6$}  
      &  $\DBW(\cdot, \Finvbench)$ 
            & 0.5  & 9.12  & 0.5          \\
       &  $\cost(\cdot)$  & 1.550 & 1.0  & 1.0\\
       & $\E[U(\cdot- c \Finvbench)]$  & 3.310 & 2.713 & 2.278 \\
        \bottomrule \bottomrule
    \end{tabular}
    \caption{Constraints for different quantile function for the same parameters as in \Cref{fig:opt-density}. The benchmarks expected utility is 1.587, its cost is 1, and $\DBW(\Finvbench, \Finvbench) = 0$.}
    \label{tab:opt-quant}
\end{table}
\Cref{tab:opt-quant} collects the $\alpha$-BW divergence, the cost, and the expected utility of each strategy displayed in \Cref{fig:opt-density}. We observe as expected that $\Ginvcost$ has a significantly larger $\alpha$-BW divergence than $\Ginv_{\bfeta^*}$, similarly, $\Ginvbw$ has a significantly larger budget $\Ginv_{\bfeta^*}$; as $\Ginv_{\bfeta^*}$ binds both constraints. Moreover, the expected utility of $\Ginvcost$ and $\Ginvbw$ are larger than $\Ginv_{\bfeta^*}$, as they only bind of of the constraint. However, when compared the the expected utility of the benchmark alone, which is 1.587, all strategies outperform the benchmark significantly. Comparing the left and right panel of \Cref{fig:opt-density}, we observe that for the symmetric case $p = 2$, we see a decrease in the spread of $g^*$. An in-depth comparison of the optimal density for different $p$ is discussed in \Cref{tab:statistics-various-alpha}.

\begin{figure}
    \centering
    \includegraphics[width=0.48\textwidth]{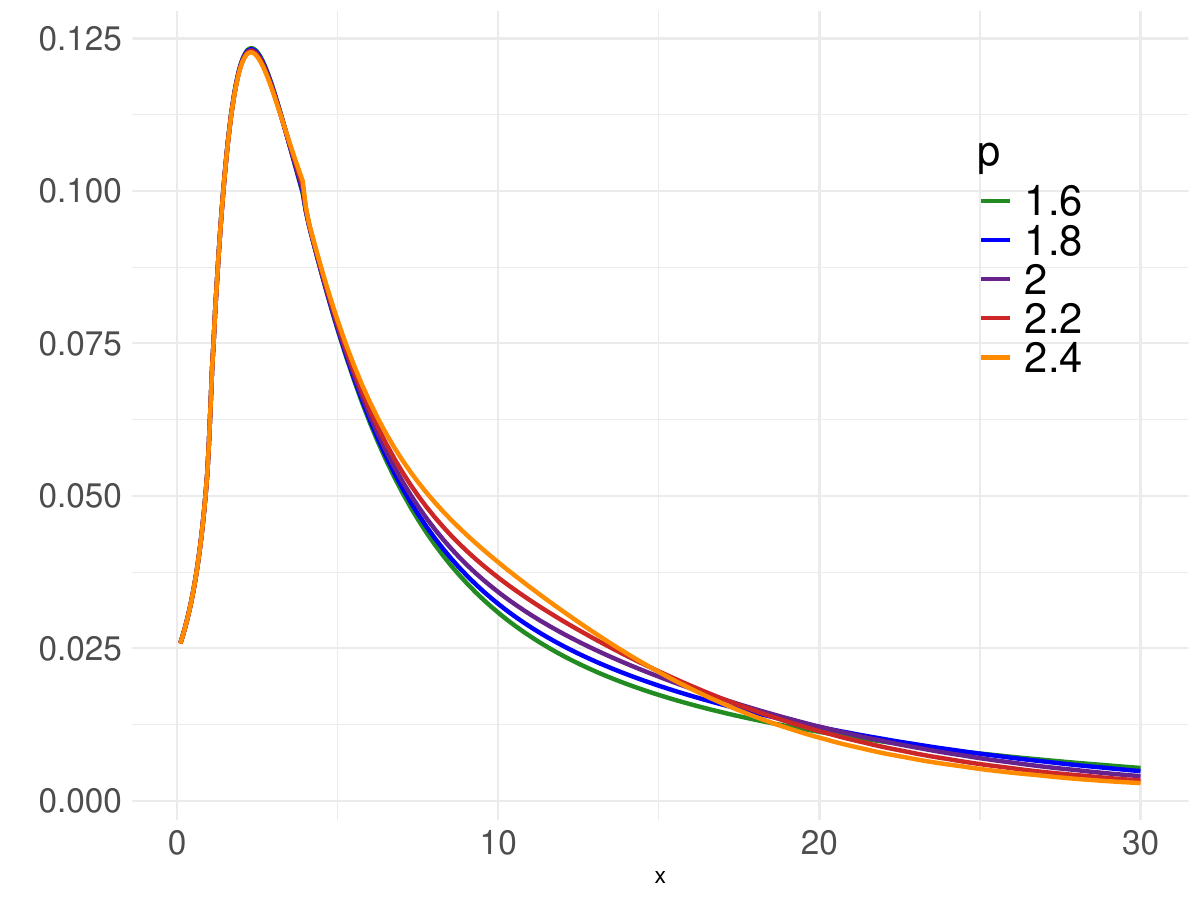}
    \includegraphics[width=0.48\textwidth]{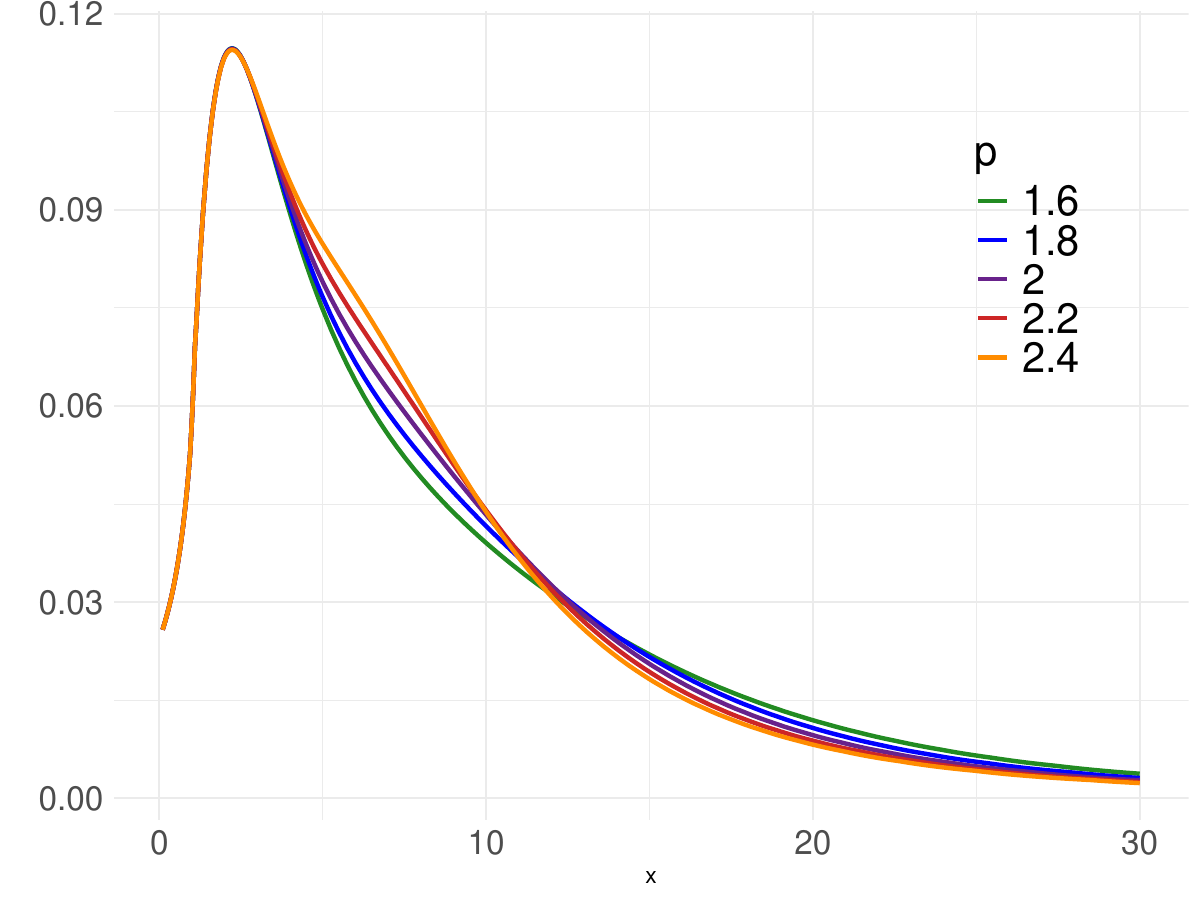}
    
\caption{Optimal density function $g{^*}$ of \eqref{opt:main-prime} when both constraints are binding for different choices of $p \in \{1.6, 1.8, 2, 2.2, 2.4\}$ of the Power Bregman generator. Left $\alpha = 0.1$, right $\alpha = 0.5$.
    Common parameters are $x_0 = 1$, $c = 0.9$, $\gamma = \frac12$, and $\ep = 0.5$.}
    \label{fig:density-BW-p}
\end{figure}

Next, we illustrate the density of the optimal terminal wealth $g^*$ corresponding to $\Ginv_{\bfeta^*}$, when both constraints are binding, for varying parameters. \Cref{fig:density-BW-p} displays $g^*$ for different choices of the Poisson Bregman generator $p \in \{1.6, 1.8, 2, 2.2, 2.4\}$ for $\alpha = 0.1$ (left panel) and $\alpha = 0.5$ (right panel). First, we observe that the left tails are all close to each other and the benchmark, which is due to choice of the Poisson generator. For $\alpha = 0.5$ and $p>2$, mass from the right tail of $g^*$ is shifted towards the middle of the distribution, while for $\alpha = 0.5$ and $p<2$ mass from the centre is shifted towards the right tail (leading to larger gains). This is inline with $p>2$ penalising more deviation from the right tail of the benchmark. For $\alpha = 0.1$ in the left panel of \Cref{fig:density-BW-p}, we observe a similar pattern, that the left tails and part of the centre of the distributions are close to each other. Moreover the shifting of mass happens further to the right. This is inline with the choice of $\alpha = 0.1$, as smaller $\alpha$ penalise underperforming the benchmark more than outperforming it.

In \Cref{fig:density-BW-alpha} we consider $g^*$ (again for both constraints binding) for different choices of $\alpha\in \{0.1, 0.5, 0.9\}$. The left panel for the Poisson Bregman generator $p = 1.6$, the middle panel for the symmetric $p = 2$, and the right panel for $p = 2.4$. The red dashed line indicates the change from under- to outperformance of the benchmark. Specifically, in all panels, to the left of the red dashed line the benchmark outperforms the investor's strategy, i.e. $\Finvbench(\cdot)> \Ginv_{\bfeta^*}(\cdot)$, and to the right of the red dashed line the investor's strategy outperforms the benchmark, i.e., $\Ginv_{\bfeta^*}(\cdot)>\Finvbench(\cdot)$. The red lines are at 3.574 ($p = 1.6$), at 3.108 ($p=  2)$, and at 2.629 ($p = 2.4$). In each the panels, the red lines are numerically indistinguishable between the different $\alpha$ parameters. From \Cref{fig:density-BW-alpha} we observe that for larger $p$ (right panel) the mass is shifted further to the middle of the density compared to $p = 1.6$ (left panel). Moreover, for larger $\alpha$ the shift is more pronounced. 
\begin{figure}[h]
    \centering
    \includegraphics[width=0.32\linewidth]{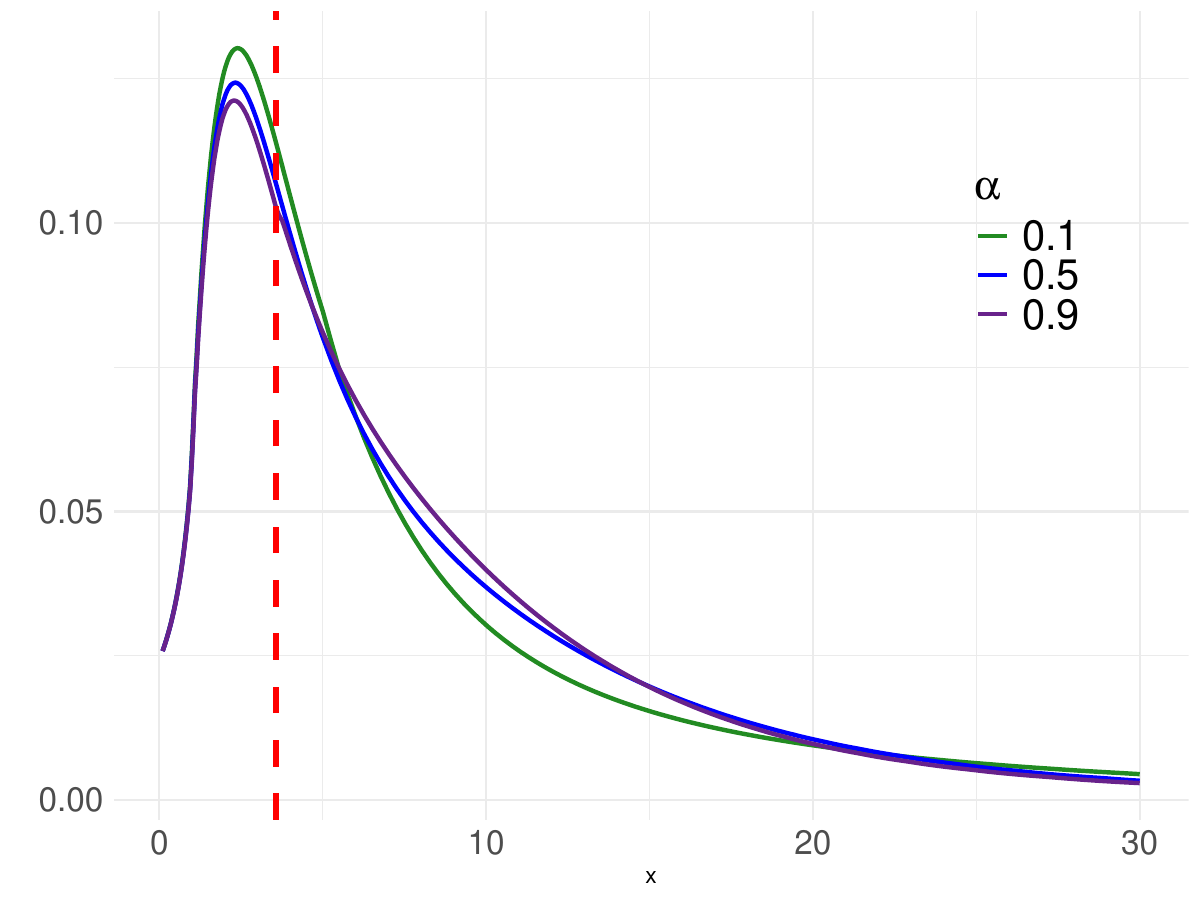}
    \includegraphics[width=0.32\linewidth]{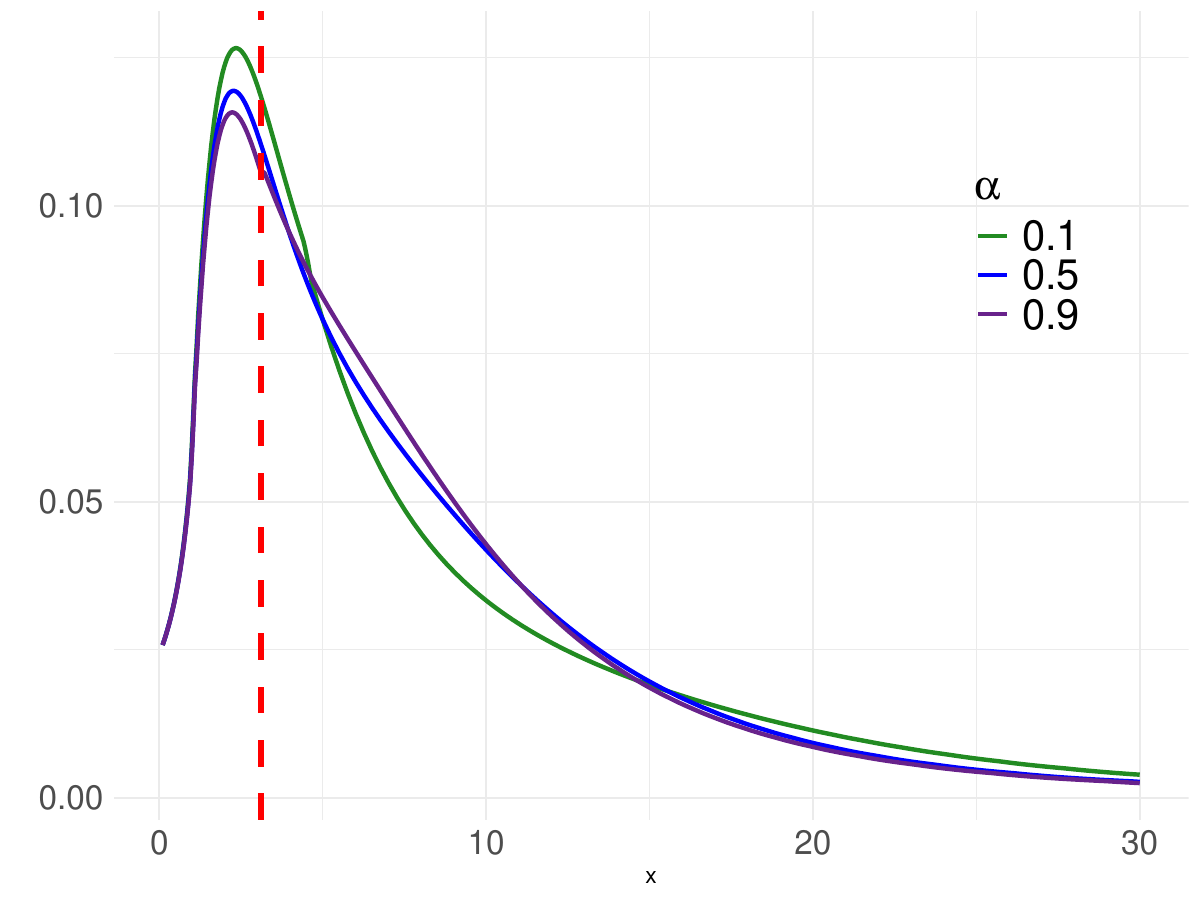}
    \includegraphics[width=0.32\linewidth]{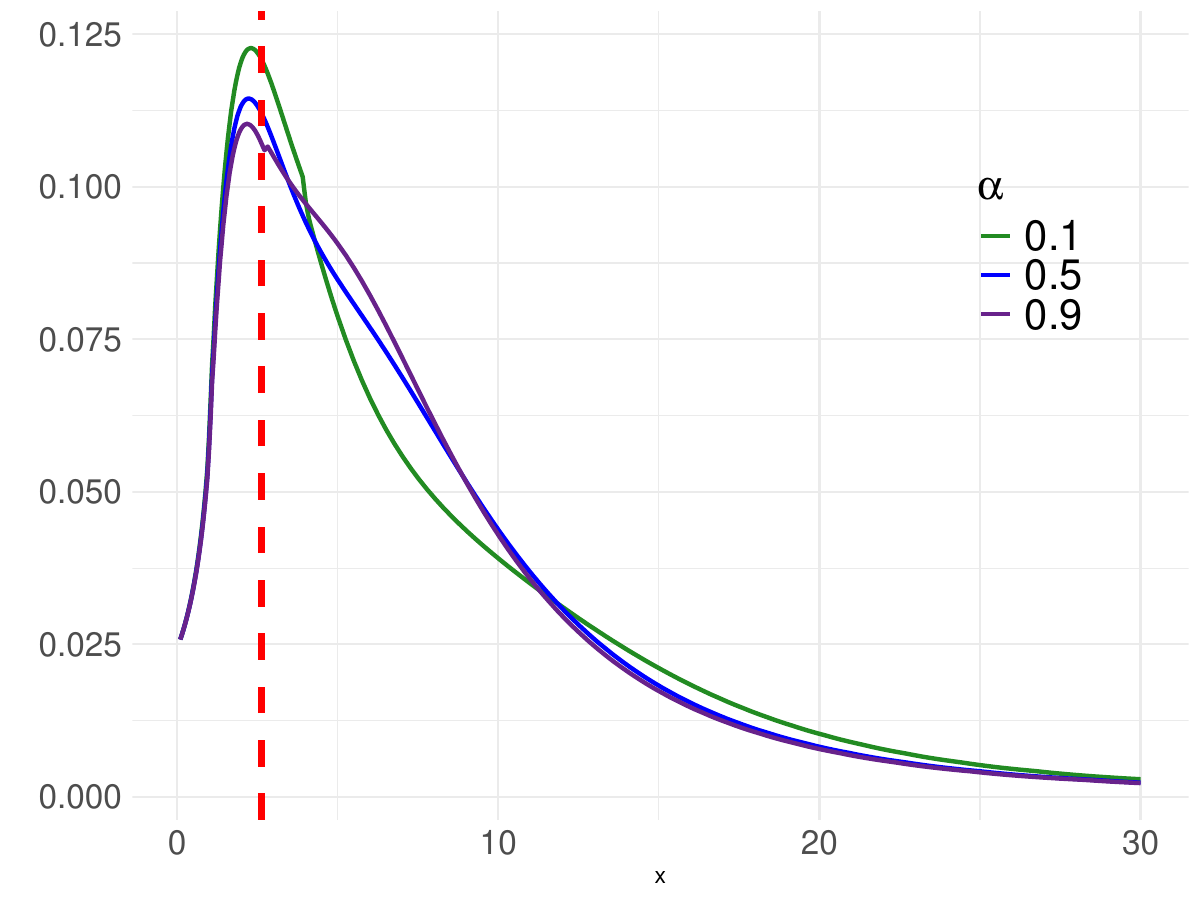}
    \caption{Optimal density function $g{^*}$ of \eqref{opt:main-prime} when both constraints are binding for different choices of $\alpha \in \{0.1, 0.5, 0.9\}$. Left $p = 1.6$, middle $p = 2$, right $p = 2.4$.  To the left of the red dashed vertical line, the strategies underperform the benchmark, i.e., $\Finvbench(\cdot)> \Ginv_{\bfeta*, \alpha}(\cdot)$. To the right of the red dashed vertical line, the strategies outperform the benchmark. The red lines are numerically same for $\Ginv_{\bfeta*, \alpha}$, $\alpha\in \{0.1, 0.5, 0.9\}$ and fixed $p$.
    Common parameters are $c = 0.9$, $x_0 = 1$, $\gamma = \frac12$, and $\ep = 0.5$.
    }
    \label{fig:density-BW-alpha}
\end{figure}

For the parameter choices in \Cref{fig:density-BW-alpha}, we compute summary statistics of the benchmark's and the optimal strategies' terminal wealth. We report the gain-loss ratio (GLR) using the expected wealth of the benchmark strategy as the reference point \citep{deLanghe2015MS}. Specifically, the GLR for a strategy with terminal wealth $X$ and initial cost $X_0$ is defined as 
\begin{equation*}
    \GLR:= \frac{\E\left[\left(\frac{X}{X_0} - \E\left[\frac{Y}{y_0}\right]\right)_+\right]}{\E\left[\left(\E\left[\frac{Y}{y_0}\right] - \frac{X}{X_0}\right)_+\right]}\,.
\end{equation*}
We further report the mean, standard deviation, Value-at-Risk (VaR) at level 0.05, Expected Shortfall (ES) at level 0.05, and the upper tail expectation (UTE) at level 0.9. The VaR at level $\beta \in (0,1)$ of a random variable $Z$ with quantile function $\Finv_Z$ is given by $\VaR_{\beta}(Z):= -\Finv_Z(\beta)$. The ES of a random variable $Z$ at level $\beta \in (0,1)$ is 
\begin{equation*}
    \ES_{\beta}(Z):= -\frac{1}{\beta}\int_0^\beta \Finv_Z(u) \, \d u\,
\end{equation*}
and the UTE of $Z$ at level $\beta \in (0,1)$ is 
\begin{equation*}
    \UTE_\beta(Z):= \frac{1}{1-\beta} \int_{\beta}^1\Finv_Z(u) \, \d u\,.
\end{equation*}

\Cref{tab:statistics-various-alpha} reports these statistics for the benchmark's terminal wealth, as well as for the varying $p = \{1.6, 2.2, 2.4\}$ and $\alpha  = \{0.1, 0.5, 0.9\}$.
\begin{table}[t!]
    \centering
    \begin{tabular}{cccccccc}
   $p$ & $\alpha$ & GLR & mean & std. dev & $\VaR_{0.05}$ & $\ES_{0.05}$ & $\UTE_{0.9}$
    \\
    \toprule \toprule
    \multirow{3}{*}{$p = 1.6$}   
    & 0.1& 1.864 & 9.401 &  11.113 & -1.305&  -0.968&  36.242\\
    & 0.5& 1.475 & 8.451 &   8.384&  -1.311&  -0.971 &  27.673\\
    & 0.9& 1.379 &8.225  & 7.918 &-1.315 &-0.973 & 26.208\\
    \midrule
    \multirow{3}{*}{$p = 2$}
    & 0.1 & 1.557 &8.656 &  8.756 &-1.308 &-0.970 &28.939\\
    & 0.5 & 1.317 &8.082 &  7.586 &-1.318 &-0.974 &25.171\\
    & 0.9 & 1.254 &7.939 &  7.376 &-1.323 &-0.977 &24.532\\
    \midrule
    \multirow{3}{*}{$p = 2.4$}
    & 0.1 &  1.383& 8.239   &7.793& -1.313& -0.972& 25.782\\
    & 0.5&   1.220& 7.863 &  7.237& -1.325& -0.978& 24.086\\
    & 0.9 &  1.176& 7.768  & 7.139& -1.333& -0.981& 23.806\\
    \midrule
    \multicolumn{2}{c}{$\Finvbench$} & 
     1.000  & 7.389  & 6.996 & -1.439 & -1.071 &  23.280 \\
    \bottomrule\bottomrule
    \end{tabular}
    \caption{Various summary statistics of the benchmark's terminal wealth and the optimal strategy's terminal wealth for varying $p \in \{1.6, 2, 2.4\}$ and $\alpha \in \{0.1, 0.5, 0.9\}$. Common parameters are $c = 0.9$, $x_0 = 1$, $\gamma = \frac12$, and $\ep = 0.5$.
    }
    \label{tab:statistics-various-alpha}
\end{table}
First, we observe that the strategies for all choices of $p$ and $\alpha$, lead to larger mean, standard deviation, and UTE than the benchmark, while the risk measures VaR and ES are comparable to the benchmark. Recall that larger $p$ penalises deviations from the right tail of the benchmark more compared to smaller $p$, which is observed by the mean, GLR, and UTE being larger for $p = 1.6$ compared the $p = 2.4$. Moreover, $\alpha = 0.1$, puts a strong penalty on underperforming the benchmark, and indeed, the pair $(p = 2.4, \alpha = 0.5)$ yields the best performing strategy, particularly in the GLR, mean, and UTE. This is consistent with \Cref{fig:density-BW-alpha}, as for smaller $p$ (left panel) the mass from the centre is shifted to the very right of the distribution.

\begin{figure}[h]
    \centering
    \includegraphics[width=0.48\linewidth]{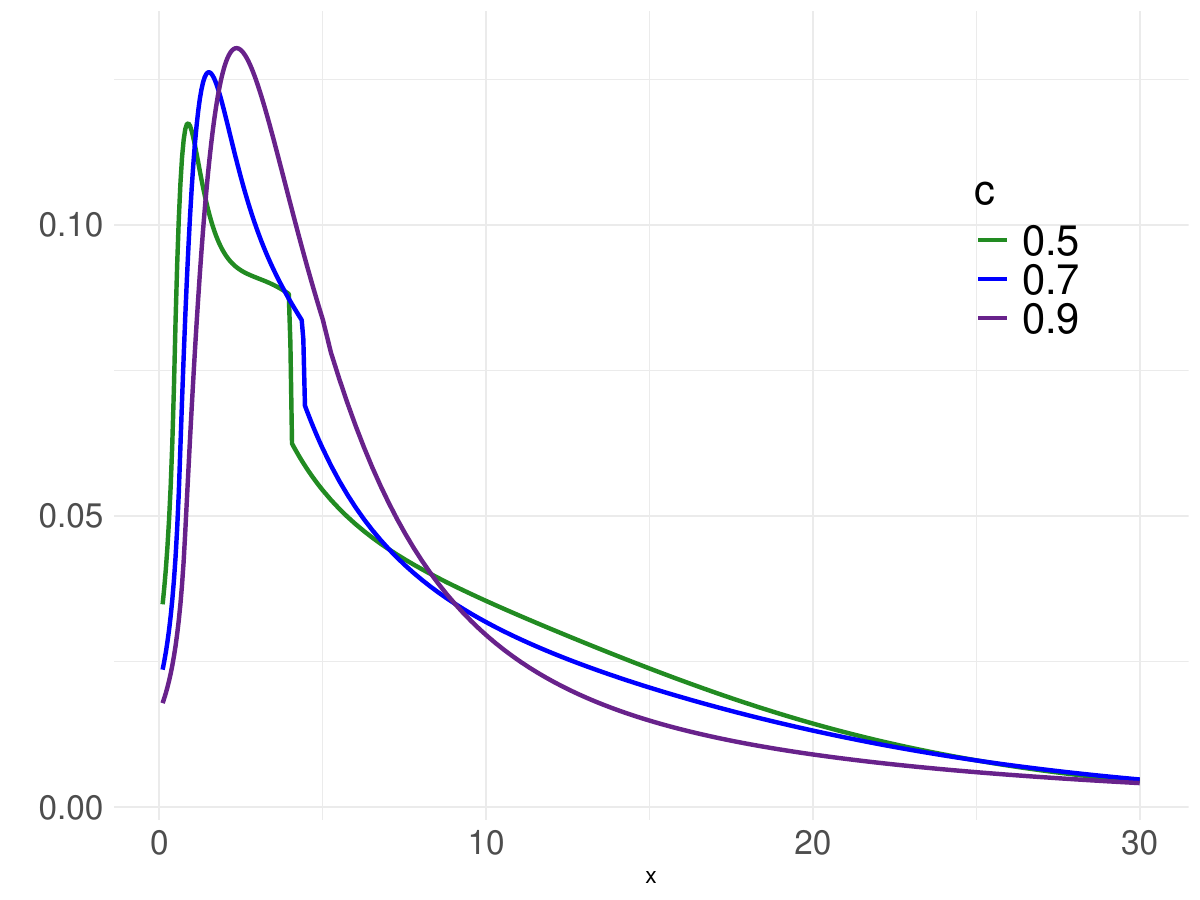}
    \includegraphics[width=0.48\linewidth]{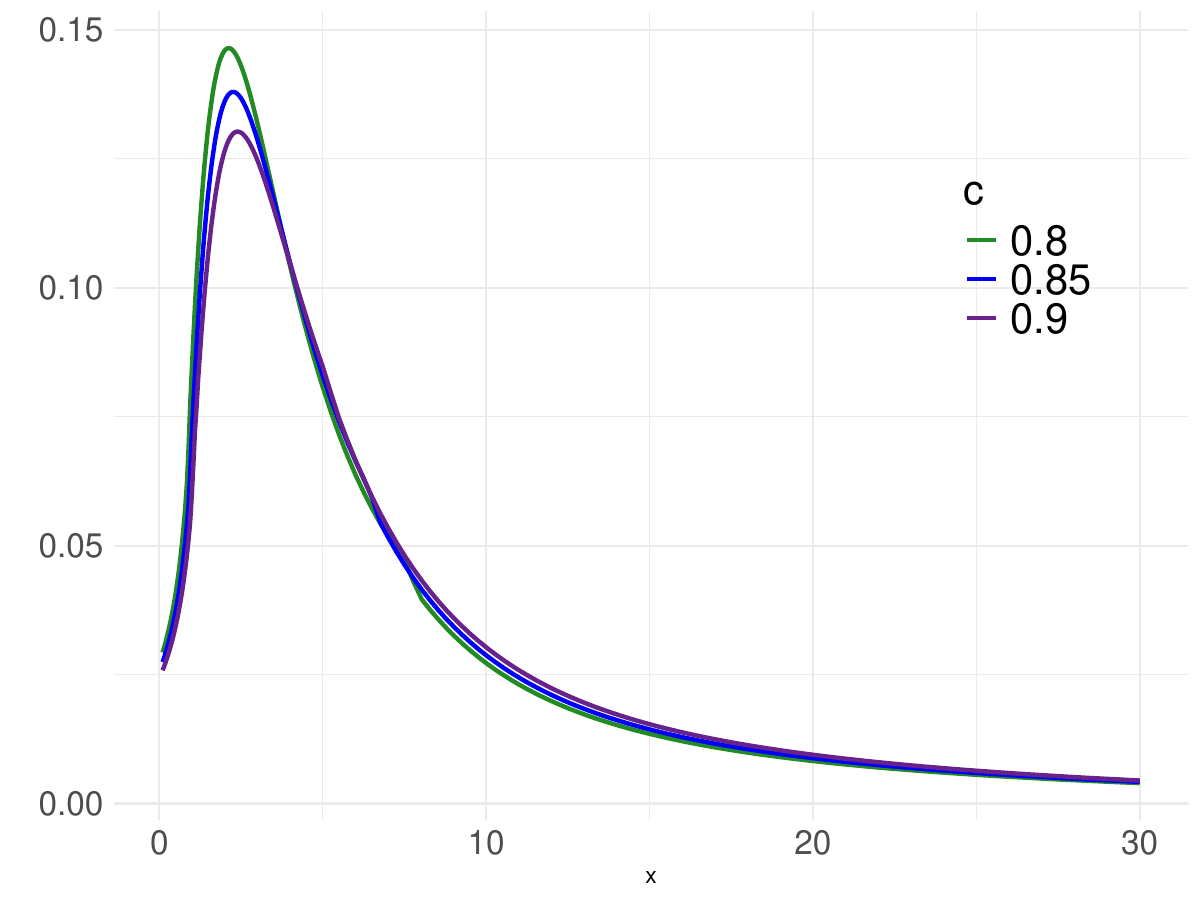}
    \caption{Optimal density function $g{^*}$ of \eqref{opt:main-prime} when both constraints are binding for different choices of $c$. Left $c \in \{0.5, 0.7, 0.9\}$ and $x_0 = 1$. Right $(c, x_0) \in \{(0.8, 0.9), (0.85, 0.95), (0.9, 1)\}$. 
    Common parameters are $\alpha = 0.1$, $p=1.6$, $\gamma = \frac12$, and $\ep = 0.5$.
    }
    \label{fig:density-BW-c}
\end{figure}

Finally, in \Cref{fig:density-BW-c} we consider the optimal density $g^*$ for different choices of $c$. Recall that the investor's criterion is the expected utility of their terminal wealth minus a proportion of the benchmark's, i.e. $\E[U(X - cY)]$. The left panel of \Cref{fig:density-BW-c} depicts the optimal densities $g^*$ for $c \in \{0.5, 0.7, 0.9\}$ and $x_0 = y_0 = 1$. As the investor's budget $x_0$ is equal to the cost of the benchmark $y_0$, the benchmark is attainable in this scenario. Thus the smaller $c$, the less the investor objective depends on the benchmark, and in the limiting case, when $c = 0$, it equals the expected utility of the investor's strategy's terminal wealth, i.e. $\E[U(X)]$. Note however that in all cases in the figure, $g^*$ still attains the $\alpha$-BW divergence constraints, thus is close in distribution to the benchmark's terminal wealth. We observe from the left panel of \Cref{fig:density-BW-c} that the smaller $c$ the more $g^*$ deviates from the benchmark and for small $c$ the density becomes close to bimodal. The right panel of \Cref{fig:density-BW-c} displays $g^*$ for jointly varying $c$ and $x_0$, i.e for $(c, x_0) \in \{(0.8, 0.9), (0.85, 0.95), (0.9, 1)\}$. A joint choice of $c$ and $x_0$ reflects the attainability of the benchmark. For example $(c, x_0) = (0.8, 0.9)$, means that the investor's budget it 0.9 while the cost of the benchmark is 1. Thus, the investor's criterion of $\E[U(X - 0.8\, Y)]$ reflects that the benchmark is not attainable and thus the investor should only compare their strategy to a proportion of that of the benchmark's, displayed in the rigth panel of \Cref{fig:density-BW-c}.

\section{Conclusion}
We solve the portfolio choice problem of an investor who aims at outperforming a benchmark strategy. The investor's criterion is the expected utility of the difference of their strategy's terminal wealth to a proportion of the benchmark's. The investor's strategy needs to satisfy a budget constraint and deviations from the benchmark's terminal wealth are taken into account via the asymmetric $\alpha$-Bregman-Wasserstein divergence. The $\alpha$-Bregman-Wasserstein divergence allows through the Bregman generator to penalise deviations from the benchmark according to absolute wealth, e.g., a larger penalty for deviations from the benchmark that result in losses compared to gains. The $\alpha$ parameter in the divergences, moreover, allows the investor to distinguish the penalty between over- and underperforming the benchmark.  

We solve the portfolio choice problem by deriving its quantile reformulation, and characterise the solutions when only one constraint (either the budget or the $\alpha$-Bregman-Wasserstein divergence) is binding. For the case when both constraints are binding, we derive the optimal quantile function and give explicit conditions when a solution exists. We conclude with an illustration of the optimal quantile functions in a geometric Brownian Motion market model.

\appendix
\setcounter{section}{0}
\justifying

\section{Auxiliary Results}
In this appendix, we understand inequalities and equalities on the level of quantile function pointwise in $u\in(0,1)$, e.g., $\Finv\le\Ginv$ means $\Finv(u) \le  \Ginv(u)$ for all $u\in(0,1)$, that is $\Ginv$ first-order stochastically dominates $\Finv$.

\begin{lemma}\label{lemma:auxiliary-1-main-proof}
   Let $\ell, h\colon \mcM \to \R$ be two strictly convex functions. Assume that $\ell(\Ginv)\ge h (\Ginv)$ for all $\Ginv\in \mcM$ and there exists $\Ginv_0\in \mcM$ such that $\ell(\Ginv_0)=h(\Ginv_0)$. Let $x_\ell^*,x_h^*\in \mcM$ be the unique minimiser of $\ell$ and $h$ respectively. Then, $\Ginv_0(u)\notin \big(\min\{x_\ell^*(u),x_h^*(u)\}, \max\{x_\ell^*(u),x_h^*(u)\}\big)$ for all $u \in (0,1)$. 
\end{lemma}
\begin{proof}
If $x_h^*= x_\ell^*$, then  $\big(\min\{x_\ell^*(u),x_h^*(u)\}, \max\{x_\ell^*(u),x_h^*(u)\}\big) = \emptyset$. 

For the other cases we argue by contradiction. Assume first that $x_h^*< x_\ell^*$, therefore $x_h^*<\Ginv_0<x_\ell^*$. Due to strict convexity, $h$ (resp. $\ell$) is strictly increasing (resp. strictly decreasing) on $(x_h^*,x_\ell^*)$. It follows that $h(x_\ell^*)>h(\Ginv_0)=\ell(\Ginv_0)>\ell(x_\ell^*)$, which contradicts the assumption that $\ell(\Ginv)\ge h (\Ginv)$ for all $\Ginv\in \mcM$.
Next, suppose that  $x_\ell^*<x_h^*$, i.e. $x_\ell^*<\Ginv_0<x_h^*$. Similarly, $\ell$ (resp. $h$) is strictly increasing (resp. decreasing) on $(x_\ell^*,x_h^*)$, which implies that $h(x_\ell^*)>h(\Ginv_0)=\ell(\Ginv_0)>\ell(x_\ell^*)$, which also contradicts the assumption that $\ell(\Ginv)\ge h (\Ginv)$ for all $\Ginv\in \mcM$.
\end{proof}

\begin{lemma}\label{lemma:auxiliary-2-main-proof}
   Let $h, \ell \colon \mcM \times (0,1) \to \R$ be strictly convex in their first components and assume that
      \begin{enumerate}[label = $\roman*)$]
       \item 
       $\ell(\Ginv(u), u) \ge h(\Ginv(u), u)$, for all $\Ginv \in \mcM$ and $u \in (0,1)$.

       \item 
       there exit $\Ginv_1\in \mcM$ which is continuous and satisfies $\ell(\Ginv_1(u),u)= h(\Ginv_1(u), u)$, for all $u \in (0,1)$. 
   \end{enumerate}
Consider the following functions $\bar{h}, \bar{\ell}\colon \mcM \to \R$ given by 
\begin{equation*}
    \bar{h}(\Ginv):= \int_0^1 h(\Ginv(u), u) \, du \quad \text{and} \quad
    \bar{\ell}(\Ginv):= \int_0^1 \ell(\Ginv(u), u) \, du\, ,
\end{equation*}
and denote by $x_h^*, x_\ell^*\in \mcM$ the minimiser of $\bar{h}$ and $\bar{\ell}$, respectively. That is $x_h^* =\argmin_{\Ginv\in\mcM}\bar{h}(\Ginv)$ and $x_\ell^* =\argmin_{\Ginv\in\mcM}\bar{\ell}(\Ginv)$. 

Next, for $\Ginv_0\in \mcM$ that satisfies $\Ginv_0(u) \le \Ginv_1(u)$, for all $u \in (0,1)$, define the function $H \colon \mcM \to \R$ by
   \begin{equation*}
       H(\Ginv)
       :=
       \int_0^1        
         h\big(\Ginv(u), u\big)\Id_{\Ginv(u) 
         \in [\Ginv_0(u), \Ginv_1(u))}
        +
       \ell\big(\Ginv(u), u\big)\Id_{\Ginv(u) \ge \Ginv_1(u)} 
    \, du \,.
   \end{equation*}
If $x_h^*(u) \ge \Ginv_0(u)$ for all $u \in (0,1)$, then the unique minimiser of $H$ is
\begin{equation*}
    \argmin_{\Ginv \in \mcM} H(\Ginv) =
    \begin{cases}
        x^*_\ell & \text{if} \qquad \Ginv_1 \le \min\{x_\ell^*, x_h^*\}\,,
        \\
        x^*_h & \text{if} \qquad \Ginv_1 > \max\{x_\ell^*, x_h^*\}\,,
    \end{cases}
\end{equation*}
which defines a quantile function.
\end{lemma}
\begin{proof}
By the assumptions on $\ell$ and $h$, we have that $\bar{\ell}$ and $\bar{h}$ are strictly convex and satisfy $\bar{\ell}(\Ginv) \ge \bar{h}(\Ginv) $ for all $\Ginv\in\mcM$. Moreover, it holds that $\bar{\ell}(\Ginv_1) =\bar{h}(\Ginv_1)$. Then, applying \Cref{lemma:auxiliary-1-main-proof} to $\bar{\ell}$ and $\bar{h}$, we obtain
\begin{equation}\label{eq:interval-rep}
\begin{split}
    \Ginv_1 &\not\in \big(\min\{x_\ell^*, x_h^*\}, \, \max\{x_\ell^*, x_h^*\}\big)\,,
\end{split}
\end{equation}
where the inclusions above are to be understood pointwise in $u \in (0,1)$. Now, there are two possible cases: 

\underline{Case 1}: Let $\Ginv_1 \le \min\{x_\ell^*, x_h^*\}$. As $x^*_\ell \ge \Ginv_1$, we have
\begin{equation*}
    \inf_{\Ginv \in \mcM} \bar{\ell}(\Ginv) 
    =
    \bar{\ell}(x_\ell^*) 
    =
    \inf_{\Ginv \in \mcM} \int_0^1 \ell(\Ginv(u), u)\Id_{\Ginv(u) \ge \Ginv_1(u)} \,du \,.
\end{equation*}
Further, as $x_h^*\ge \Ginv_1$, we have
\begin{equation*}
    \inf_{\Ginv \in \mcM} \int_0^1 h(\Ginv(u), u)\Id_{\Ginv(u) \in[ \Ginv_0(u), \Ginv_1(u))}\,du
    >
    \bar{h}(\Ginv_1)
    =
    \bar{\ell}(\Ginv_1) 
    \ge 
    \bar{\ell}(x_\ell^*)\,,
\end{equation*}
where the equality follows by definition of $\Ginv_1$ and the last inequality as $x_\ell^*$ is the minimiser of $\ell$.

Combining the above arguments and that $x_\ell^*(u) \ge \Ginv_1(u) \ge \Ginv_0(u)$, $u \in (0,1)$, the infimum of $H$ is attained at $x^*_\ell$.

\underline{Case 2}: Let $\Ginv_1 > \max\{x_\ell^*, x_h^*\}$. Then, as $x^*_h < \Ginv_1$ and by assumption that $x_h^* \ge \Ginv_0$, we obtain
\begin{equation*}
    \inf_{\Ginv \in \mcM} \int_0^1 h(\Ginv(u), u)\Id_{\Ginv(u) \in[ \Ginv_0(u), \Ginv_1(u))} \,du 
    =
    \bar{h}(x_h^*)\, .
\end{equation*}

Next, as $x_\ell^* < \Ginv_1$, we obtain that 
\begin{equation*}
    \inf_{\Ginv \in \mcM} \int_0^1 \ell(\Ginv(u), u)\Id_{\Ginv(u) \ge \Ginv_1(u)}\,du
    \ge 
    \bar{\ell}(\Ginv_1)
    =
    \bar{h}(\Ginv_1)
    \ge 
    \bar{h}(x_h^*)\,,
\end{equation*}
where the equality follows by definition of $\Ginv_1$ and the last inequality as $x_h^*$ is the minimiser of $h$. Thus, the infimum of $H$ is attained at $x_h^*$, which satisfies $x_h^* \ge \Ginv_0$.

Combining both cases gives the representation of the argmin of $H$ as
\begin{equation*}
    \argmin_{\Ginv \in \mcM} H(\Ginv) =
    \begin{cases}
        x^*_\ell & \text{if} \quad  \Ginv_1 \le \min\{x_\ell^*, x_h^*\}\,,
        \\
        x_h^* & \text{if} \quad \Ginv_1 > \max\{x_\ell^*, x_h^*\}\,,
    \end{cases}
\end{equation*}
Thus, we obtain the representation of the argmin of $H$ in the statement. 

Finally, we show that the minimiser of $H$ is a quantile function, i.e.,  it is non-decreasing. Since $\Ginv_1$ is continuous and satisfies \eqref{eq:interval-rep}, there must exist a $ \uu \in (0,1)$ such that $\Ginv_1(\uu) = \min\{x^*_\ell(\uu), x^*_h(\uu)\} = \max\{x^*_\ell(\uu), x^*_h(\uu)\}$, and in particular $x^*_\ell(\uu)= x^*_h(\uu)$. Thus, the minimiser of $H $ is the function
\begin{equation}\label{eq:argmin-H}
   x_\ell^*(u)\,  \Id_{u \le \uu} 
    + 
    x_h^*(u) \, \Id_{u > \uu}
\,, \quad u \in (0,1).
\end{equation}
Now as $x_\ell^*, x_h^*\in \mcM$, and in particular are non-decreasing, and we have that  $x_\ell^*(\uu) = x_h^*(\uu)$, the function in \eqref{eq:argmin-H} is indeed non-decreasing, and thus a quantile function. 
\end{proof}

The next results follows along the lines of the proof of \Cref{lemma:auxiliary-2-main-proof} and is omitted.

\begin{corollary}\label{lemma:auxiliary-2-main-proof-alpha-small}
    Let all assumptions of \Cref{lemma:auxiliary-2-main-proof} be satisfied but define 
    \begin{equation*}
       H(\Ginv)
       :=
       \int_0^1        
         \ell\big(\Ginv(u), u\big)\Id_{\Ginv(u) 
         \in [\Ginv_0(u), \Ginv_1(u))}
        +
       h\big(\Ginv(u), u\big)\Id_{\Ginv(u) \ge \Ginv_1(u)} 
    \, du \,.
   \end{equation*}
If $x_\ell^*(u) \ge \Ginv_0(u)$ for all $u \in (0,1)$, then the unique minimiser of $H$ is
\begin{equation*}
    \argmin_{\Ginv \in \mcM} H(\Ginv) =
    \begin{cases}
        x^*_h & \text{if} \qquad \Ginv_1 \le \min\{x_\ell^*, x_h^*\}\,,
        \\
        x^*_\ell & \text{if} \qquad \Ginv_1 > \max\{x_\ell^*, x_h^*\}\,,
    \end{cases}
\end{equation*}
which defines a quantile function.
\end{corollary}

\begin{lemma}\label{lemma:Ginv-monotonicity-beta}
    For $\Ginv_{\bfeta}^{(\beta)}(u)$, $\beta\in(0,1)$, defined in \eqref{eq:Ginv-beta}. Then it holds that 
    \begin{enumerate}[label = $\roman*)$]
        \item $\Ginv_\bfeta^{(\beta)}(u)$ is strictly decreasing in $\beta$ on the set $\{u \in (0,1)~|~ \Finvbench(u) < \Ginv_\bfeta^{(\beta)}(u)\}$. 
        
        \item $\Ginv_\bfeta^{(\beta)}(u)$ is strictly increasing in $\beta$ on the set $\{u \in (0,1)~|~ \Finvbench(u) > \Ginv_\bfeta^{(\beta)}(u)\}$. 
        
    \end{enumerate}
\end{lemma}

\begin{proof}
    Recall that $\Ginv_\bfeta^{(\beta)}$, $\beta \in (0,1)$, satisfies
    \begin{equation}\label{eq:Ginv-beta-indirect}
        \tilde{H}_{\beta}\big(\Ginv^{(\beta)}_{\bfeta}(u)\big)
 =
 - U'\big(\Ginv^{(\beta)}_{\bfeta}(u) - c\Finvbench(u)\big) + \eta_2\, \beta\, \phi'\big(\Ginv^{(\beta)}_{\bfeta}(u)\big)
    =
    \eta_2 \, \beta\, \phi'\big(\Finvbench(u)\big) - \eta_1 \xi(u)\,.
    \end{equation}
Next, note that $\Ginv_\bfeta^{(\beta)}$ is differentiable in $\beta$ and taking derivatives with respect to $\beta$ in \eqref{eq:Ginv-beta-indirect} yields for all $u \in (0,1)$
\begin{equation*}
    \tfrac{\partial}{\partial \beta} \Ginv_\bfeta^{(\beta)}(u)\; 
    \left\{ 
    - U''\big(\Ginv^{(\beta)}_{\bfeta}(u) - c\Finvbench(u)\big) + \eta_2\, \beta\, \phi''\big(\Ginv^{(\beta)}_{\bfeta}(u)\big)
    \right\}
    =
    \eta_2 \, \left(\, \phi'\big(\Finvbench(u)\big) 
    -\phi'\big(\Ginv^{(\beta)}_{\bfeta}(u)\big)\right)\,.
\end{equation*}
Due to strict concavity of $U$ and strict convexity of $\phi$, the term in the curly brackets is always strictly positive. Thus the sign of $\frac{\partial}{\partial \beta} \Ginv_\bfeta^{(\beta)}(u)$ equals the sign of $\left(\, \phi'\big(\Finvbench(u)\big) 
    -\phi'\big(\Ginv^{(\beta)}_{\bfeta}(u)\big)\right)$. By strict convexity of $\phi$, we observe that $\tfrac{\partial}{\partial \beta} \Ginv_\bfeta^{(\beta)}(u)<0$, whenever $\Finvbench(u) < \Ginv_\bfeta^{(\beta)}(u)$. And that $\tfrac{\partial}{\partial \beta} \Ginv_\bfeta^{(\beta)}(u)>0$, whenever $\Finvbench(u) > \Ginv_\bfeta^{(\beta)}(u)$.     
\end{proof}

\begin{lemma}  \label{Le:monoblock} 
Let \(\mfU\colon \R \to \R\) with $\mfU(x) = -\infty$ for $x< 0$ and $\mfU(x) \le 0$ for $x >0$, and let $\mfU$ when restricted to $(0, +\infty)$, i.e. $\mfU \colon (0,+\infty) \to (-\infty, 0]$, be strictly increasing. Further define $\varsigma \colon (0,1)\to (-\infty, 0]$ and $\varphi\colon \R\to \R$ and assume that $\varsigma$ and $\varphi$ are strictly increasing.

Moreover, let \(f\in \mcM\) be square integrable. Then for each \(\eta>0\) and $x \in (0,1)$, 
denote by \(g(\eta,x)\) the solution in $g$ of
\begin{equation}\label{lemma:eq:indirect}
      \mfU\bigl(g-c f(x)\bigr)
      +\eta\,\Big(\varphi\big(g\big) - \varphi\bigl(f(x)\bigr)\Big)
      = \varsigma(x)\,.    
\end{equation}
   Then, \(g(\eta,x) > c \, f(x)\) for all $\eta>0$ and \(x\in(0,1)\). Moreover
\begin{enumerate}[label = $(\roman*)$]
  \item The mapping \(\eta\mapsto g(\eta,x)\) is strictly increasing on \(\{x \in (0,1)~|~  f(x) >g(\eta,x)>c\, g(x)\bigr\}\).
  \item The mapping \(\eta\mapsto g(\eta,x)\) is strictly decreasing on \(\{x\in (0,1)~|~ g(\eta,x)> f(x) \bigr\}\).
\end{enumerate}
\end{lemma}

\begin{proof}
First, we show that for all $\eta>0$ and $x \in (0,1)$ a solution to \eqref{lemma:eq:indirect} exists. For this fix $\eta>0$ and $x \in(0,1)$ and note that $\mfU$ is only finite if $g(\eta, x) > c\,  f(x)$. Moreover, as $\mfU$ takes on all values in $(0,+\infty)$, and as $\varsigma(x)<0 $ for all $x \in (0,1)$, there always exists a solution to \eqref{lemma:eq:indirect}. 

Next, for $(i)$ we define the functional $D\colon\mcM \times [0, +\infty)\to \mcM$ by 
\begin{equation*}
    D\big[\ell, \eta\big]:=  \mfU\bigl(\ell-c f\bigr)
      +\eta\,\Big(\varphi(\ell) - \varphi\bigl(f\bigr)\Big)
      - \varsigma\,,
      \quad \text{for} \quad \ell \in \mcM\,,\; \eta >0\,.
\end{equation*}
Moreover, we have that $D$ is strictly increasing in $\ell$, i.e. for two functions $\ell_1, \ell_2 \in \mcM$ satisfying $\ell_1(x)\le \ell_2(x)$ for all $x\in(0,1)$, it holds that $D\big[\ell_1, \eta\big](x)\le D\big[\ell_2, \eta\big](x)$ for all $x \in (0,1)$.

Next, let $0<\ueta< \oeta$ and define the functions
\begin{equation*}
    \ug(\cdot):= g(\ueta, \cdot)
    \quad \text{and}\quad
    \og(\cdot):= g(\oeta, \cdot)\,,
\end{equation*}
that is $\ug$ (resp. $\og$) is the solution to \eqref{lemma:eq:indirect} with $\ueta$ (resp. $\oeta$). Then it holds that $D[\ug, \ueta] = D [\og, \oeta] = 0$ and moreover
\begin{align*}
   D[\ug, \oeta] 
    &= D[\ug, \oeta] - D[\ug, \ueta]
    \\
    &=   \mfU\bigl(\ug-c f\bigr)
      +\oeta\,\Big(\varphi(\ug) - \varphi\bigl(f\bigr)\Big)
      - \varsigma 
       -\mfU\bigl(\ug-c f\bigr)
      -\ueta\,\Big(\varphi(\ug) - \varphi\bigl(f\bigr)\Big)
      + \varsigma 
      \\
      &= 
     \big( \oeta -\ueta\big)\,\Big(\varphi(\ug) - \varphi\bigl(f\bigr)\Big)\,.
\end{align*}
Thus, for $x\in\{y \in(0,1)~|~ f(y) <\ug(y)\} $ it holds by strict increasingness of $\varphi$ that $D[\ug, \oeta](x)>0$. As furthermore, $D[\og, \oeta](x)=0$ on that set and $D$ is strictly increasing in its first component, we obtain that $\eta \mapsto g(x, \eta)$ is strictly decreasing for all $x\in\{y \in(0,1)~|~ f(y) <\ug(y)\} $.

Using similar arguments, we obtain that on the set $x \in \{y \in(0,1)~|~ f(y) >\ug(y)\}$ we have that $D[\ug, \oeta]<0$, which then implies that $\eta \mapsto g(x, \eta)$ is strictly increasing for all $x\in\{y \in(0,1)~|~ f(y) \ge \ug(y) > c\, f(y)\} $, where the inequality $\ug(y) = g(\ueta, y) > c\, f(y)$ follows by definition of $g$.
\end{proof}

\begin{lemma}\label{Le:monoG}
Let $\xi$ be non-increasing (that is \Cref{asm:xi-non-increasing} is satisfied) and $\Ginv_{\bfeta}(u)$ defined in \eqref{eq:Ginv-eta}. Then for all $u \in (0,1)$, and $\bfeta >0$, it holds that 
\begin{itemize}
    \item [(i)] for $\eta_2>0$ fixed, $\Ginv_{\bfeta}(u)$ is strictly decreasing in $\eta_1$. 
\item [(ii)] for $\eta_1> 0$ fixed, $\Ginv_{\bfeta}(u)$ is strictly decreasing {in $\eta_2$} on $\{u\in[0,1]~|~ \Ginv_{\bfeta}(u)\ge \Finvbench(u)\}$ and strictly increasing on $\{u\in[0,1]~|~ c\Finvbench(u)< \Ginv_{\bfeta}(u)<\Finvbench(u)\}$. 
\end{itemize}
\end{lemma}

\begin{proof} 
First we note that $\Ginv_\bfeta$ is a combination of $\Ginv_\bfeta^{(\alpha)}$ and $\Ginv_\bfeta^{(1-\alpha)}$. To show the two cases, let $\beta \in (0,1)$ and $\Ginv^{{(\beta)}}_{\bfeta}(u)$ defined as in \eqref{eq:Ginv-beta}. 

For case $(i)$, note that $\Ginv^{{(\beta)}}_{\bfeta}(u)$ is strictly decreasing in $\eta_1$, which implies that $\Ginv_\bfeta$ is strictly decreasing in $\eta_1$.

To see case $(ii)$, let $\eta_1>0$ be fixed and recall that  $\Ginv^{{(\beta)}}_{\bfeta}$ satisfies 
\begin{equation*}
      \tilde{H}_{\beta}\big(\Ginv^{(\beta)}_{\bfeta}(u)\big)
 =
 - U'\big(\Ginv^{(\beta)}_{\bfeta}(u) - c\Finvbench(u)\big) + \eta_2\, \beta\, \phi'\big(\Ginv^{(\beta)}_{\bfeta}(u)\big)
    =
    \eta_2 \, \beta\, \phi'\big(\Finvbench(u)\big) - \eta_1 \xi(u)
 \end{equation*}
 which is equivalent to 
 \begin{equation*}
 - U'\big(\Ginv^{(\beta)}_{\bfeta}(u) - c\Finvbench(u)\big) + \eta_2\, \beta\, \Big(\phi'\big(\Ginv^{(\beta)}_{\bfeta}(u)\big) - \phi'\big(\Finvbench(u)\big)\Big)
    =
     - \eta_1 \xi(u)\,.
 \end{equation*}
Next, we apply \Cref{Le:monoblock} with $\mfU = -U'$, $f =\Finvbench$, $\varphi = \beta\phi'$, and $\varsigma = -\eta_1 \xi$. Thus we obtain that 
for $\eta_1> 0$ fixed, $\Ginv^{{(\beta)}}_{\bfeta}(u)$ is strictly decreasing {in $\eta_2$} on $\{u\in[0,1]~|~ \Ginv^{{(\beta)}}_{\bfeta}(u)\ge \Finvbench(u)\}$ and strictly increasing on $\{u\in[0,1]~|~ c\Finvbench(u)< \Ginv^{{(\beta)}}_{\bfeta}(u)<\Finvbench(u)\}$. Thus, by definition of $\Ginv_\bfeta$, we obtain that also $\Ginv_{\bfeta}(u)$ is strictly decreasing {in $\eta_2$} on $\{u\in[0,1]~|~ \Ginv_{\bfeta}(u)\ge \Finvbench(u)\}$ and strictly increasing on $\{u\in[0,1]~|~ c\Finvbench(u)< \Ginv_{\bfeta}(u)<\Finvbench(u)\}$. This concludes the proof.
\end{proof}

\section{Isotonic projection}\label{app:iso}

\begin{theorem}[Isotonic projection -  \cite{Barlow1972JASA}]\label{thm:generalised-iso}
Let $T\colon \R \to \R$ be a proper convex function with derivative $T^\prime(x) = \frac{d}{dx}T(x)$ and let $\ell \colon (0,1) \to \R$ be square integrable. Then the optimisation problem 
\begin{equation*}
    \argmin_{g \in \mcM} \int_\R  T\big(g(u)\big) - \ell(u)\, g(u) \,  du\,
\end{equation*}
has solution 
\begin{equation*}
    \ell^*(u) := \breve{T^{\prime}}\big(\, \ell^{\uparrow}(u)\, \big)\,,
\end{equation*}
where $\breve{T}^\prime$ denotes the generalised inverse of $T^\prime$. The solution is unique if $T$ is strictly convex. 
\end{theorem}

Next we collect some properties of the isotonic projection.
\begin{proposition}\label{prop:isotonic-properties}
    Let $\ell \in \mcH$. Then it holds that
    \begin{enumerate}[label = $\roman*)$]
        \item $-\, \ell^\downarrow = \big[- \ell\big]^\uparrow$.

    \item Let $\ell_1, \ell_2 \in \mfH$ satisfying $\ell_1(u) \le \ell_2(u)$ for all $u \in [0,1]$. Then $\ell_1^\uparrow(u) \le \ell_2^\uparrow(u)$ for all $u \in [0,1]$.
    
      \item Let $g_1, g_2 \in \mcH$ and define $\ell_\lambda(u):= g_1(u) + \lambda g_2(u)$, $u \in [0,1]$. Then for $\lambda_n \to \lambda $ as $n\to +\infty$, it holds that $\displaystyle\lim_{n \to +\infty} \textstyle \int_0^1 \big(\ell_{\lambda_n}(u) - \ell_\lambda(u)\big)^2\, du = 0$.
    
    \end{enumerate}
\end{proposition}

\begin{proof}
Part $i)$: 
    By definition, and using the change of variables $m = - h$ in the second equation,
    \begin{align*}
        \ell^\downarrow
        &=
        \argmin_{-h \in \mcM} \int_0^1 \big(h(x) - \ell(x)\big)^2 \, du
        \\
        &= 
        - \argmin_{m \in \mcM} \int_0^1 \big(-m(x) - \ell(x)\big)^2 \, du
        \\
        &= 
        - \argmin_{m \in \mcM} \int_0^1 \big(m(x) - [- \ell(x)]\big)^2 \, du
        \\
        &= 
        - [- \ell(x)]^\uparrow\,.
    \end{align*}
    
  Part $ii)$ follows by proposition 5 in \citet{Tam2025WP} and part $iii)$ by proposition A.2. in \citet{Bernard2024MF}. 
\end{proof}

\section{Details to the GMB market model}\label{app:GBM}

Here we provide the derivation of the formulas in \Cref{ex: GBM-MM,ex:gbm-xi}.
From \eqref{eq:MM-SDE} and the GBM market model in \eqref{eqn:dS-GBM}, the benchmark satisfies
\begin{equation*}
    dX_t^\pi = X_t^\pi\left(\left(\big(\mu_t-r_t\,\mathbf{1}\big)^\intercal\pi_t + r_t \right)\,dt + \sum_{i\in\mcD}\pi_t^i\,\sigma_t^i\,dW_t^i\right).
\end{equation*} 
and from Ito's lemma,  
\begin{equation*}
    d\log(X_t^\pi)=\left(\left(\big(\mu_t-r_t\,\mathbf{1}\big)^\intercal\pi_t + r_t\right) -{\textstyle\frac12} \pi_t^\intercal\,\Upsilon_t\,\pi_t\right)\,dt + \pi_t^\intercal\,\sigma\,dW_t.
\end{equation*}
Hence, for deterministic $\pi$, integration and exponentiation yields $X_T^\pi\stackrel{d}{=} X_0^\pi\,\exp\{\Gamma-\frac{1}{2}\Psi^2 + \Psi\,Z\}$, where $Z\stackrel{\P}{\sim} N(0,1)$, and $\Gamma$ and $\Psi$ are as in \eqref{eqn:Gamma-and-Psi}. \eqref{eqn:F-for-deterministic-bench} follows immediately from this result and \eqref{eqn:Finv-for-deterministic-bench} follows from elementary calculations. 

For the computation of $\xi(u)$, first note $f_Y(x)=\frac{d}{dx}\Fbench(x) = 
\phi\left((\log(x/X_0^\pi)-\Gamma+\tfrac12 \Psi^2)/\Psi\right)/(x\,\Psi)$.
Next, from \eqref{eqn:sdf-gbm}  we have
\begin{align*}
 \hat{f}_Y(x) 
 &:= \frac{d}{dx}\E[\varsigma_T\,\Id_{\{X_T^\pi<x\}}]
 \\
 &= e^{-\int_0^Tr_s\,ds}\;\frac{d}{dx}\E\left[e^{-\frac12\int_0^T\lambda^\intercal_s\rho\lambda_s\,ds-\int_0^T\lambda_s^\intercal\,dW_s}\,\Id_{\{X_T^\pi<x\}}\right]
 \\
 &= e^{-\int_0^Tr_s\,ds}\;\frac{d}{dx}\Q\left(X_T^\pi<x\right)
\end{align*}
where $\Q$ is the risk-neutral probability measure defined by $\Q(\omega)=e^{-\frac12\int_0^T\lambda^\intercal_s\rho\lambda_s\,ds-\int_0^T\lambda_s^\intercal\,dW_s(\omega)}\,\P(\omega)$ for all $\omega\in\Omega$. From Girsanov's Theorem, $\hW_t:=\int_0^t\lambda_s\,ds+W_t$ is a $\Q$-Brownian motion. Thus, we have
\begin{equation*}
    dX_t^\pi = X_t^\pi\,\left(r_t\,dt+\sum_{i\in\mcD} \pi_t^i\,\sigma^i_t\,d\hW_t^i\right)\,.
\end{equation*}
By Ito's lemma, integration, and exponentiation, we have $X_T^\pi\stackrel{d}{=} X_0^\pi\,\exp\{R-\frac{1}{2}\Psi^2 + \Psi\,Z_\Q\}$, 
 where $Z_\Q\stackrel{\Q}{\sim}N(0,1)$ and $R:=\int_0^ Tr_s \,ds$. Hence, $ \Q(X_T^\pi<x)= \Phi\left(\frac{\log(x/X_0^\pi)-R+\frac{1}{2}\Psi^2}{\Psi}\right)$ and moreover $\hat{f}_Y(x) =e^{-R}\, \phi\left((\log(x/X_0^\pi)-R+\frac{1}{2}\Psi^2)/\Psi\right)/(x\,\Psi)$.

Putting these two results together, we have
\begin{align*}
    \xi(u) = \frac{\hat{f}_Y(\Finvbench(u))}{f_Y(\Finvbench(u))} = 
    e^{-R}\frac{\phi\left(\left(\log(\Finvbench(u)/X_0^\pi)-R+\frac{1}{2}\Psi^2\right)/\Psi\right)}
    {\phi\left(\left(\log(\Finvbench(u)/X_0^\pi)-\Gamma+\tfrac12 \Psi^2\right)/\Psi\right)},
\end{align*}
and after inserting \eqref{eqn:Finv-for-deterministic-bench} and simplifying we find \eqref{eqn:R-N-deterministic-delta}.

\setlength{\bibsep}{6pt}
\titlespacing*{\section}{0pt}{6pt}{12pt}
\singlespacing
\begin{spacing}{0.0}
\bibliographystyle{apalike}
\bibliography{references}
\end{spacing}

\end{document}